\documentclass[aps,prx,superscriptaddress,nofootinbib,amsfonts,amsmath,amssymb,noeprint,reprint,longbibliography]{revtex4-2}
\usepackage{tikz}
\usetikzlibrary{decorations.pathreplacing,calligraphy,decorations.markings}
\usetikzlibrary {arrows.meta}

\definecolor{tensorcolor}{rgb}{0.65,0.77,0.95}
\definecolor{btensorcolor}{rgb}{0.65,0.50,0.69}
\definecolor{whitetensorcolor}{rgb}{0.93,0.93,0.93}
\definecolor{gtensorcolor}{rgb}{0.6,0.8,0.5}
\definecolor{operatorcolor}{rgb}{1.0,1.0,1.0}

\newcommand\singledx{1.8}

\newcommand{\GTensor}[5]{
	\begin{scope}[shift={(#1)}]
    \ifnum#5=0
		\draw[very thick, draw=red] (-#2,0) -- (#2,0);
            \draw[very thick] (0,#2) -- (0,0);
    \fi
    \ifnum#5=-1
		\draw[very thick] (0,0) -- (#2,0);
            \draw[very thick] (0,#2) -- (0,0);
    \fi
    \ifnum#5=1
		\draw[very thick] (-#2,0) -- (0,0);
            \draw[very thick] (0,#2) -- (0,0);
    \fi

    \ifnum#5=2
		\draw[very thick,draw=red] (-#2,0) -- (#2,0);
    \fi
    \ifnum#5=3
		\draw[very thick] (0,-#2) -- (0,#2);
    \fi
    \ifnum#5=4
		\draw[very thick] (-#2,0) -- (#2,0);
    \fi
    \ifnum#5=5
		\draw[very thick, draw=red] (-#2,0) -- (#2,0);
		\draw[very thick] (0,#2) -- (0,-#2);
    \fi
    \ifnum#5=6
		\draw[very thick] (-#2,0) -- (#2,0);
            \draw[very thick] (0,#2) -- (0,0);
    \fi
    \ifnum#5=7
		\draw[very thick, draw=red] (-#2,0) -- (#2,0);
            \draw[very thick] (0,-#2) -- (0,0);
    \fi
    \ifnum#5=8
		\draw[very thick] (-#2,0) -- (#2,0);
    \fi
    \ifnum#5=9
		\draw[very thick] (-#2,0) -- (#2,0);
            \draw[very thick] (0,-#2) -- (0,0);
    \fi
        \draw[ thick, fill=tensorcolor, rounded corners=2pt] (-#3,-#3) rectangle (#3,#3);
		\draw (0,0) node {\scriptsize #4};
	\end{scope}
}

\newcommand{\GFTensor}[5]{
	\begin{scope}[shift={(#1)}]
    \ifnum#5=0
		\draw[very thick, draw=red] (-#2,0) -- (\singledx+#2,0);
            \draw[very thick, draw=black] (0,#2) -- (0,0);
            \draw[very thick, draw=black] (\singledx,#2) -- (\singledx,0);
    \fi
    \ifnum#5=1
		\draw[very thick, draw=red] (-#2,0) -- (\singledx+#2,0);
            \draw[very thick, draw=black] (0,-#2) -- (0,0);
            \draw[very thick, draw=black] (\singledx,-#2) -- (\singledx,0);
    \fi
    \ifnum#5=2
		\draw[very thick, draw=red] (0,-#2) -- (0,#2);
            \draw[very thick, draw=red] (\singledx,-#2) -- (\singledx,#2);
    \fi
    \ifnum#5=3
		\draw[very thick] (0,-#2) -- (0,#2);
    \fi
    \ifnum#5=4
		\draw[very thick] (-#2,0) -- (#2,0);
    \fi
        \draw[ thick, fill=tensorcolor, rounded corners=2pt] (-#3,-#3) rectangle (#3+\singledx,#3);
		\draw (0+0.5*\singledx,0) node {\scriptsize #4};
	\end{scope}
}

\newcommand{\UFTensor}[5]{
	\begin{scope}[shift={(#1)}]
    \ifnum#5=0
            \draw[very thick] (0,#2) -- (0,-#2);
            \draw[very thick] (\singledx,#2) -- (\singledx,-#2);
    \fi
    \ifnum#5=1
		\draw[very thick, draw=red] (0,-#2) -- (0,0);
            \draw[very thick, draw=red] (\singledx,-#2) -- (\singledx,0);
            \draw[very thick] (0.5*\singledx,#2) -- (0.5*\singledx,0);
    \fi

    \ifnum#5=2
		\draw[very thick, draw=red] (0,#2) -- (0,0);
            \draw[very thick, draw=red] (\singledx,#2) -- (\singledx,0);
            \draw[very thick] (0.5*\singledx,-#2) -- (0.5*\singledx,0);
    \fi
    \ifnum#5=3
		\draw[very thick, draw=red] (0,#2) -- (0,-#2);
            \draw[very thick, draw=red] (\singledx,#2) -- (\singledx,-#2);
    \fi
    \ifnum#5=4
		\draw[very thick] (-#2,0) -- (#2,0);
    \fi
        \draw[ thick, fill=whitetensorcolor, rounded corners=2pt] (-#3,-#3) rectangle (#3+\singledx,#3);
		\draw (0+0.5*\singledx,0) node {\scriptsize #4};
	\end{scope}
}

\newcommand{\GPTensor}[5]{
	\begin{scope}[shift={(#1)}]
    \ifnum#5=0
		\draw[very thick, draw=red] (-#2,0) -- (#2,0);
            \draw[very thick, draw=red] (0,-#2) -- (0,#2);
    \fi
    \ifnum#5=1
		\draw[very thick, draw=red] (-#2,0.2) -- (#2,0.2);
            \draw[very thick, draw=red] (-#2,-0.2) -- (#2,-0.2);
            \draw[very thick, draw=red] (0.2,-#2) -- (0.2,#2);
            \draw[very thick, draw=red] (-0.2,-#2) -- (-0.2,#2);
    \fi

    \ifnum#5=2
		\draw[very thick, draw=red] (-#2,0) -- (#2,0);
            \draw[very thick, draw=red] (0,-#2) -- (0,#2);
    \fi

    \ifnum#5=3
		\draw[very thick] (0,-#2) -- (0,#2);
    \fi
        \draw[ thick, fill=tensorcolor, rounded corners=2pt] (-#3,-#3) rectangle (#3,#3);
	\draw (0,0) node {\scriptsize #4};
    \ifnum#5=0
		\draw[very thick] (#3/2,#3/2) -- (#2,#2);
    \fi
	\end{scope}
}

\newcommand{\UPTensor}[5]{
	\begin{scope}[shift={(#1)}]
    \ifnum#5=0
		\draw[very thick, draw=red] (-#2,0) -- (#2,0);
            \draw[very thick, draw=red] (0,-#2) -- (0,#2);
    \fi
    \ifnum#5=1
		\draw[very thick, draw=red] (-#2,0.2) -- (#2,0.2);
            \draw[very thick, draw=red] (-#2,-0.2) -- (#2,-0.2);
            \draw[very thick, draw=red] (0.2,-#2) -- (0.2,#2);
            \draw[very thick, draw=red] (-0.2,-#2) -- (-0.2,#2);
    \fi

    \ifnum#5=2
		\draw[very thick] (-#2,0) -- (#2,0);
            \draw[very thick] (0,-#2) -- (0,#2);
    \fi

    \ifnum#5=3
		\draw[very thick, draw=red] (-#2,0) -- (#2,0);
            \draw[very thick, draw=red] (0,-#2) -- (0,#2);
    \fi
        \draw[ thick, fill=whitetensorcolor, rounded corners=2pt] (-#3,-#3) rectangle (#3,#3);
	\draw (0,0) node {\scriptsize #4};
    \ifnum#5=0
		\draw[very thick] (#3/2,#3/2) -- (#2,#2);
    \fi
    \ifnum#5=2
		\draw[very thick] (#3/2,#3/2) -- (0.75*#2,0.75*#2);
    \fi
    \ifnum#5=3
		\draw[very thick] (#3/2,#3/2) -- (0.75*#2,0.75*#2);
    \fi
	\end{scope}
}

\newcommand{\EPTensor}[5]{
	\begin{scope}[shift={(#1)}]
    \ifnum#5=0
		\draw[very thick, draw=red] (-#2,0.2) -- (#2,0.2);
            \draw[very thick, draw=red] (-#2,-0.2) -- (#2,-0.2);
            \draw[very thick, draw=red] (0.2,-#2) -- (0.2,#2);
            \draw[very thick, draw=red] (-0.2,-#2) -- (-0.2,#2);
    \fi
    \ifnum#5=1
		\draw[very thick, draw=red] (-#2,0.2) -- (#2,0.2);
            \draw[very thick, draw=red] (-#2,-0.2) -- (#2,-0.2);
            \draw[very thick, draw=red] (0.2,-#2) -- (0.2,#2);
            \draw[very thick, draw=red] (-0.2,-#2) -- (-0.2,#2);
    \fi

    \ifnum#5=2
		\draw[very thick,draw=red] (-#2,0) -- (#2,0);
    \fi

    \ifnum#5=3
		\draw[very thick] (0,-#2) -- (0,#2);
    \fi
        \draw[ thick, fill=tensorcolor, rounded corners=2pt] (-#3,-#3) rectangle (#3,#3);
	\draw (0,0) node {\scriptsize #4};
	\end{scope}
}

\newcommand{\Unitary}[5]{
	\begin{scope}[shift={(#1)}]
    \ifnum#5=0
		\draw[very thick, draw=red] (-#2,0) -- (#2,0);
            \draw[very thick] (-#2,2*#3) -- (#2,2*#3);
    \fi
    \ifnum#5=-1
		\draw[very thick] (0,0) -- (#2,0);
            \draw[very thick] (0,#2) -- (0,0);
    \fi
    \ifnum#5=1
		\draw[very thick, draw=red] (-#2,0) -- (#2,0);
            \draw[very thick] (0,2*#3) -- (#2,2*#3);
    \fi

    \ifnum#5=2
		\draw[very thick, draw=red] (-#2,0) -- (#2,0);
            \draw[very thick, draw=red] (-#2,2*#3) -- (#2,2*#3);
    \fi

    \ifnum#5=3
		\draw[very thick] (0,-#2) -- (0,#2);
    \fi
        \draw[ thick, fill=tensorcolor, rounded corners=2pt] (-#3,-#3) rectangle (#3,3*#3);
		\draw (0,#3) node {\scriptsize #4};
	\end{scope}
}

\newcommand{\FUnitary}[5]{
	\begin{scope}[shift={(#1)}]
    \ifnum#5=0
		\draw[very thick, draw=red] (-#2,0) -- (#2,0);
            \draw[very thick, draw=red] (-#2,2*#3) -- (#2,2*#3);
            \draw[very thick, draw=red] (-#2,4*#3) -- (#2,4*#3);
    \fi
    \ifnum#5=1
		\draw[very thick, draw=red] (-#2,0) -- (#2,0);
            \draw[very thick, draw=red] (0,2*#3) -- (#2,2*#3);
            \draw[very thick, draw=red] (0,4*#3) -- (#2,4*#3);
    \fi
    \ifnum#5=2
		\draw[very thick,draw=red] (-#2,0) -- (#2,0);
    \fi

    \ifnum#5=3
		\draw[very thick] (0,-#2) -- (0,#2);
    \fi
        \draw[ thick, fill=tensorcolor, rounded corners=2pt] (-#3,-#3) rectangle (#3,5*#3);
		\draw (0,2*#3) node {\scriptsize #4};
	\end{scope}
}

\newcommand{\PTensor}[5]{
	\begin{scope}[shift={(#1)}]
    \ifnum#5=0
		\draw[very thick, draw = red] (0,-#2) -- (0,0);
            \draw[very thick, draw = red] (0,0) -- (0,#2);
    \fi
    \ifnum#5=1
		\draw[very thick, draw = red] (-#2,0) -- (0,0);
            \draw[very thick, draw = red] (0,0) -- (#2,0);
    \fi
    \ifnum#5=2
		\draw[very thick, draw = red] (0,-#2) -- (0,0);
            \draw[very thick, draw = red] (0,0) -- (0,#2);
    \fi
    \ifnum#5=3
		\draw[very thick, draw = red] (-#2,0) -- (0,0);
            \draw[very thick, draw = red] (0,0) -- (#2,0);
    \fi
        \draw[ thick, fill=tensorcolor, rounded corners=2pt] (-#3,-#3) rectangle (#3,#3);
		\draw (0,0) node {\scriptsize #4};
	\end{scope}
}

\newcommand{\BTensor}[5]{
	\begin{scope}[shift={(#1)}]
    \ifnum#5=0
		\draw[very thick, draw=red] (-#2,0) -- (#2,0);
    \fi
    \ifnum#5=1
		\draw[very thick,draw=red] (0,#2) -- (0,-#2);
    \fi
    \ifnum#5=2
		\draw[very thick] (-#2,0) -- (#2,0);
    \fi
    \ifnum#5=3
		\draw[very thick] (0,#2) -- (0,-#2);
    \fi

        \draw[ thick, fill=tensorcolor, rounded corners=2pt] (-#3,-#3) rectangle (#3,#3);
		\draw (0,0) node {\scriptsize #4};
	\end{scope}
}

\newcommand{\DTensor}[5]{
	\begin{scope}[shift={(#1)}]
    \ifnum#5=0
		\draw[very thick, draw=red] (-#2,0) -- (#2,0);
    \fi
    \ifnum#5=1
		\draw[very thick,draw=red] (0,#2) -- (0,-#2);
    \fi
    \ifnum#5=2
		\draw[very thick] (-#2,0) -- (#2,0);
    \fi
    \ifnum#5=3
		\draw[very thick] (0,#2) -- (0,-#2);
    \fi

        \draw[ thick, fill=whitetensorcolor, rounded corners=2pt] (-#3,-#3) rectangle (#3,#3);
		\draw (0,0) node {\scriptsize #4};
	\end{scope}
}

\newcommand{\UTensor}[5]{
	\begin{scope}[shift={(#1)}]
    \ifnum#5=0
		\draw[very thick, draw=red] (-#2,0) -- (#2,0);
    \fi
    \ifnum#5=1
		\draw[very thick,draw=red] (0,#2) -- (0,-#2);
    \fi
    \ifnum#5=2
		\draw[very thick] (-#2,0) -- (#2,0);
    \fi
    \ifnum#5=3
		\draw[very thick] (0,#2) -- (0,-#2);
    \fi
        \draw[ thick, fill=operatorcolor, rounded corners=2pt] (-#3,-#3) rectangle (#3,#3);
		\draw (0,0) node {\scriptsize #4};
	\end{scope}
}

\newcommand{\RTensor}[5]{
	\begin{scope}[shift={(#1)}]
    \ifnum#5=0
		\draw[very thick, draw=red] (-#2,0) -- (#2,0);
    \fi
    \ifnum#5=1
		\draw[very thick,draw=black] (0,#2) -- (0,-#2);
    \fi
    \ifnum#5=2
		\draw[very thick] (-#2,0) -- (#2,0);
    \fi
    \ifnum#5=3
		\draw[very thick,draw=red] (-#2,0) -- (#2,0);
    \fi
        \draw[ thick, fill=white, rounded corners=2pt] (-#3,-#3) rectangle (#3,#3);
		\draw (0,0) node {\scriptsize #4};
	\end{scope}
}

\newcommand{\GDTensor}[5]{
	\begin{scope}[shift={(#1)}]
    \ifnum#5=0
		\draw[very thick] (-#2,0) -- (#2,0);
		\draw[very thick] (0,#2) -- (0,-#2);
    \fi
    \ifnum#5=-1
		\draw[very thick] (0,0) -- (#2,0);
		\draw[very thick] (0,#2) -- (0,-#2);
    \fi
    \ifnum#5=1
		\draw[very thick] (-#2,0) -- (0,0);
		\draw[very thick] (0,#2) -- (0,-#2);
    \fi
        \draw[ thick, fill=tensorcolor, rounded corners=2pt] (-#3,-#3) rectangle (#3,#3);
    \def\dx{#3/3};
	\draw [thick]  (-#3+\dx, \dx) -- (- \dx,#3-\dx);
	\draw [thick] (-#3+1.5*\dx,-#3+1.5*\dx) -- (#3-1.5*\dx,#3-1.5*\dx);
	\draw [thick]  ( \dx, -#3 + \dx) -- (#3 - \dx,-\dx);
	\draw (0,0) node {\scriptsize #4};
	\end{scope}
}

\newcommand{\VecCirc}[5]{%
  \begin{scope}[shift={(#1)}]
    \ifnum#4=0 \draw[very thick] (#3,0) -- (#2,0); \fi
    \ifnum#4=1 \draw[very thick] (0,#3) -- (0,#2); \fi
    \ifnum#4=2 \draw[very thick] (-#3,0) -- (-#2,0); \fi
    \ifnum#4=3 \draw[very thick] (0,-#3) -- (0,-#2); \fi
    \draw[very thick, fill=white] (0,0) circle (#3);
    \node at (0,0) {\scriptsize #5};
  \end{scope}%
}

\newcommand\subsetsim{\mathrel{%
  \ooalign{\raise0.2ex\hbox{$\subset$}\cr\hidewidth\raise-0.8ex\hbox{\scalebox{0.9}{$\sim$}}\hidewidth\cr}}}

\newcommand{\PlainTensorFour}[4]{%
  \begin{scope}[shift={(#1)}]
    \draw[very thick] (-#2,0) -- (-#3,0);
    \draw[very thick] ( #3,0) -- ( #2,0);
    \draw[very thick] (0,-#2) -- (0,-#3);
    \draw[very thick] (0, #3) -- (0, #2);
    \draw[thick, fill=tensorcolor, rounded corners=2pt] (-#3,-#3) rectangle (#3,#3);
    \node at (0,0) {\scriptsize #4};
  \end{scope}%
}

\newcommand{\EndVecCirc}[5]{%
  \begin{scope}[shift={(#1)}]
    \ifnum#4=0 \draw[very thick] (-#3,0) -- (-#2,0); \fi 
    \ifnum#4=1 \draw[very thick] (0,-#3) -- (0,-#2); \fi 
    \ifnum#4=2 \draw[very thick] (#3,0) -- (#2,0); \fi 
    \ifnum#4=3 \draw[very thick] (0,#3) -- (0,#2); \fi 
    \draw[very thick, fill=white] (0,0) circle (#3);
    \node at (0,0) {\scriptsize #5};
  \end{scope}%
}

\usetikzlibrary{calc,arrows.meta} 
\usepackage{pgffor} 

\newcommand{\BPEdge}[4]{%
  \pgfmathsetmacro{\rr}{#3}
  \pgfmathsetmacro{\gg}{#4}
  \pgfmathsetmacro{\dd}{0.5*\gg + \rr} 

  \coordinate (BPm)  at ($(#1)!0.5!(#2)$);
  \coordinate (BPcL) at ($(BPm)!\dd!(#1)$);
  \coordinate (BPcR) at ($(BPm)!\dd!(#2)$);

  \coordinate (BPpL) at ($(BPcL)!\rr!(#1)$);
  \coordinate (BPpR) at ($(BPcR)!\rr!(#2)$);

  \draw[very thick] (#1) -- (BPpL);
  \draw[very thick] (#2) -- (BPpR);
  \draw[very thick, fill=white] (BPcL) circle (\rr);
  \draw[very thick, fill=white] (BPcR) circle (\rr);
}

\definecolor{opmaroon}{RGB}{120,30,45}

\usepackage{physics}
\usepackage{amsfonts}
\usepackage{amssymb}
\usepackage{amsthm}
\usepackage{amsmath}
\usepackage{hyperref}
\usepackage{placeins}

\newtheoremstyle{bracketthm}%
  {6pt}
  {6pt}
  {\itshape}
  {}
  {\bfseries}
  {.}
  {0.5em}
  {\thmname{#1}\thmnumber{ #2}\thmnote{ {\normalfont[#3]}}}

\theoremstyle{bracketthm}
\usepackage[ruled,vlined, linesnumbered]{algorithm2e}
\newtheorem{theorem}{Theorem}[section]
\newtheorem{lemma}{Lemma}[section]
\newtheorem{corollary}{Corollary}[section]
\newtheorem{prop}{Proposition}[section]
\newtheorem{defn}{Definition}[section]

\usepackage{cleveref}
\usepackage{mathtools}
\usepackage{bbm}
\usepackage{comment}
\usepackage[dvipsnames]{xcolor}
\definecolor{DarkPastelGreen}{RGB}{72,115,82}
\usepackage{bm}

\newcommand{\llangle}{\langle\!\langle}
\newcommand{\rrangle}{\rangle\!\rangle}
\usepackage[caption=false]{subfig}
\usepackage{graphicx}
\graphicspath{{figures/}}

\def\equationautorefname~#1\null{Eq. (#1)\null}

\newcommand{\C}{\mathbb{C}}

\newcommand{\N}{\mathbb{N}}

\newcommand{\W}{\mathbf{W}}

\newcommand{\loopZ}[1]{Z_{{#1},\lambda}^A}
\newcommand{\zbpl}{Z_{BP,\lambda}^A}
\newcommand{\nloopZ}[1]{\tilde{Z}_{{#1},\lambda}^A}

\newcommand{\tloopZ}[1]{Z_{{#1},\bm{\lambda}}^{\mathbf{A}}}

\newcommand{\NN}{\mathcal{N}}

\newcommand{\FF}{\mathcal{F}}

\newcommand{\LL}{\mathcal{L}}

\newcommand{\TT}{\mathcal{T}}
\newcommand{\ZZ}{\mathcal{Z}}
\newcommand{\MM}{\mathcal{M}}

\newcommand{\BB}{\mathcal{B}}

\newcommand{\supp}{\text{supp}}

\DeclarePairedDelimiterX{\inner}[2]{\langle}{\rangle}{#1|#2}
\DeclarePairedDelimiterX{\expect}[3]{\langle}{\rangle}{#1|#2|#3}

\definecolor{myblue}{RGB}{20, 60, 180}

\definecolor{col}{RGB}{250, 64, 47}
\hypersetup{colorlinks=true, linkcolor=col, citecolor=col, urlcolor=col}

\hypersetup{
    colorlinks=true,
    linkcolor=magenta,
    citecolor=magenta,
    urlcolor=magenta
}

\makeatletter

\newcommand{\l@smsection}{\@dottedtocline{1}{1.5em}{2.3em}}
\newcommand{\l@smsubsection}{\@dottedtocline{2}{3.8em}{3.2em}}

\newcommand{\listofsmentries}{%
  \section*{Contents}%
  \@starttoc{smtoc}%
}

\newcommand{\smsection}[1]{%
  \section{#1}%
  \addcontentsline{smtoc}{smsection}{\protect\numberline{\thesection}#1}%
}

\newcommand{\smsubsection}[1]{%
  \subsection{#1}%
  \addcontentsline{smtoc}{smsubsection}{\protect\numberline{\thesubsection}#1}%
}

\makeatother

\begin{document}
\title{Belief Propagation and Tensor Network Expansions for Many-Body Quantum Systems: Rigorous Results and Fundamental Limits}

\author{Siddhant Midha}
\altaffiliation{Equal contribution}
\email{siddhantm@princeton.edu}
\affiliation{Princeton Quantum Initiative, Princeton University, Princeton, NJ 08544}

\author{Grace M. Sommers}
\thanks{Equal contribution}
\email{gsommers@princeton.edu}
\affiliation{Department of Physics, Princeton University, Princeton, NJ 08544}
\affiliation{Center for Computational Quantum Physics, Flatiron Institute, New York, NY, 10010}

\author{Joseph Tindall}
\affiliation{Center for Computational Quantum Physics, Flatiron Institute, New York, NY, 10010}

\author{Dmitry A. Abanin}
\affiliation{Princeton Quantum Initiative, Princeton University, Princeton, NJ 08544}
\affiliation{Department of Physics, Princeton University, Princeton, NJ 08544}

\begin{abstract}
Belief propagation (BP) provides a scalable heuristic for contracting tensor networks on loopy graphs, but its success in quantum many-body settings has largely rested on empirical evidence. Developing upon a recently introduced cluster-expansion framework for tensor networks, we rigorously study the applicability of BP to many-body quantum systems. For a state represented as a PEPS satisfying a ``loop-decay" condition, we prove that BP supplemented by cluster corrections approximates local observables with exponentially small relative error, and we give explicit formulas expressing local expectation values as BP predictions dressed by connected clusters intersecting the observable region. This representation establishes a direct link between cluster corrections and physical correlation functions. As a result, we show that ``loop-decay" \emph{necessarily implies} exponential decay of connected correlations, yielding sharp, rigorous 
criteria for when BP can and cannot succeed, and ruling out its validity at critical points. Numerical simulations of the two- and three-dimensional transverse field Ising model at zero and finite temperature confirm our analytical predictions, demonstrating quantitative accuracy deep in gapped phases and systematic failure near criticality.
\end{abstract}
\maketitle
\newpage 
\section{Introduction}
Tensor networks have become a cornerstone tool for understanding quantum many-body systems \cite{white1992density,white1993density,schollwock2011density,orus2014practical,cirac2021matrix,verstraete2004matrix,vidal2007entanglement}, naturally encoding gapped ground states and enabling efficient descriptions of short-time dynamics. By encoding correlations through local ``virtual bonds," tensor networks efficiently capture area-law states and enable controlled approximations of their static and dynamical properties.

Tensor networks on a tree admit exact and efficient contraction algorithms. The absence of loops allows contractions to be organized recursively, with all intermediate tensors having dimensions bounded by the bond dimension. Consequently, contraction scales polynomially in system size and bond dimension \cite{Fannes1992,PerezGarcia2007,Shi2006}. The situation is fundamentally different in Euclidean dimensions beyond 1D, where loops are unavoidable. Even when a state obeys an area law and thus admits a tensor-network representation with bounded bond dimension, contraction in the presence of loops can be computationally hard \cite{Schuch2007PEPS,HaferkampPEPS,HarleyPEPS} and the complexity can grow exponentially with system size. 

Practical,~\emph{approximate} contraction algorithms in higher dimensions, such as corner transfer matrix renormalization group (CTMRG)~\cite{nishino1996corner,orus2012exploring,orus2009simulation}, boundary methods 
~\cite{verstraete2004renormalization, lubasch2014unifying, murg2007variational}, or tensor network renormalization based methods~\cite{levin2007tensor,evenbly2015tensor} therefore exploit additional structural properties of the tensor network (e.g., translational invariance) 
{in order to introduce controlled truncations of the intermediate tensors and avoid an exponential scaling. Such truncations of intermediate tensors are typically justified when correlations in the state decay sufficiently fast and can be quantified via the spectrum of an associated transfer operator.} 

A particularly promising approach to contraction is \emph{belief propagation} (BP), a message-passing algorithm originally developed for probabilistic inference on graphical models~\cite{Gallager1962,pearl1988probabilistic,mceliece1998turbo,yedidia2003,Yedidia2005}, with even earlier roots in statistical mechanics~\cite{bethe1935statistical,Peierls1936,Thouless1987,Mezard1986cavity,Klein1979,Katsura1979,Nakanishi1981}. When applied to tensor networks~\cite{robeva2016duality, BPsimpleupdate,sahu2022efficient, BP_gauging}, BP provides a {tree-like} approximation that can be computed at polynomial cost. BP is exact on trees; on loopy graphs, it becomes an approximation whose accuracy depends on how strongly loops in the network affect the contraction. While there exist rigorous results on loopy BP in the classical setting~\cite{Weiss2000,Tatikonda2002,NIPS2002,Heskes2004,Mooij2005,chertkov2008exactness,Sudderth2008}, the theoretical foundation for when and why BP works in the quantum many-body setting has remained largely heuristic, despite widespread empirical success~\cite{liao2023simulation,  tindall2024efficient, Begusic2024, tindall2025dynamics,rudolph2025simulatingsamplingquantumcircuits, park2025}.

Recent advances in cluster expansion techniques for tensor networks \cite{evenbly2024loopseriesexpansionstensor,park2025,midha2025beyond,gray2025,evenbly2025partitioned} provide the missing theoretical framework. Building on these developments, we establish rigorous criteria for the applicability of BP to quantum states represented as projected entangled pair states (PEPS) on arbitrary geometries. We derive exact formulas expressing local expectation values as the BP prediction dressed by all connected clusters intersecting the observable region (see \autoref{prop:localexpexpansion} and \autoref{prop:localexpexpansion-derivative}). When loop corrections decay exponentially, truncating this expansion at finite cluster order yields a relative-error approximation with an error decaying exponentially fast with the cluster order (see \autoref{thm:localobsalgorithm}). 
This transforms BP from a heuristic into a systematically improvable algorithm with rigorous performance guarantees.

The performance of BP-based tensor network methods is intimately tied to the magnitude of loop corrections. This naturally raises a structural question: what is the physical meaning of these corrections, and what do they reveal about the underlying quantum state? We prove that loop tensors act as the ``carriers" of connected correlation functions (see \autoref{prop:correlatorexpansion} and \autoref{prop:conn-correlator}). 
This can be thought of as a generalization of the MPS ``transfer-matrix'' picture to arbitrary graphs. By applying a cluster expansion to connected correlation functions, we show decay of loopy excitations\footnote{The excitations appear as loops, as well as strings allowed to terminate on the region of observable insertion.} \emph{necessarily implies} exponential decay of connected correlations (see \autoref{thm:correlatorbound}). 
Loop decay is therefore not merely an algorithmic criterion but a physical statement about the state itself. Conversely, at critical points or in gapless phases, where correlations decay sub-exponentially, loop corrections must remain parametrically large. This establishes loop decay as a sharp diagnostic for when BP-based expansions can succeed and delineates the fundamental limits of such methods.

Our results also provide concrete operational criteria for practitioners: measuring loop decay in the tensor network immediately indicates (i) whether a given fixed-point
can be trusted, (ii) the order of cluster corrections needed for a target accuracy, and (iii) loop decay criteria to ensure clustering of correlations.
{Conversely, failure of loop decay signals either criticality or a mismatch between the chosen BP fixed point and the physical state. In such situations, the correct description may require expanding around a distinct, potentially ``unstable," fixed point. We provide explicit examples where loop decay is violated at a stable fixed point, yet restored when the expansion is centered on the appropriate unstable one.}

This emphasizes a crucial caveat in  BP based algorithms: the expansions must assume a \emph{good} BP fixed point has been found. A convergent series cannot be built atop a poor mean field. While good fixed points may exist, they may be unattainable from standard iterative message passing, particularly near criticality. This highlights a crucial algorithmic challenge beyond convergence of the cluster expansion: the fixed-point problem. While there exist strategies for circumventing this issue, such as partitioned network expansions~\cite{evenbly2025partitioned} and generalized belief propagation~\cite{Kikuchi1951,Yedidia2001,NIPS2001,Yedidia2005}, which we develop upon in upcoming work~\cite{tindallgbp}, our focus here will be on the performance of ``vanilla'' BP. In upcoming work~\cite{midha26upcoming}, we provide a class of PEPS states for which the fixed-point problem can be rigorously solved.

We validate our analytical predictions through extensive numerical simulations of the two- and three-dimensional transverse field Ising model at zero and finite temperature, represented as a translationally-invariant PEPS. Deep in gapped phases, cluster corrections converge rapidly and achieve high quantitative accuracy. Approaching the critical point, loop decay weakens and corrections grow, demonstrating the systematic failure of the method at criticality as predicted by the theory.

The rest of the paper is organized as follows. In \autoref{sec:TNexpansions} we first review the formalisms of tensor networks, belief propagation, and various expansions with convergence guarantees. We study local expectation values in \autoref{sec:local_observables}. The connection between loop corrections and correlators is then outlined in \autoref{sec:correlation_fns}. We then demonstrate numerically the cluster-corrected BP algorithm on the transverse field Ising model in \autoref{sec:numerics}. We conclude with a discussion of open questions and future work in \autoref{sec:discussions}. 
\section{Tensor network expansions}\label{sec:TNexpansions}

This section reviews the technical machinery underlying our analysis: systematic series expansions for tensor network contractions on loopy graphs. While these methods were originally developed for classical statistical mechanics and factor graphs~\cite{ruelle1969statistical,Kotecky1986,Chertkov2006,Chertkov2006a,Gomez2007,chertkov2009approximate,welling2012}, recent work \cite{Evenbly2024,park2025,midha2025beyond,gray2025} has adapted them to the tensor network setting, providing the foundation for a rigorous error analysis of belief propagation.

\subsection{Setup and notation}
We consider \emph{closed} tensor networks (i.e., no ``open legs'') defined on a graph $G = (V,E)$ with $N$ vertices $V$ and edges $E$. For each vertex $v \in V$, we denote its neighbors in $G$ as $\NN(v)$. The degree of vertex $v$ is $d(v) := |\NN(v)|$, and the maximum degree of the graph is $\Delta := \max_{v\in V}d(v)$. 

Each edge $(v,w)\in E$ is associated with a \textit{bond Hilbert space} $\BB_{vw}$ of dimension $D$ (the \textit{bond dimension}, taken uniform for simplicity). Each vertex $v\in V$ is equipped with a tensor $T_v \in \bigotimes_{n \in \NN(v)}\BB_{nv}$ making up the tensor network.

Contracting all tensors in the network yields a scalar:
\begin{equation}
    \ZZ= \star_{v\in V}T_v,
\end{equation}
where $\star$ denotes contraction of tensor indices. Following the language of statistical mechanics, we call $\ZZ$ the \emph{partition function} of the tensor network. Note that $\ZZ$ is more general than in statistical mechanics: tensors $T_v$ can be complex-valued, whereas classical partition functions only sum non-negative quantities. We define the \emph{free energy} as
\begin{equation}
    \FF = -\log{\ZZ},
\end{equation}
where all logarithms in this work are base $e$.
Throughout, we choose normalizations such that $\ZZ$ is positive, storing phase information separately. This ensures the uniqueness of $\FF$ in the regime where cluster expansions are well-defined.

\subsection{Belief propagation fixed points}
The starting point for all of our expansions is the \emph{belief propagation (BP) approximation}, which approximates the tensor network by a rank-one factorization. For each edge $e=(v,w)$ in the network, we introduce \textit{message tensors} $\mu_{v \to w} \in \BB_{vw}$ (and $\mu_{w \to v}$) that satisfy a \textit{self-consistency} condition. 

Specifically, a fixed-point set $\MM = \{\mu_{v\to w}\}_{(v,w)\in E}$ requires that the contraction of all but one incoming message at any vertex $v\in V$ reproduces the outgoing message on the excluded edge. Mathematically, for each $v\in V$ and each neighbor $n_j \in \NN(v)$:
\begin{equation}
        \left(\bigotimes_{n_i \in \NN(v)\setminus\{n_j\}}\mu_{n_i \to v} \right)\star T_v \propto \mu_{v\to n_j}
\end{equation}
Schematically:
\[
\begin{tikzpicture}[scale=0.9, baseline={([yshift=-0.65ex] current bounding box.center)}]

  \def\L{1.15}   
  \def\a{0.45}   
  \def\stub{0.45}
  \def\r{0.14}

  \PlainTensorFour{0,0}{\L}{\a}{\small $T_v$}

  \EndVecCirc{-\L,0}{\stub}{\r}{2}{}
  \EndVecCirc{0,\L}{\stub}{\r}{1}{}
  \EndVecCirc{0,-\L}{\stub}{\r}{3}{}

  \node at (2.0,0) {\Large $\propto$};

\EndVecCirc{3.2,0}{1.0}{\r}{2}{}

\end{tikzpicture}
\]

where rank-one tensors $\tikz[baseline=-0.6ex]{
  \draw (0,0) circle (0.6ex);
  \draw (0.6ex,0) -- (2ex,0);
}$ 
are messages, and $T_v$ is the local tensor at vertex $v$. The self-consistency condition ensures that this rank-one subspace is invariant under contraction with surrounding tensors, making it a stationary point of a variational optimization problem. Physically, the fixed point singles out a ``mean-field" subspace that captures the dominant correlations, while fluctuations orthogonal to this subspace represent corrections.

The BP approximation amounts to replacing exact tensor contractions with contractions over these rank-one messages:

\[
\begin{tikzpicture}[scale=0.95, baseline={([yshift=-0.65ex] current bounding box.center)}]
  \def\dx{1.8}
  \def\dy{1.6}
  \def\a{0.33}

  \newcommand{\SqTN}[3]{%
    \coordinate (#1) at (#2,#3);
    \draw[thick, fill=tensorcolor, rounded corners=2pt]
      ($(#1)+(-\a,-\a)$) rectangle ($(#1)+(\a,\a)$);
    \node at (#1) {\scriptsize $T$};
  }

  \SqTN{A}{0}{0}
  \SqTN{B}{\dx}{0}
  \SqTN{C}{0}{-\dy}
  \SqTN{D}{\dx}{-\dy}
  \SqTN{E}{0}{-2*\dy}
  \SqTN{F}{\dx}{-2*\dy}

  \coordinate (AE) at ($(A)+(\a,0)$);   \coordinate (BW) at ($(B)+(-\a,0)$);
  \coordinate (CE) at ($(C)+(\a,0)$);   \coordinate (DW) at ($(D)+(-\a,0)$);
  \coordinate (EE) at ($(E)+(\a,0)$);   \coordinate (FW) at ($(F)+(-\a,0)$);

  \coordinate (AS) at ($(A)+(0,-\a)$);  \coordinate (CN) at ($(C)+(0,\a)$);
  \coordinate (BS) at ($(B)+(0,-\a)$);  \coordinate (DN) at ($(D)+(0,\a)$);
  \coordinate (CS) at ($(C)+(0,-\a)$);  \coordinate (EN) at ($(E)+(0,\a)$);
  \coordinate (DS) at ($(D)+(0,-\a)$);  \coordinate (FN) at ($(F)+(0,\a)$);

  \draw[very thick] (AE) -- (BW);
  \draw[very thick] (CE) -- (DW);
  \draw[very thick] (EE) -- (FW);

  \draw[very thick] (AS) -- (CN);
  \draw[very thick] (BS) -- (DN);
  \draw[very thick] (CS) -- (EN);
  \draw[very thick] (DS) -- (FN);
\end{tikzpicture}
\;\approx\;
\begin{tikzpicture}[scale=0.95, baseline={([yshift=-0.65ex] current bounding box.center)}]
  \def\dx{1.8}
  \def\dy{1.6}
  \def\a{0.33}
  \def\r{0.1}
  \def\gap{0.4}

  \newcommand{\SqTN}[3]{%
    \coordinate (#1) at (#2,#3);
    \draw[thick, fill=tensorcolor, rounded corners=2pt]
      ($(#1)+(-\a,-\a)$) rectangle ($(#1)+(\a,\a)$);
    \node at (#1) {\scriptsize $T$};
  }

  \SqTN{A}{0}{0}
  \SqTN{B}{\dx}{0}
  \SqTN{C}{0}{-\dy}
  \SqTN{D}{\dx}{-\dy}
  \SqTN{E}{0}{-2*\dy}
  \SqTN{F}{\dx}{-2*\dy}

  \coordinate (AE) at ($(A)+(\a,0)$);   \coordinate (BW) at ($(B)+(-\a,0)$);
  \coordinate (CE) at ($(C)+(\a,0)$);   \coordinate (DW) at ($(D)+(-\a,0)$);
  \coordinate (EE) at ($(E)+(\a,0)$);   \coordinate (FW) at ($(F)+(-\a,0)$);

  \coordinate (AS) at ($(A)+(0,-\a)$);  \coordinate (CN) at ($(C)+(0,\a)$);
  \coordinate (BS) at ($(B)+(0,-\a)$);  \coordinate (DN) at ($(D)+(0,\a)$);
  \coordinate (CS) at ($(C)+(0,-\a)$);  \coordinate (EN) at ($(E)+(0,\a)$);
  \coordinate (DS) at ($(D)+(0,-\a)$);  \coordinate (FN) at ($(F)+(0,\a)$);

  \BPEdge{AE}{BW}{\r}{\gap}
  \BPEdge{CE}{DW}{\r}{\gap}
  \BPEdge{EE}{FW}{\r}{\gap}

  \BPEdge{AS}{CN}{\r}{\gap}
  \BPEdge{BS}{DN}{\r}{\gap}
  \BPEdge{CS}{EN}{\r}{\gap}
  \BPEdge{DS}{FN}{\r}{\gap}
\end{tikzpicture}\,,
\]
where the rank-one subspaces are suitably normalized, as we discuss in the following section [see \autoref{eq:identityresolution}].
The circles indicate where exact bond contractions are replaced by products of BP messages. The errors in this approximation---the content of the ``fluctuations"---are systematically captured by the series expansions.  

\subsection{Loop expansion}
The first systematic approach to quantifying BP errors is the \emph{loop series expansion} \cite{Chertkov2006,Chertkov2006a,Gomez2007,chertkov2009approximate}, applied to tensor networks in Ref.~\cite{Evenbly2024}. The key idea is to decompose each bond Hilbert space into the BP subspace (spanned by the messages) plus its orthogonal complement, capturing ``excitations" away from the mean field ``vacuum.''

For each edge $e = (v,w)\in E$, we expand the identity on the bond space as,
\begin{eqnarray} \label{eq:identityresolution}
    \mathbbm{1}_{vw} = \underbrace{\frac{\mu_{v\to w} \otimes \mu_{w\to v}}{\mu_{v\to w} \star \mu_{w\to v}}}_{\text{BP subspace}} + \underbrace{\mathcal{P}_{vw}^\perp}_{\text{excitation subspace}},
\end{eqnarray}
The normalization ensures that these subspaces are orthogonal: $\mathcal{P}_{vw}^\perp\star \mu_{v\to w} = 0$ and $\mu_{w\to v}\star\mathcal{P}_{vw}^\perp = 0$, hence $\mathcal{P}_{vw}^\perp$ carries contributions orthogonal to the BP vacuum. For any subset of edges $F \subseteq E$, contraction with the projector leads to the following excitation, 
\begin{equation}
  \tilde{Z}_F = \left( \prod_{(w,v)\in F} \mathcal{P}_{vw}^\perp \right) \star \left( \prod_{(v,w) \notin F} \mu_{wv} \otimes \mu_{vw} \right) \star \left( \prod_{v} T_v \right).
\end{equation}
The self-consistency condition ensures that the only non-zero excitations that appear have no dangling edges -- they are leafless edge-induced subgraphs of $G$ -- leading to the following loop-series expansion. 

\begin{equation}\label{eq:loopseriesnaive}
  \ZZ = Z_{\text{BP}} + \sum_{\text{all loops}} \tilde{Z}_l
\end{equation}
where $\tilde{Z}_l$ is the (unnormalized) contribution of a loop excitation $l$, and $Z_{\text{BP}} = \tilde{Z}_{\emptyset}$ is the contraction with all BP messages. For a simple single-loop TN, this is schematically, 

\[
\begin{tikzpicture}[
  scale=0.6,
  baseline={([yshift=-0.65ex] current bounding box.center)},
  tensor/.style={draw, line width=0.9pt, fill=tensorcolor, rounded corners=2pt, minimum size=8mm},
  bond/.style={line width=1.6pt},
  redbond/.style={line width=1.6pt, draw=red}
]
  \def\dx{2.5}
  \def\dy{2.5}

  \node[tensor] (A) at (0,0) {};
  \node[tensor] (B) at (\dx,0) {};
  \node[tensor] (C) at (\dx,-\dy) {};
  \node[tensor] (D) at (0,-\dy) {};

  \draw[bond] (A.east) -- (B.west);
  \draw[bond] (B.south) -- (C.north);
  \draw[bond] (D.east) -- (C.west);
  \draw[bond] (A.south) -- (D.north);
\end{tikzpicture}
\;=\;
\begin{tikzpicture}[
  scale=0.6,
  baseline={([yshift=-0.65ex] current bounding box.center)},
  tensor/.style={draw, line width=0.9pt, fill=tensorcolor, rounded corners=2pt, minimum size=8mm},
  bond/.style={line width=1.6pt}
]
  \def\dx{2.5}
  \def\dy{2.5}

  \def\r{0.11}    
  \def\gap{0.4}  

  \node[tensor] (A) at (0,0) {};
  \node[tensor] (B) at (\dx,0) {};
  \node[tensor] (C) at (\dx,-\dy) {};
  \node[tensor] (D) at (0,-\dy) {};

  \BPEdge{A.east}{B.west}{\r}{\gap}
  \BPEdge{B.south}{C.north}{\r}{\gap}
  \BPEdge{D.east}{C.west}{\r}{\gap}
  \BPEdge{A.south}{D.north}{\r}{\gap}
\end{tikzpicture}
\;+\;
\begin{tikzpicture}[
  scale=0.6,
  baseline={([yshift=-0.65ex] current bounding box.center)},
  tensor/.style={draw, line width=0.9pt, fill=tensorcolor, rounded corners=2pt, minimum size=8mm},
  redbond/.style={line width=1.6pt, draw=red}
]
  \def\dx{2.5}
  \def\dy{2.5}

  \node[tensor] (A) at (0,0) {};
  \node[tensor] (B) at (\dx,0) {};
  \node[tensor] (C) at (\dx,-\dy) {};
  \node[tensor] (D) at (0,-\dy) {};

  \draw[redbond] (A.east) -- (B.west);
  \draw[redbond] (B.south) -- (C.north);
  \draw[redbond] (D.east) -- (C.west);
  \draw[redbond] (A.south) -- (D.north);
\end{tikzpicture}
\]
Crucially, \autoref{eq:loopseriesnaive} sums over all loops, including non-simple cycles and even disconnected configurations. To simplify this, we normalize the tensor network by the BP value $Z_{\text{BP}}$ and define $\LL$ as the set of all \textit{connected} edge-induced subgraphs with minimum degree two---the set of connected ``generalized loops" on $G$. Loop excitations evaluated on the normalized tensor network are denoted $Z_l$. Two loops $l,l' \in \LL$ are \emph{compatible} if they share no edges or vertices, that is, they are disconnected. A subset $\Gamma\subset \LL$ is called compatible if all distinct loops in the subset are compatible with each other. This gives us the loop series expansion for a tensor network. 

\begin{prop}[Loop Expansion (Informal) \cite{Evenbly2024,midha2025beyond}]\label{prop:loopseries}
    The tensor network partition function admits the following series expansion, 
    \begin{equation}\label{eq:loopseries}
  \ZZ = Z_{\emph{BP}}\left[ 1 + \sum_{\substack{\Gamma\subseteq\mathcal{L}\\\Gamma\text{ \emph{finite, compatible}}}}
  \prod_{l\in\Gamma} Z_l\right]
\end{equation}
\end{prop}

The sum-product structure in \autoref{eq:loopseries} includes all configurations of mutually compatible loops, including disconnected collections. While this expansion is exact, its combinatorics are unfavorable: the proliferation of disconnected loops prevents straightforward convergence analysis. This motivates the cluster expansion.

\subsection{Cluster expansion}
To remedy the unfavorable combinatorics, we shift from expanding the partition function $\ZZ$ to expanding the free energy $\FF = -\log \ZZ$. Taking the logarithm converts products into sums and naturally isolates \emph{connected} contributions, dramatically improving convergence properties.

A \emph{cluster} is a collection of loops with multiplicities: $$\mathbf{W} = \{(l_1, \eta_1), (l_2, \eta_2), \ldots \}$$
where each $l_i\in \LL$ is a loop and $\eta_i \geq 1$ is its multiplicity in the cluster. The corresponding cluster evaluation is, 
\begin{equation}
  Z_{\mathbf{W}} = \prod_{i} Z_{l_i}^{\eta_i}.
\end{equation}
Moreover, the order of a given cluster $\mathbf{W}$ is the sum of all loop orders, including multiplicity,
\begin{equation}
    |\mathbf W| := \sum_i \eta_i |l_i|
\end{equation}
where $|l|$ for a loop $l\in\LL$ denotes the number of edges in $l$. We can now reorganize the loop expansion into a cluster expansion.
\begin{prop}[Cluster Expansion (Informal) \cite{midha2025beyond}] \label{prop:clusterexp}
    The free-energy of the tensor network admits the following series expansion, 
    \begin{equation}
        \FF = \FF_0 - \sum_{\emph{conn. }\mathbf W} \phi_{\mathbf W} Z_{\mathbf W}
    \end{equation}
    where $\phi_{\mathbf W}$ is the Ursell function, depending solely on the cluster geometry.
\end{prop}
The formal statement along with the definition of the Ursell function can be found in the Supplementary Material~\cite{SM}.
The key advantage is that only \emph{connected} clusters appear. A cluster is said to be connected if every pair of loops in the cluster share an edge or vertex. This eliminates the combinatorial explosion of disconnected configurations and leads directly to convergence guarantees.

We state a condition on the tensor network under which convergence can be proven.

\begin{defn}[$c-$decay of loops]
  If all the loop weights for a given TN satisfy the following decay condition: 
    \begin{equation}
        |Z_l| \leq \order{e^{-c|l|}}
    \end{equation}
    for some $c > 0$, {where $|l|$ is the number of edges in loop $l$,} then we say that the TN satisfies $c-$decay of loops. 
\end{defn}
The following result from Ref.~\cite{midha2025beyond} establishes that $c-$decay above a constant threshold $c_0$, depending on the graph geometry, renders the cluster expansion an efficient algorithm for TN contraction. 

\begin{theorem}[Convergence of cluster expansion \cite{midha2025beyond} (Informal)]
    \label{thm:clusterconv}
    Let $F_m$ be the approximation obtained by truncating the cluster expansion at order $m$. Then, $c-$decay of loops with $c > c_0 = \order{\log\Delta}$, 
    {where $\Delta$ is the maximal degree of the graph,} ensures 
    \begin{equation}
        |F_m - \FF| \leq \order{N e^{-(c-c_0)(m+1)}}
    \end{equation}
    with $N$ being the number of vertices in the network.
\end{theorem}
From here onwards, we refer to $c-$decay of loops with any $c > c_0$ simply as ``loop decay." We note that the threshold $c_0$ is a worst-case bound. In many tensor network instances (e.g., the 2D classical Ising model at high temperatures) one finds that certain loop configurations (e.g., loops containing vertices of odd degree in the Ising model) vanish identically owing to the symmetries of the problem. In practice, one can obtain convergence for weaker decay constants $c_0$ (e.g., when all odd-degree loops are vanishing). \autoref{thm:clusterconv} ensures convergence in the worst case, where loops of all weights decay as $\sim e^{-c|l|}$ for some $c > c_0$. This ensures that the BP approximation to the free-energy \emph{density} $f = \FF/N$ satisfies
\begin{equation} \label{eq:bperror}
  |f_{\text{BP}} - f| \leq e^{-(c-c_0)}.
\end{equation}
{This yields an explicit bound on the BP error controlled by loop decay. We repeat the caveat that these expansions must be performed around a ``good'' BP fixed point. 
By ``good,'' we mean one for which there exists a uniform constant $c > 0$ such that the corresponding loop corrections satisfy $c$-decay.} 
If the fixed point is poor (e.g., near criticality), even the zeroth-order approximation fails, and higher-order corrections cannot rescue it. 
We return to this point in~\autoref{sect:confusion} below.

\subsection{Cluster-cumulant expansion}
While the cluster expansion converges rapidly under loop decay, it still has redundancy: each individual loop $l$ appears in infinitely many clusters (at all multiplicities $\eta = 1, 2, 3, \ldots$). Even for a single-loop network, the cluster expansion is an infinite series. The \emph{cumulant method} resolves this by summing all orders of a given loop configuration at once, yielding a finite expansion for finite loop sets~\cite{welling2012,park2025,gray2025}.

To sum up all orders of contributions from loops within a region, we define for any finite, connected subset $\Gamma \subseteq \LL$, a \emph{cumulant} $\mathfrak{K}: \Gamma \to \C$ as follows,

\begin{equation}
    \mathfrak K(\Gamma) := \sum_{\substack{\mathbf W\ \mathrm{connected}\\ \text{loop-set}(\mathbf W)=\Gamma}}
    \phi_{\mathbf W}\,Z_{\mathbf W}
\end{equation}
where the loop-set of a cluster is the set of loops with non-zero multiplicity in the cluster. With this rearrangement of the cluster contributions, we can restructure the expansion as follows, 

\begin{equation}
    \sum_{\text{connected }\mathbf W}
  =\sum_{\Gamma\subseteq\LL}\ \sum_{\substack{\mathbf W\ \text{connected}\\ \text{loop-set}(\mathbf W)=\Gamma}}.
\end{equation}

Crucially, the sum over clusters satisfying $\text{loop-set}(\mathbf W) = \Gamma$ can be done efficiently in a recursive fashion. This makes the rearrangement efficient. To do so, we define the \emph{restricted partition function} of a subset $B$ of the excitations $\LL$ as, 
\begin{equation}
\Xi(B)=1+\!\!\sum_{\substack{\Gamma'\subseteq B\\\Gamma'\ \mathrm{compatible}}}\!
   \prod_{l\in\Gamma'} Z_l.
\end{equation}
The cluster expansion from \autoref{prop:clusterexp} implies that,
\begin{equation}
    \log\Xi(B) = \sum_{\Gamma \subseteq B} \mathfrak{K}(\Gamma),
\end{equation}
Mathematically, this expression has the structure of a \emph{cumulative} function $\log\Xi$ defined on the ordered set of subsets, where the ordering is defined by set containment. By a M\"{o}bius transformation [see Supplementary~\cite{SM} for details], we obtain
\begin{equation} \label{eq:inclusionexclusion}
  \mathfrak{K}(\Gamma)
    = \sum_{B \subseteq \Gamma} (-1)^{|\Gamma|-|B|}\, \log\Xi(B),
\end{equation}
where $|X|$ denotes the number of loops in the set $X$.
This equation can equivalently be viewed as the inclusion-exclusion principle. With this, we arrive at the cluster-cumulant expansion on a tensor network. 
\begin{prop}[Cluster-cumulant expansion (Informal)] \label{prop:clustercumulantexp}
    The free energy can be written as the sum of cumulants of all connected subgraphs $\Gamma$ of the loop set $\LL$. 
    \begin{equation}\label{eq:cc-expansion}
  {\;
 \FF 
  = \FF_0
  -\sum_{\emph{conn.} \Gamma\subseteq\LL}
   \mathfrak K(\Gamma).
  \;}
\end{equation}
\end{prop}
For example, a single-loop network has $\Gamma=\{l\}$, giving $\mathfrak{K}(\Gamma) = \log{(1+Z_l)}$---a single term that exactly captures all orders of that loop.

Computationally, evaluating the cumulant of a set $\Gamma$ with $m$ loops costs $\sim 2^m$ (exponential in $m$), the same scaling as an order-$m$ cluster correction. However, the cumulant representation is often more compact, requiring fewer terms for the same accuracy. Crucially, convergence is inherited from the cluster expansion: exponential loop decay guarantees exponential convergence in the cumulant series as well.

Before proceeding, let us remark upon an equivalent formulation of the cluster-cumulant expansion that comes from reorganizing the sum in~\autoref{eq:cc-expansion}:
\begin{align}\label{eq:cc-top-down}
\sum_{\mathrm{conn.} \Gamma \subseteq\LL} \mathfrak{K}(\Gamma) &= \sum_{\mathrm{conn.} \Gamma\subseteq \LL} \sum_{B \subseteq \Gamma} (-1)^{|\Gamma| - |B|} \log \Xi(B) \notag \\
&:= \sum_{\mathrm{conn.} B\subseteq \LL} b_\LL(B) \log \Xi(B)
\end{align}
where we have defined the ``counting number''
\begin{equation}
b_{\LL}(B) = 1 - \sum_{\mathrm{conn.}\, \Gamma \supset B} b_\LL(\Gamma).
\end{equation}
with $b_\LL(A) = 1$ for all parentless subsets at that order.~\autoref{eq:cc-top-down} is a ``top-down''  formulation of the inclusion-exclusion principle, with the counting numbers computed recursively starting from the largest subsets. 

From \autoref{eq:cc-top-down}, we can further make the connection to prior discussions of the cluster cumulant expansion that appear in the literature~\cite{welling2012,gray2025}. To do so, we shift perspective from subsets of generalized loops to vertex-induced connected subgraphs of $G$ (``regions'')~\cite{welling2012,gray2025}. After fixing a maximum region size $k$ and identifying parentless connected regions up to that size, the remaining  regions are obtained via recursive intersection. Further details on the construction of the regions and the associated expansion in $k$ are provided in the Supplemental Material~\cite{SM}. 

\subsection{The Fixed Point Problem}\label{sect:confusion}
When the message-passing equations possess a unique fixed point, checking loop decay is straightforward. In generic tensor networks, however, the fixed point may be non-unique, and the fixed point which minimizes the Bethe free energy $\FF_0$ may not be the appropriate starting point for a cluster or cluster-cumulant expansion. Even if there exists \textit{some} fixed point for which the expansion converges, it may be an unstable fixed point of the message passing algorithm, and thus will not be found unless one initializes precisely at that point. In frustrated models, message passing may fail to converge to any fixed point, unless one resorts to more sophisticated update schemes~\cite{NIPS2002,heskes2002}.

As an example, consider a model with a phase boundary between a paramagnetic and ferromagnetic phase. Because mean-field approximations neglect certain correlations and thus overestimate the tendency to order, one always encounters what we dub a ``confusion regime'' on the paramagnetic side of the transition. That is, the message-passing algorithm will converge towards a stable, symmetry-broken fixed point before the true phase boundary is reached. This stable fixed point is in the wrong phase and thus not a good basis for a cluster expansion. Fundamentally, this manifests as a consequence of approximating a lattice as a tree. In some cases, we can rescue the expansion by using an unstable ``paramagnetic'' fixed point, as in the classical 2D Ising model between the BP transition $\beta_{\rm BP} = \log(2)/2$ and the exact transition $\beta_c = \log(1+\sqrt{2})/2$. 
Finding unstable fixed points is a delicate matter, however: while one can explicitly solve for all translation-invariant fixed points of the BP equations at small constant bond dimensions, at moderate bond dimensions one must typically resort to iterative message passing, which generically misses the unstable solutions. One potential remedy, left to future work, is to impose the symmetry at the level of the tensors, and initialize message passing from symmetric messages.
\section{Local observables}
\label{sec:local_observables}
Having established the cluster expansion framework for partition functions, we now turn to the central question of this work: computing local expectation values in quantum states represented as PEPS. This is where the practical power of the expansion becomes apparent.

\subsection{Setup: PEPS and observable insertion}
Consider a quantum state $\ket{\Psi}$ represented as a PEPS on a graph $G$:

\[
\ket{\Psi}\;\equiv\;
\begin{tikzpicture}[
  scale=0.8,
  baseline={([yshift=-0.65ex] current bounding box.center)},
  tensor/.style={draw, thick, fill=tensorcolor, rounded corners=2pt, minimum size=8mm},
  bond/.style={very thick},
  phys/.style={very thick}
]

\def\dx{1.55}     
\def\dy{1.40}
\def\plen{0.60}   
\def\shear{0.0}  

\begin{scope}[xslant=\shear, yscale=0.98]

  \foreach \r in {0,1,2,3,4}{
    \foreach \c in {0,1,2,3,4}{
      \node[tensor] (T\c\r) at (\c*\dx, -\r*\dy) {};
      \coordinate (C\c\r) at (T\c\r.center);
    }
  }

  \foreach \r in {0,1,2,3,4}{
    \foreach \c in {0,1,2,3}{
      \draw[bond] (T\c\r.east) -- (T\the\numexpr\c+1\relax\r.west);
    }
  }

  \foreach \c in {0,1,2,3,4}{
    \foreach \r in {0,1,2,3}{
      \draw[bond] (T\c\r.south) -- (T\c\the\numexpr\r+1\relax.north);
    }
  }

\end{scope}

\foreach \r in {0,1,2,3,4}{
  \foreach \c in {0,1,2,3,4}{
    \draw[phys] (C\c\r) -- ++(-\plen,\plen);
  }
}

\end{tikzpicture}
\]

The \emph{norm network} $\ZZ = \braket{\psi}{\psi}$ is obtained by contracting the bra and ket:

\[  
\langle\Psi|\Psi\rangle \;\equiv\;
\begin{tikzpicture}[
  scale=0.8,
  baseline={([yshift=-0.65ex] current bounding box.center)},
  tensor/.style={
    draw,
    line width=0.9pt,
    fill=tensorcolor!80!black, 
    rounded corners=2pt,
    minimum size=8mm
  },
  bond/.style={line width=2.0pt} 
]

\def\dx{1.55}
\def\dy{1.40}
\def\shear{0.0}

\begin{scope}[xslant=\shear, yscale=0.98]

  \foreach \r in {0,1,2,3,4}{
    \foreach \c in {0,1,2,3,4}{
      \node[tensor] (T-\c-\r) at (\c*\dx, -\r*\dy) {};
    }
  }

  \foreach \r in {0,1,2,3,4}{
    \foreach \c in {0,1,2,3}{
      \pgfmathtruncatemacro{\cp}{\c+1}
      \draw[bond] (T-\c-\r.east) -- (T-\cp-\r.west);
    }
  }

  \foreach \c in {0,1,2,3,4}{
    \foreach \r in {0,1,2,3}{
      \pgfmathtruncatemacro{\rp}{\r+1}
      \draw[bond] (T-\c-\r.south) -- (T-\c-\rp.north);
    }
  }

\end{scope}
\end{tikzpicture}
\]

As before, we denote the set of connected generalized loops on $G$ as $\LL$. We can apply the BP approach along with the associated cluster expansion machinery to this network, denoting a fixed point by $\MM = \{\mu_{v\to w}\}$.

Now consider a local observable $O_A$ supported on a finite region $A \subset V$ with $|A| = \order{1}$. Inserting this observable creates a modified tensor network $\ZZ^A = \expect{\psi}{O_A}{\psi}$, which differs from $\ZZ$ only within region $A$:

\[
\langle\Psi|\,O_A\,|\Psi\rangle \;\equiv\;
\begin{tikzpicture}[
  scale=0.8,
  baseline={([yshift=-0.65ex] current bounding box.center)},
  tensor/.style={
    draw,
    line width=0.9pt,
    fill=tensorcolor!80!black,
    rounded corners=2pt,
    minimum size=8mm
  },
  optensor/.style={
    draw,
    line width=0.9pt,
    fill=opmaroon,
    rounded corners=2pt,
    minimum size=8mm
  },
  bond/.style={line width=2.0pt}
]

\def\dx{1.55}
\def\dy{1.40}
\def\shear{0.}

\def\cA{2}
\def\rA{2}

\begin{scope}[xslant=\shear, yscale=0.98]

  \foreach \r in {0,1,2,3,4}{
    \foreach \c in {0,1,2,3,4}{
      \ifnum\c=\cA
        \ifnum\r=\rA
          \node[optensor] (T-\c-\r) at (\c*\dx, -\r*\dy) {};
        \else
          \node[tensor]   (T-\c-\r) at (\c*\dx, -\r*\dy) {};
        \fi
      \else
        \node[tensor]     (T-\c-\r) at (\c*\dx, -\r*\dy) {};
      \fi
    }
  }

  \foreach \r in {0,1,2,3,4}{
    \foreach \c in {0,1,2,3}{
      \pgfmathtruncatemacro{\cp}{\c+1}
      \draw[bond] (T-\c-\r.east) -- (T-\cp-\r.west);
    }
  }

  \foreach \c in {0,1,2,3,4}{
    \foreach \r in {0,1,2,3}{
      \pgfmathtruncatemacro{\rp}{\r+1}
      \draw[bond] (T-\c-\r.south) -- (T-\c-\rp.north);
    }
  }

\end{scope}
\end{tikzpicture}
\]
The expectation value is then a ratio of partition functions:
\begin{equation}\label{eq:ratio}
    \expval{O_A}_\psi = \frac{\expect{\psi}{O_A}{\psi}}{\braket{\psi}{\psi}} \equiv \frac{\ZZ^A}{\ZZ}
\end{equation}

Alternatively, we can express the partition function in terms of a derivative of the free energy:
\begin{equation}\label{eq:derivative}
    \expval{O_A}_\psi = -\frac{\partial}{\partial \lambda}\mathcal{F}_\lambda(O_A)\Big|_{\lambda=0}
\end{equation}
where we have defined the perturbed free energy $\mathcal{F}_\lambda( O_A) = -\log \expect{\psi}{\exp[\lambda O_A]}{\psi}$. As we will see below, these two viewpoints lead to two different expansions for expectation values.

\subsection{Cluster expansion for local observables}
We now use the technology of \autoref{prop:clusterexp} to write down cluster expansions for local expectation values. This is achieved in two ways, either by expanding both the numerator and denominator as in \autoref{eq:ratio} or through the derivative method of \autoref{eq:derivative}.

We start by expanding the numerator and denominator in \autoref{eq:ratio} using the same BP fixed point $\MM = \{\mu_{v\to w}\}_{(v,w)\in E}$ obtained from the norm network $\ZZ = \inner{\psi}{\psi}$. However, since $\ZZ^A=\expect{\psi}{O_A}{\psi}$ differs from $\ZZ$ within region $A$, the BP fixed point is no longer perfectly consistent there---the operator insertion breaks the self-consistency condition locally.  

Mathematically, we handle this by extending the loop set to include \emph{strings} which can terminate on $A.$ We substitute $\LL \to \LL_A$, consisting of connected ``stringy" subgraphs which are closed in the bulk and allowed to terminate in $A$. This captures the fact that the operator modifies the local environment. Clusters $\mathbf{W}$ can be defined similarly on the set of strings, and the notion of connectedness of a cluster holds as before. We denote evaluation of clusters on $\ZZ$ and $\ZZ^A$ as $Z_{\mathbf{W}}$ and $Z^A_{\mathbf{W}}$, respectively.

The cluster expansion then gives us an exact formula for the expectation value in terms of only those clusters that overlap with the observable region:

\begin{prop}[Local observable expansion, ratio version (Informal)] \label{prop:localexpexpansion}
    Local expectation values admit the following expansion in terms of clusters over $\LL_A$:
     \begin{equation} \label{eq:expvalexpansion}
    \expval{O_A} = \expval{O_A}_{\emph{BP}} \, \exp{\left(\sum_{\substack{\emph{conn} \, \mathbf{W} \leftarrow \LL_A \\ 
   \emph{supp}(\mathbf{W}) \cap A \neq \emptyset }} \phi_{\mathbf{W}} (Z^A_{\mathbf{W}} - Z_{\mathbf{W}})\right)}
\end{equation}
where the sum includes only connected clusters whose support intersects the region $A$. 
\end{prop}

Alternatively, we can leverage \autoref{eq:derivative} to obtain a cluster expansion in which clusters contribute \textit{additively}. Since the insertion of $\exp[\lambda O_A]$ only affects the tensors on $A$, it is manifest that only clusters which intersect $A$ will contribute to the derivative. Again, in the expansion for the tensor network $\ZZ^A_\lambda = \expect{\psi}{\exp[\lambda O_A]}{\psi}$, we use the $\lambda=0$ BP fixed point, denoting the corresponding cluster evaluations as ${Z_\lambda^A}_{\mathbf{W}}$. With this, we arrive at the additive expansion of local expectation values.
\begin{prop}[Local observable expansion, derivative version (Informal)] \label{prop:localexpexpansion-derivative}
Local expectation values admit the following expansion in terms of clusters over $\LL_A$:
\begin{equation}\label{eq:localobs-deriv}
    \expval{O_A} = \expval{O_A}_{\emph{BP}} + \sum_{\substack{\emph{conn} \, \mathbf{W} \leftarrow \LL_A \\ 
   \emph{supp}(\mathbf{W}) \cap A \neq \emptyset }} \phi_{\mathbf{W}} \partial_\lambda  \loopZ{\mathbf{W}}\Big|_{\lambda=0}.
\end{equation} 
\end{prop}

For a local region of $2\times 2$ sites, we illustrate a few strings involved in the expansion of \autoref{prop:localexpexpansion}:
\[
\begin{tikzpicture}[
  baseline={([yshift=-0.65ex] current bounding box.center)},
  tensor/.style={
    draw,
    line width=1pt,
    fill=tensorcolor!80!black,
    rounded corners=2pt,
    minimum size=8mm 
  },
  optensor/.style={
    draw,
    line width=1.2pt,
    fill=opmaroon,
    rounded corners=2pt,
    minimum size=8mm
  },
  bond/.style={line width=1.6pt, draw=black},
  redbond/.style={line width=2.4pt, draw=red} 
]

\def\dx{1.3}
\def\dy{1.3}
\def\shear{0.0}

\begin{scope}[xslant=\shear, yscale=0.98]

  \foreach \r in {0,...,5}{
    \foreach \c in {0,...,5}{
      \node[tensor] (T-\c-\r) at (\c*\dx, -\r*\dy) {};
    }
  }

  \node[optensor] (T-2-2) at (2*\dx, -2*\dy) {};
  \node[optensor] (T-3-2) at (3*\dx, -2*\dy) {};
  \node[optensor] (T-2-3) at (2*\dx, -3*\dy) {};
  \node[optensor] (T-3-3) at (3*\dx, -3*\dy) {};

  \foreach \r in {0,...,5}{
    \foreach \c in {0,...,4}{
      \pgfmathtruncatemacro{\cp}{\c+1}
      \draw[bond] (T-\c-\r.east) -- (T-\cp-\r.west);
    }
  }

  \foreach \c in {0,...,5}{
    \foreach \r in {0,...,4}{
      \pgfmathtruncatemacro{\rp}{\r+1}
      \draw[bond] (T-\c-\r.south) -- (T-\c-\rp.north);
    }
  }


  \draw[redbond] (T-2-2.north) -- (T-2-1.south);
  \draw[redbond] (T-2-1.west)  -- (T-1-1.east);
  \draw[redbond] (T-1-1.north) -- (T-1-0.south);
  \draw[redbond] (T-1-0.east)  -- (T-2-0.west);
  \draw[redbond] (T-2-0.south) -- (T-2-1.north);

  \draw[redbond] (T-2-3.south) -- (T-2-4.north);
  \draw[redbond] (T-2-4.west)  -- (T-1-4.east);
  \draw[redbond] (T-1-4.north) -- (T-1-3.south);
  \draw[redbond] (T-1-3.east)  -- (T-2-3.west);
  \draw[redbond] (T-1-4.south) -- (T-1-5.north);
  \draw[redbond] (T-1-5.west)  -- (T-0-5.east);
  \draw[redbond] (T-0-5.north) -- (T-0-4.south);
  \draw[redbond] (T-0-4.east)  -- (T-1-4.west);

  \draw[redbond] (T-3-2.east)  -- (T-4-2.west);
  \draw[redbond] (T-4-2.east)  -- (T-5-2.west);
  \draw[redbond] (T-5-2.south) -- (T-5-3.north);
  \draw[redbond] (T-5-3.west)  -- (T-4-3.east);
  \draw[redbond] (T-4-3.west)  -- (T-3-3.east);

\end{scope}
\end{tikzpicture}
\]
The lowest-weight clusters are exactly the set of such single strings, with non-linearity in the expansion \autoref{prop:localexpexpansion} formed by higher-weight connected combinations.
This result says that the exact expectation value equals the BP prediction, dressed by the cluster corrections that intersect the observable. Clusters not intersecting $A$ do not contribute---they cancel exactly between numerator and denominator, or vanish under the derivative. This locality enables the computational efficiency {of the expansion}.

\subsection{Convergence and algorithmic implications}

When do these expansions converge? The answer is controlled by the same decay condition that governed the free-energy expansion, on the set of excitations $\LL_A$. If loops (or strings terminating on $A$) decay exponentially, then clusters that touch $A$ but extend to distance $r$ are suppressed as $\sim e^{-cr}$, ensuring rapid convergence. The decay of strings in $\LL_A$ is parametrically consistent with the decay of loops in $\LL$, since decay is a bulk property, unchanged by a changing tensors in a local region $A$.\footnote{Technically, we can ensure convergence by demanding decay of strings in $\LL_A$ on both $\ZZ$ and $\ZZ^A$. While typically consistent with the decay of loops in $\LL$ on $\ZZ$, it is not strictly guaranteed; see~\cite{SM} for details.}

Truncating the series at finite order yields an efficient algorithm with rigorous error bounds:

\begin{theorem}[Estimating local observables (Informal)]\label{thm:localobsalgorithm}
    Given decay of strings on $\LL_A$, truncating the cluster expansion for a local observable $O$ {with $\|O\|=1$} at order $m$ leads to:
    \begin{enumerate}
        \item[(i)] Relative error in the ratio expansion \autoref{eq:expvalexpansion},
        \begin{equation}
            \left|\frac{\expval{O} - \expval{O}_m}{\expval{O}}\right|  \leq  \order{|A| e^{-d(m+1)}}
        \end{equation}
        \item[(ii)] Additive error in the derivative expansion \autoref{eq:localobs-deriv},
            \begin{equation}
        |\expval{O}-\expval{O}_m| \leq \order{|A| m e^{-d(m+1)}} 
    \end{equation}
    \end{enumerate}

    where $d=c-c_0=\order{1}$ is a constant.
\end{theorem}

The error decays \emph{exponentially} in the cluster order $m$. This means that in gapped phases with fast loop decay, even low-order corrections (small $m$) can achieve high accuracy. Conversely, near criticality where loop decay weakens, higher-order corrections become necessary---or the expansion may fail to converge altogether, signaling the breakdown of the BP approximation.

\section{Correlation functions}
\label{sec:correlation_fns}

We now turn to the question of the physical interpretation of the ``loop tensors." What do they tell us about the quantum state itself? The key insight is that loop corrections are intimately related to the spatial decay of connected correlations in the state. The decay of loops (and strings on $\LL_{AB}$, see below) condition is sufficient to show exponential decay of connected correlations with a finite correlation length.

\subsection{Intuition: From MPS to PEPS}
To understand this connection, it is helpful to first recall how correlations arise in matrix product states. Consider a translation-invariant MPS:

\begin{align}
   \ket{\Psi} =  \dotso \quad
\begin{tikzpicture}[scale=0.5,baseline={([yshift=-0.65ex] current bounding box.center) }]
        \GTensor{0,0}{1.2}{.6}{\small $V$}{9};
      \GTensor{1*\singledx,0}{1.2}{.6}{\small $V$}{9};
      \GTensor{2*\singledx,0}{1.2}{.6}{\small $V$}{9};
      \GTensor{3*\singledx,0}{1.2}{.6}{\small $V$}{9};
      \GTensor{4*\singledx,0}{1.2}{.6}{\small $V$}{9};
\end{tikzpicture}
    \quad \dotso
\end{align}
where each $V$ is a $D \times D \times d$ tensor, with $d$ the on-site physical Hilbert space dimension. To this MPS one associates the ``transfer matrix'' $T_V$:
\begin{align}
  T_V := \quad  
\begin{tikzpicture}[scale=0.5,baseline={([yshift=-0.65ex] current bounding box.center) }]
      \GTensor{0,0}{1.2}{.6}{\small $V$}{9};
      \GTensor{0,-2}{1.2}{.6}{\small $\bar{V}$}{6};
    \end{tikzpicture}
    \quad  
\end{align} 

In one dimension, correlations between two distant regions $A$ and $B$ are mediated by a \emph{single path} through the intervening sites. Normalizing the leading eigenvalue of $T_V$ to unity, connected correlations on this MPS $\expval{O_AO_B}_c = \expval{O_AO_B} - \expval{O_A}\expval{O_B}$ decay as $\sim \lambda^{d(A,B)}$, where $d(A,B)$ is the distance between $A$ and $B$, and $\lambda < 1$ is the sub-leading eigenvalue of $T_V$, corresponding to a spectral gap of $(1-\lambda)$. Thus, a gapped transfer matrix in 1D is sufficient to guarantee exponentially decaying correlations with a finite correlation length.

\subsection{The challenge on general graphs}
For a PEPS on a general graph $G= (V,E)$, the situation is fundamentally different. Now there are \emph{multiple paths} connecting any two regions $A$ and $B$, owing to the existence of loops in the graph. Each path contributes to the correlation, and these contributions must be summed coherently. To perform this sum, the set of loop excitations is similarly promoted to $\LL \to \LL_{AB}$ where strings in $\LL_{AB}$ are closed in the bulk but may terminate in $A$ and $B$

The excitation tensors $Z_l$ for $l\in \LL_{AB}$ precisely quantify these multi-path contributions. A cluster $\mathbf{W}$ connecting $A$ and $B$ represents correlations mediated through its support. If the excitations decay sufficiently fast---meaning the effective coupling along alternative paths is strongly suppressed---then distant correlations must also decay. 
Conversely, if loop corrections remain large, it signals that multiple interfering paths contribute significantly, potentially leading to long-range correlations or even critical behavior. 

This intuition can be made rigorous using the cluster expansion framework. Applying \autoref{prop:localexpexpansion} to the individual terms in the connected correlator, we obtain an exact expression that makes the role of loops explicit.
\begin{prop}[Correlator Cluster Expansion, Ratio Version (Informal)]\label{prop:correlatorexpansion}
    Connected correlation functions of the form $\expval{O_AO_B}_c = \expval{O_AO_B} - \expval{O_A}\expval{O_B}$ admit the following cluster expansion, 
    \begin{equation}\label{eq:correlatorexpansion}
\expval{O_AO_B}_c = 
\expval{O_A} \expval{O_B}
\left[
\exp\!\Biggl\{
\sum_{\substack{\emph{conn.}\, \mathbf{W} \leftarrow \LL_{AB} \\
   \emph{supp}(\mathbf{W}) \cap A \neq \emptyset
    \\
   \emph{supp}(\mathbf{W}) \cap B \neq \emptyset} }
\phi_\mathbf{W} Z_\mathbf{W} ^{\emph{conn}}
\Biggr\}
- 1 \right]
\end{equation}
where 
\begin{equation}
   Z_\mathbf{W}^{\emph{conn}} =  Z^{AB}_\mathbf{W} + Z_\mathbf{W} - Z^{A}_\mathbf{W} - Z^{B}_\mathbf{W}.
\end{equation}
\end{prop}

While the expectation values $\expval{O_{A(B)}}$ in \autoref{prop:correlatorexpansion} themselves involve an expansion of the form \autoref{prop:localexpexpansion}, we can in fact wrap everything into a single expansion by expressing the connected correlation function via derivatives of a perturbed free energy:
\begin{equation}\label{eq:derivative2}
    \expval{O_A O_B}_c = -\frac{\partial}{\partial \lambda_a}\frac{\partial}{\partial \lambda_b} \mathcal{F}_{\lambda_a,\lambda_b}(O_A, O_B)\Big|_{\lambda_a=\lambda_b=0},
\end{equation}
where the perturbed free energy is now $\mathcal{F}_{\lambda_a,\lambda_b}(O_A,O_B) = -\log \bra{\psi} \exp[\lambda_a O_A + \lambda_b O_B] \ket{\psi}$. Denoting $\bm{\lambda} = (\lambda_a, \lambda_b)$ and $\partial_{\bm \lambda} := \partial_{\lambda_a} \partial_{\lambda_b}$ leads to the following expansion on the network $\ZZ^{AB}_{\bm{\lambda}}$, with cluster evaluation over $\LL_{AB}$ denoted $Z^{AB}_{\mathbf{W},\mathbf{\lambda}}$.

\begin{prop}[Correlator Cluster Expansion, Derivative Version (Informal)]\label{prop:conn-correlator}
Connected correlation functions of the form $\expval{O_AO_B}_c$ admit the following additive cluster expansion, 
\begin{equation}
   \expval{O_AO_B}_c = \sum_{\substack{\emph{conn.}\, \mathbf{W} \leftarrow \LL_{AB} \\
   \emph{supp}(\mathbf{W}) \cap A \neq \emptyset \\ 
    \emph{supp}(\mathbf{W}) \cap B \neq \emptyset }} \phi_{\mathbf{W}} \partial_{\bm \lambda} Z^{AB}_{\mathbf{W},\mathbf{\lambda}}\Big|_{\bm{\lambda}=0}.
\end{equation}
\end{prop}
This is straightforwardly generalized to the case of $p-$point connected correlation functions, as outlined in \cite{SM}.

Crucially, the expansions in \autoref{prop:correlatorexpansion} and \autoref{prop:conn-correlator} are over all connected clusters intersecting both $A$ and $B$, which requires that the clusters are of size at least $d(A,B)$. For a square lattice with $A$ and $B$ being local regions of $2\times 2$ sites, we illustrate a few strings connecting $A$ and $B$ which form the low-weight clusters in the expansion:

\[
\begin{tikzpicture}[
  baseline={([yshift=-0.65ex] current bounding box.center)},
  tensor/.style={
    draw,
    line width=1pt,
    fill=tensorcolor!80!black,
    rounded corners=2pt,
    minimum size=8mm 
  },
  optensor/.style={
    draw,
    line width=1.2pt,
    fill=opmaroon,
    rounded corners=2pt,
    minimum size=8mm
  },
  bond/.style={line width=1.6pt, draw=black},
  redbond/.style={line width=2.4pt, draw=red} 
]

\def\dx{1.3}
\def\dy{1.3}
\def\shear{0.0}

\begin{scope}[xslant=\shear, yscale=0.98]

  \foreach \r in {0,...,5}{
    \foreach \c in {0,...,5}{
      \node[tensor] (T-\c-\r) at (\c*\dx, -\r*\dy) {};
    }
  }

  \node[optensor] (T-0-0) at (0*\dx, -0*\dy) {};
  \node[optensor] (T-1-0) at (1*\dx, -0*\dy) {};
  \node[optensor] (T-0-1) at (0*\dx, -1*\dy) {};
  \node[optensor] (T-1-1) at (1*\dx, -1*\dy) {};

  \node[optensor] (T-4-4) at (4*\dx, -4*\dy) {};
  \node[optensor] (T-5-4) at (5*\dx, -4*\dy) {};
  \node[optensor] (T-4-5) at (4*\dx, -5*\dy) {};
  \node[optensor] (T-5-5) at (5*\dx, -5*\dy) {};

  \foreach \r in {0,...,5}{
    \foreach \c in {0,...,4}{
      \pgfmathtruncatemacro{\cp}{\c+1}
      \draw[bond] (T-\c-\r.east) -- (T-\cp-\r.west);
    }
  }

  \foreach \c in {0,...,5}{
    \foreach \r in {0,...,4}{
      \pgfmathtruncatemacro{\rp}{\r+1}
      \draw[bond] (T-\c-\r.south) -- (T-\c-\rp.north);
    }
  }


  \foreach \c in {1,...,4}{
      \pgfmathtruncatemacro{\cp}{\c+1}
      \draw[redbond] (T-\c-0.east) -- (T-\cp-0.west);
  }
  \foreach \r in {0,...,3}{
      \pgfmathtruncatemacro{\rp}{\r+1}
      \draw[redbond] (T-5-\r.south) -- (T-5-\rp.north);
  }

  \draw[redbond] (T-0-1.south) -- (T-0-2.north);
  \draw[redbond] (T-0-2.south) -- (T-0-3.north);
  \draw[redbond] (T-0-3.east)  -- (T-1-3.west);
  \draw[redbond] (T-1-3.east)  -- (T-2-3.west);

  \draw[redbond] (T-2-3.north) -- (T-2-2.south);
  \draw[redbond] (T-2-2.east)  -- (T-3-2.west);
  \draw[redbond] (T-3-2.south) -- (T-3-3.north);
  \draw[redbond] (T-3-3.west)  -- (T-2-3.east);

  \draw[redbond] (T-2-3.south) -- (T-2-4.north);
  \draw[redbond] (T-2-4.east)  -- (T-3-4.west);
  \draw[redbond] (T-3-4.east)  -- (T-4-4.west);

\end{scope}
\end{tikzpicture}
\]

One can also have clusters $\mathbf{W}$ intersecting both regions $A$ and $B$ where no single loop intersects both the regions. The simplest example is $\mathbf{W} = \{(l_A,1), (l_B,1)\}$ where $l_{A(B)}$ intersects $A(B)$ and the two loops intersect in the bulk. 

\subsection{From loop decay to clustering of correlations}

\begin{figure*}[t]
    \centering
    \includegraphics[width=\linewidth]{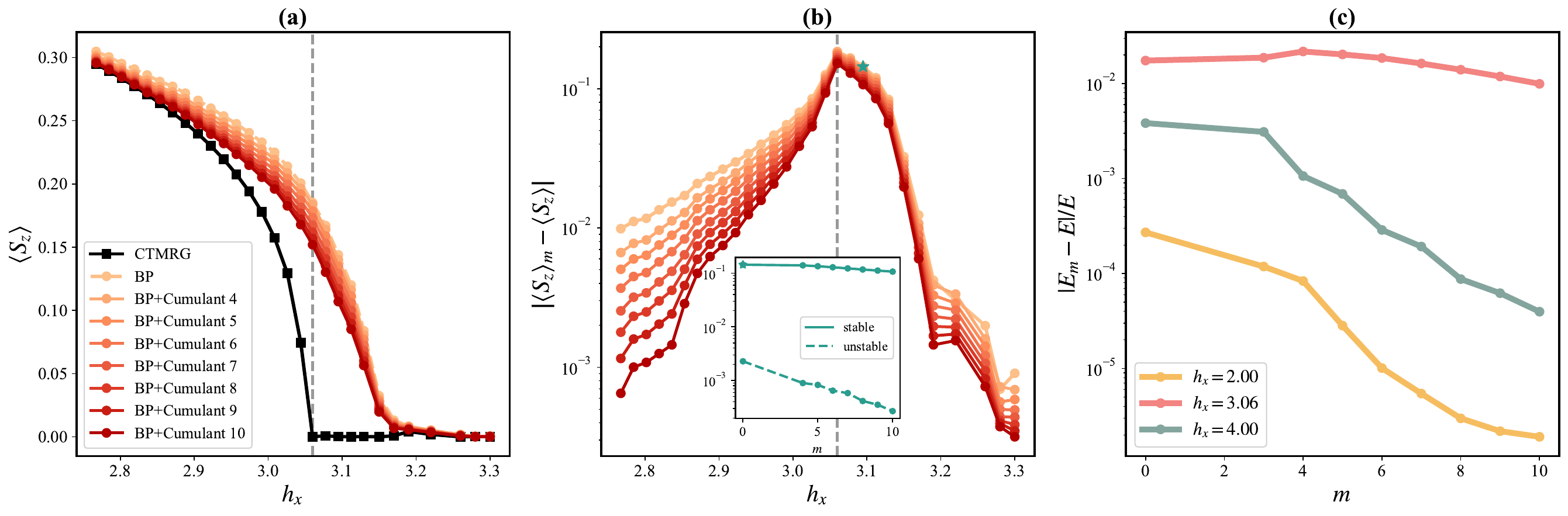}
    \caption{Tensor network belief propagation on the ground state of the 2D TFIM obtained via CTMRG-based optimization. (a) Longitudinal magnetization computed via CTMRG and BP with cumulant expansion up to cumulant order $m=10$ as a function of the transverse field $h_x$ around the critical region. The dashed line at $h_x \approx 3.06$ marks the point the CTMRG magnetization goes to zero. (b) Error in the $m^{\text{th}}$-order cumulant-corrected value (BP is $m=0$) as compared to the ground truth CTMRG magnetization. Inset: comparison of the error as a function of $m$ at $h_x=3.1$ for an expansion around a stable message-passing fixed point (solid curve) and unstable fixed point (dashed curve). The star at $m=0$ on the solid curve corresponds to the star in the main panel. (c) Relative error of the energy as a function of cumulant order $m$ deep in the ferromagnetic phase at $h_x=2.0$, near the critical point $h_x = 3.06$, and deep in the paramagnetic phase at $h_x=4.0$.}
    \label{fig:figTFIM_T=0}
\end{figure*}

\autoref{prop:correlatorexpansion} and~\autoref{prop:conn-correlator} reveal that connected correlations are carried entirely by clusters that connect the two regions $A$ and $B$. The key question is: what determines whether these correlations decay with distance?

In one dimension or on tree graphs, the answer is straightforward. A gapped transfer matrix with spectral gap $\delta<1$ implies exponential decay: correlators at graph distance $r$ scale as $\sim \delta^{\,r}$, yielding a finite correlation length
\[
\xi = (-\log \delta)^{-1}.
\]
This is because there is a \emph{unique path} connecting distant regions, and the gap controls the suppression along this path.

On general graphs, however, the situation is more subtle. The existence of a transfer-matrix gap alone does not suffice to guarantee decay of correlations. The obstruction is graph combinatorics: 
correlations between distant regions $A$ and $B$ are now mediated by an exponentially large family of connected clusters linking them. According to~\autoref{eq:correlatorexpansion}, each cluster $\mathbf{W}$ contributes with weight $Z_\mathbf{W}^{\text{conn}}$. While each individual cluster may be suppressed by the transfer-matrix gap, the \emph{number} of such clusters grows rapidly with distance, and this entropic proliferation can overwhelm the energetic suppression. 

To see this explicitly, consider clusters of size $|\mathbf{W}| = m$ connecting $A$ and $B$. If each such cluster contributes $\sim e^{-a m}$ (energetic suppression), but there are $\sim e^{b m}$ such clusters (entropic growth), the total contribution scales as $\sim e^{(b-a)m}$. Correlations will decay only if $b < a$---that is, only if the energetic suppression dominates the combinatorial growth. 

The loop-decay condition precisely ensures this dominance. This is why $c-$decay with an arbitrary $c>0$ is not enough to control correlations on the TN. Instead, to show decay of correlations, cluster contributions must decay faster than the entropy of connecting paths grows, 
{which is guaranteed when $c>c_0$}. 
This restores a finite correlation length despite the presence of loops. 
A comparison of the classical Ising model in 1D vs. 2D presents a familiar example of this phenomenon. At any temperature and in any dimension, BP has a ``paramagnetic'' fixed point, in which the $\mathbb{Z}_2$ symmetry is unbroken. In 1D, this fixed point is unique, and correlations decay according to the spectral gap $\delta = \tanh(\beta)$: the system is disordered at any finite temperature. This same gap appears in the cluster expansion around the paramagnetic fixed point in higher dimensions, with loops decaying as $(\tanh\beta)^{|l|}$. However, the expansion only converges at high temperatures; indeed, in 2D, this cluster expansion is equivalent to the standard high-temperature expansion, which converges down to $T_c$~\cite{Friedli_Velenik_2017}.

We thus arrive at the following theorem:
\begin{theorem}[Decay of correlations (Informal)] \label{thm:correlatorbound}
    The decay of loops condition on $\LL_{AB}$ ensures that all connected correlators cluster with a finite correlation length. That is given $O_{A(B)}$ with $\emph{supp}(O_{A(B)}) \subseteq A(B)$ and $\|O_{A(B)}\| = 1$, we have
    \begin{equation}
        \expval{O_AO_B}_c \leq \order{e^{-d(A,B) / \xi}}
    \end{equation}
    with a correlation length, 
    \begin{equation}
        \xi \leq \order{\frac{1}{c-c_0}},
    \end{equation}
    where $d(A,B)$ is the graph distance between disjoint local regions $A$ and $B$.
\end{theorem}

This result rules out the validity of belief-propagation–based expansions for states exhibiting sub-exponential decay of correlations, such as those at criticality. Notably, the argument is entirely graph-theoretic and does not rely on any lattice embedding. It establishes that, on tensor networks defined over arbitrary constant-degree graphs, exponential decay of correlations follows directly from string decay on $\LL_{AB}$. Back to the MPS example, therein the relevant set of excitations has a single string terminating at both ends, and we recover the usual transfer matrix picture.

\section{Numerical illustrations}
\label{sec:numerics}
In this section, we study numerically the applicability of the BP-based tensor network expansions. We use the corner-transfer matrix renormalization group (CTMRG) method~\cite{nishino1996corner,orus2012exploring} as the ``ground truth" to study systematically the error scaling of BP and subsequent cumulant corrections. For concreteness, we focus on the 2D transverse field Ising model, described by the Hamiltonian, 
\begin{equation}
    H = - \left(\sum_{\langle ij\rangle } S_z^i S_z^j + h_x \sum_i S_x^i\right)
\end{equation}
where $\langle ij\rangle $ denotes nearest-neighbors, and $S_\alpha^i$ are Pauli operators at site $i$. 
\subsection{The Ground State}
We begin by illustrating the performance of the BP and cumulant expansion methods on the ground state of the 2D TFIM. We start with an iPEPS representation of the ground state, obtained through CTMRG based variational optimization in the \texttt{peps-torch}~\cite{hasik2020peps} library with bond dimension $D=4$.

Given the iPEPS tensor for a given value of the transverse-field $T(h_x)$, we begin by finding the ``message-passing'' fixed point, $\mu_0(h_x)$ by repeated iteration of the fixed-point equation until convergence starting from the uniform vector.
This enables the computation of the BP approximation to local expectation values $\expval{O_A}_{\text{BP}}$. Then, we have the cumulant-corrected values at order $m$ as, 
\begin{equation}   
    \expval{O_A}_m =\expval{O_A}_{\text{BP}} \exp{\left(\sum_{\substack{\text{conn.} \Gamma \subseteq\LL_A \\ \mathrm{supp}(\Gamma) \cap A \neq \emptyset \\ |\Gamma| \leq m}} (\mathfrak K^A(\Gamma) - \mathfrak K(\Gamma))\right)}
\end{equation}
where we truncate to cumulants $\Gamma$ containing $A$ up to order $m$. We concern ourselves with $A$ being a single site (for longitudinal magnetization) and two sites (for energy). We compare BP and cumulant corrections to expectation values computed via CTMRG with a bond dimension of $\chi = 256$, which serves as our ``ground truth.''

The results using the ``message-passing'' fixed point $\mu_0(h_x)$ are shown in Fig.~\ref{fig:figTFIM_T=0}. In panel (a), we show the magnetization $\expval{S_z}$ computed through CTMRG, BP, and subsequent cumulant corrections. We note qualitative agreement between the different curves in either phase, for $h_x \lesssim 2.8$ and $h_x \gtrsim 3.2$ with a systematic disagreement in the vicinity of the critical point $h_c \approx 3.044$. In panel (b), we show the errors in different cluster orders against the CTMRG value $\Big|\expval{S_z}_m - \expval{S_z}\Big|$. This shows a clear peak near the critical point. Approaching from the ferromagnetic side, the slowing down of convergence is precisely a manifestation of critical correlations in the loops: the error between a finite cumulant correction and the ground truth encodes all loop corrections above that order, which shows parametrically slower decay leading to a maximum in the error. To reach the same conclusion when approaching from the paramagnetic side, one has to be more careful, since there exists an intermediate regime where the message-passing fixed point is not a good starting point for the expansion.  
We discuss this momentarily.


\autoref{thm:localobsalgorithm} bounds a relative error for local observables with BP based expansions, given decay of loops, which is concomitant with being in a gapped phase. To illustrate this, we show in Fig.~\ref{fig:figTFIM_T=0}(c) the relative error in approximating the ground state energy per site as a function of cumulant order $m$. Deep in the paramagnetic phase at $h_x = 4.0$, and the ferromagnetic phase at $h_x=2.0$, one notes exponential convergence in relative error. Near the critical point  the convergence of the curve is markedly slower, as expected from the fact that loop corrections encode correlation functions.

We now clarify a point regarding the accuracy of BP on the ground state iPEPS. In the range of $h_x \in (3.06, 3.2)$ we note that the CTMRG magnetization goes to zero, whereas the BP magnetization does not. Relying on the message-passing fixed point $\mu_0(h_x)$ leads to a poor accuracy in this part of paramagnetic phase. This is perhaps surprising since one might expect BP to work away from a critical point, especially in a short-ranged correlated phase. However, the landscape of fixed points is governed by the model on a tree, whose phase diagram (as discussed in \autoref{sect:confusion}) differs from that of the loopy model. Owing to this, decay of loops is indeed violated around the message-passing fixed point in this ``confusion regime," since there exists a constant $\order{1}$ contribution in the ``all-edges-excited" sector occupied by the other symmetry broken fixed point.

The remedy is to note that there exist \emph{more fixed points}, which indeed respect the symmetry and will provide vanishing magnetization. Of these fixed points, one is the paramagnetic fixed point; expanding around this fixed point in the interval $h_x \in (3.06, 3.2)$ indeed yields much better accuracies in BP and subsequent cumulant/cluster expansions. To highlight this, we show in the inset of \autoref{fig:figTFIM_T=0}(b) the error of the cumulant expansion as a function of cumulant order for both the stable and unstable fixed point at $h_x=3.1$. We indeed note that, the error of BP around the unstable fixed point is two orders of magnitude lower, and is further accompanied by a faster decay of the error.
\subsection{Finite Temperature: Gibbs States}\label{sect:finiteT}

\begin{figure}[t]
\centering
\includegraphics[width=0.9\linewidth]{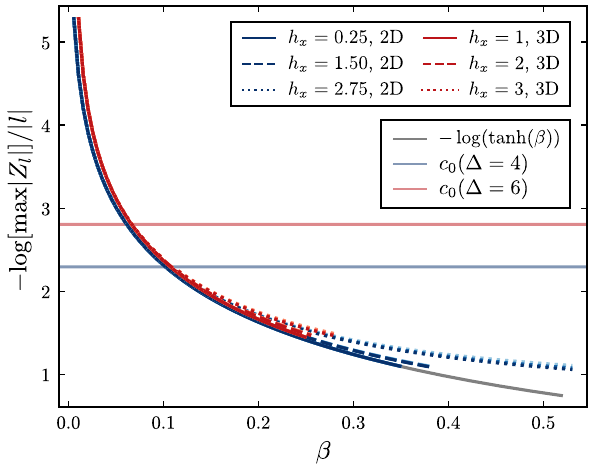}
\caption{$c$-decay of even-weight loops across a range of transverse fields in the 2D (blue) and 3D (red) TFIM {at finite temperature}. Data from increasing loop weights up to $|l|=12$ (2D) and $|l|=10$ (3D) (darker shading indicates higher weight) nearly coincide. Solid gray curve is the coefficient of classical loop decay. Dashed blue and red curves indicate the loose bound for $c_0 = \mathcal{O}(\log \Delta)$ for $\Delta=4$ (2D) and $\Delta = 6$ (3D) respectively. \label{fig:loops}} 
\end{figure}

\begin{figure*}[t]
\includegraphics[width=\linewidth]{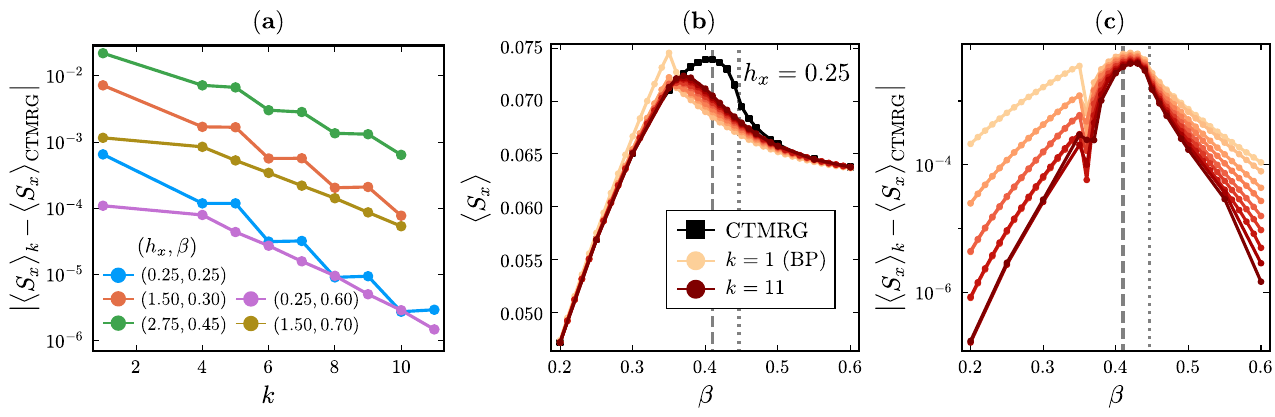}
\caption{(a) Convergence of region-based cumulant cluster expansion for the observable $S_x$, as a function of the maximum region size $k$, for five different pairs of $(h_x,\beta)$ in the 2D TFIM. Blue, orange, and green curves are above the critical temperature but below the temperature at which $c=c_0$; purple and gold curves are below the critical temperatures at their respective transverse fields.  (b) $\expval{S_x}$ as a function of $\beta$ at $h_x=0.25$, evaluated using CTMRG (black squares), vs. increasing orders of the region-based cumulant cluster expansion (ranging from $k=1$ in light orange, to $k=11$ in dark red). (c) Additive error in the cumulant cluster expansion for $\expval{S_x}$ at $h_x=0.25$, using CTMRG as the ``ground truth.'' In (b) and (c), the dashed gray line marks the inverse temperature at which $\expval{S_x}_{\rm CTMRG}$ has a peak, which is slightly below $\beta_c \approx 0.446$ (dotted line)~\cite{roughening}.  \label{fig:finite-temp}}
\end{figure*}

Now we consider the BP landscape for thermal Gibbs states, $\rho \propto \exp(-\beta H)$, which are known to be efficiently approximable by PEPS \cite{molnar2014approximating}. For translation-invariant Hamiltonians, the purified Gibbs state (or, the ``thermo-field double") can be represented as iPEPS. Explicitly, we obtain the mixed state $\sqrt{\rho}$ via imaginary time evolution starting from the maximally-mixed (infinite-temperature) state. Viewing $\sqrt{\rho}$ as a pure state $|\sqrt{\rho}\rrangle $ in a doubled Hilbert space, the norm of the associated tensor network is then $\Tr \rho$, and expectation values are evaluated as $
\expval{O} = {\llangle \sqrt{\rho} | O | \sqrt{\rho}\rrangle}/{\llangle \sqrt{\rho} | \sqrt{\rho}\rrangle}
$.

Let us first consider high temperatures. Prior work has shown that Gibbs states of local Hamiltonians at high temperatures are unentangled~\cite{bakshi2024high} and show decay of correlations~\cite{bluhm2022exponential,kuwahara2020clustering}, and thus we might hope that this manifests as a strong decay of loops. Note, however, that clustering of correlations does not necessarily imply loop decay; only the converse (\autoref{thm:correlatorbound}) is guaranteed to hold. The obstruction is two-fold: first, the loop decay property depends not just on the structure of the state but also on the chosen BP fixed point; second, even if a reasonable BP fixed point is found, a finite correlation length could result from the (e.g., fine-tuned) destructive interference of loops of opposite signs rather than uniform loop decay. 

Nevertheless, high-temperature Gibbs states present an example where we can show loop decay around the ``high-temperature'' BP fixed point. The ``proof'' is straightforward: there is a fixed point $\mu(\beta)$, smoothly connected to the infinite-temperature fixed point, satisfying $\|\mu(\beta) - \mu(0)\|_\infty \leq \mathcal{O}(\beta)$ where $\mu(0)$ is the infinite-temperature fixed point. Each insertion of a projector onto the excitation subspace thus carries a factor of $\mathcal{O}(\beta)$, ensuring loop decay with $c=\mathcal{O}(-\log(\beta))$. Hence, we can ensure $c > c_0$ by requiring $\beta < \order{e^{-c_0}}$, demonstrating that Gibbs states satisfy loop decay 
{around an appropriate fixed point} below a constant $\beta_0 = \order{e^{-c_0}}$. For the transverse field Ising model in both 2D and 3D, we can refine this estimate: we numerically observe that the classical ``gap" $\tanh\beta$ is an upper bound on the weight of loops in the presence of a transverse field (see \autoref{fig:loops}).

Concretely, we prepare the iPEPS with bond dimension $D=8$ via simple update with the \texttt{pepskit.jl} package (2D) and \texttt{TensorNetworkQuantumSimulator.jl} (3D). For each prepared state, we iterate the message-passing equations starting from uniform messages until approximate convergence. We then populate a finite periodic lattice of linear dimension $l_{\max}+1$ with the PEPS tensor, and evaluate all loops up to weight $l_{\max}$. In~\autoref{fig:loops}, we plot
\begin{equation}
    -\log[\max|Z_l|]/|l|
\end{equation}
for even $|l|$, where the maximum is taken over all loops containing $|l|$ edges.

Recall that in the cluster expansion around the paramagnetic fixed point of the \textit{classical} 2D Ising model, the weight of any loop containing a vertex of odd degree is identically zero; thus, the entropic contribution of the clusters is significantly reduced from the worst-case estimate. As a result, there is a wide gap between the rate of $c$-decay at the critical temperature, corresponding, $c=\sinh^{-1}(1) = 0.881...$, and the worst-case threshold $c_0 = \log(2(\Delta-1))+1/2 =2.29...$ for models on the square lattice. When we turn on a weak transverse field, odd-degree loops (including all those with odd $l$) acquire a nonzero weight, but remain parametrically suppressed compared to even-degree loops. We numerically observe a strong suppression even in the strongly quantum regime ($h_x=3$)~\cite{SM}. Consequently, we expect the cluster expansion to converge well beyond the regime guaranteed by our loose bound.

To verify this claim, we carry out the cumulant cluster expansion for single-site observables. For the sake of demonstration, we here use the derivative form of the top-down version of the cumulant cluster expansion, which in the ``regions-based'' formulation simplifies to~\cite{gray2025,SM}
\begin{equation}   
    \expval{O_A}_k =\sum_{R} b_k(R) \expval{O_A}_R,
\end{equation}
where $\expval{O_A}_R$ is the expectation value obtained from a contraction on region $R$ with BP messages inserted along its boundary, and $b_k(R)$ is the counting number of region $R$ in the expansion truncated at maximum region size $k$. As shown in~\autoref{fig:finite-temp}a, the additive error in $\expval{S_x}$ (using CTMRG with $\chi=128$ as the ground truth) rapidly decays in both phases, including at temperatures within the paramagnetic phase but beyond the regime guaranteed by the worst-case bound. At high temperatures, the disparity between odd and even loops  loops manifests in an odd-even effect in the additive error, whereas at low temperatures, the error decays smoothly as a function of $k$.      

Again, we caution that within the BP approximation, 
the paramagnetic fixed point becomes unstable at a temperature above the true critical temperature
, resulting in a ``confusion regime'' in which CTMRG correctly describes the paramagnetic phase whereas the message passing algorithm instead converges to a symmetry-broken fixed point.~\autoref{fig:finite-temp}b shows this phenomenon for the cumulant cluster expansion of $\expval{S_x}$ at $h_x=0.25$. Strictly speaking, an expansion around this symmetry-broken fixed point in the confusion regime should not converge, and in the interval $T_{c,BP} > T > T_{BP}$, better accuracy could be achieved by expanding around the unstable paramagnetic fixed point.  Still, as for the ground state phase diagram, we can draw two conclusions from~\autoref{fig:finite-temp}c: (1) the reduction of error with respect to CTMRG as a function of region size $k$  is significantly slower near the critical point than deep in the two phases, and (2) at a given $k$, the error is maximal in the vicinity of $T_c$ (\autoref{fig:finite-temp}c).

\section{Discussions} \label{sec:discussions}
This work examines the efficiency and fundamental limitations of tensor-network belief-propagation methods for many-body quantum systems. Building on recent advances in cluster expansion techniques for belief propagation, we provide rigorous criteria for when these methods can and cannot succeed. This is achieved through the ``loop-decay" condition introduced in Ref.~\cite{midha2025beyond}. Focusing on the task of approximating local expectation values, we provide exact formulas in terms of the BP prediction dressed by connected clusters intersecting the observable region. We further compare two distinct forms of the cluster expansion used in the literature, whose truncation leads to algorithms with rigorous guarantees on relative and additive error of the approximation. Using this methodology to estimate connected correlation functions, we prove that ``loop decay'' (with the ``loop'' set expanded to include strings terminating on the inserted observables) necessarily implies decay of correlations, leading to a sharp criteria on the validity of BP based methods. More generally, this provides new language to reason about tensor networks in higher dimensions and on arbitrary graphs, which may find applications elsewhere, e.g., in inference problems, and computational complexity theory. 

We validate our analytical predictions through extensive numerical simulations of the 2D transverse-field Ising model, both at finite temperature and in the ground state. We observe a pronounced slowdown in the convergence of the cluster expansion near criticality. In contrast, at high temperatures we find that loop weights exhibit a consistent decay consistent with a general ansatz, in both two and three dimensions, which we expect to hold broadly for Gibbs states of local Hamiltonians.

The fixed-point problem was also highlighted, wherein we note the crucial dependence of BP-based expansions on the chosen fixed-points, as also addressed in recent work~\cite{evenbly2025partitioned}. In general, the fixed-point problem remains open: it is not known for what classes of TNs fixed point exists, and whether they can be found efficiently (e.g., guarantees on convergence of message passing on TNs). We conjecture that for some hard to contract TN instances, the fixed points can exist and (i) either be hard to find, or (ii) severely violate decay of loops; these issues arise even in the tensor networks corresponding to partition functions of classical models, as we explore in upcoming work~\cite{tindallgbp}. More broadly, studying the complexity of tensor network contraction through this lens of BP and cluster expansions may lead to new insights in complexity theory. In another upcoming work~\cite{midha26upcoming}, we provide rigorous results for the fixed-point problem for a sub-class of tensor network states.
\newline 

\begin{acknowledgments}
We thank Yifan Frank Zhang, Dries Sels, Bert Kappen, Subhayan Sahu, and Juraj Hasik for helpful discussions and collaborations on related projects.  
S.M. acknowledges helpful discussions at the International Centre for Theoretical Sciences (ICTS) Bangalore during the program ``Generalised Symmetries and Anomalies in Quantum Phases of Matter 2026" (code: ICTS/GSYQM2026/01). We acknowledge the use of open source libraries \texttt{peps-torch}, \texttt{pepskit}, and \texttt{TensorNetworkQuantumSimulator} for the numerical simulations. The simulations presented in this article were performed in part on computational resources managed and supported by Princeton Research Computing, a consortium of groups including the Princeton Institute for Computational Science and Engineering (PICSciE) and the Office of Information Technology's High Performance Computing Center and Visualization Laboratory at Princeton University. J.T. and G.M.S. are grateful for ongoing support through the Flatiron Institute, a division of the Simons Foundation. This work was partially supported by a Brown Investigator Award (S.M. and D.A.). 
\end{acknowledgments}
\bibliography{refs}


\clearpage
\onecolumngrid
\makeatletter
\let\@outputdblcol\@outputpage
\makeatother
\setcounter{section}{0}
\setcounter{subsection}{0}
\setcounter{subsubsection}{0}
\setcounter{equation}{0}
\setcounter{figure}{0}
\setcounter{table}{0}
\setcounter{footnote}{0}
\renewcommand{\thesection}{S\arabic{section}}
\renewcommand{\theHsection}{S\arabic{section}}
\renewcommand{\thesubsection}{\thesection.\arabic{subsection}}
\renewcommand{\theequation}{S\arabic{equation}}
\renewcommand{\thefigure}{S\arabic{figure}}
\renewcommand{\thetable}{S\arabic{table}}
\renewcommand{\theHequation}{S\arabic{equation}}

\section*{Supplementary Material}
\listofsmentries

\smsection{Recap: Loop and Cluster expansion}
\smsubsection{Preliminaries}
Given a graph $G = (V,E)$, we associate to each edge $e \in E$ a \textit{bond space} $\BB_e$ of dimension $\dim(\BB_e) = D$, where $D$ is called the bond dimension, chosen to be uniform for convenience. For each vertex $v \in V$, We denote the neighbors of $v$ in $G$ as $\NN(v)$. The degree of a vertex in the graph is denoted $d(v) := |\NN(v)|$. The degree of a graph is denoted $\Delta(G) := \max_{v\in V}d(v)$. With this, we define the notion of a tensor network.
\begin{defn}[Tensor Network]
    A graph $G=(V,E)$ equipped with a set of tensors at each vertex, $\TT = \{T_v\}_v$ with  $T_v \in \otimes_{e \ni v}\BB_e$ is called a \emph{tensor network}.
\end{defn}
Now, the BP algorithm proceeds by associating two messages at each edge, i.e., for $e = (v,w)$ we have $\mu_{v\to w} , \mu_{w\to v} \in \BB_e$. The messages are said to be self-consistent if they satisfy the following conditions.
\begin{defn}[Self-consistent fixed point]
    The self consistent set of messages $\MM = \{\mu_{v\to w}\}_{v,w}$ satisfies for each $v\in V$ and each $n_j \in \NN(v)$,
\begin{equation}
        \left(\underset{n_i \in \NN(v)/\{n_j\}}{\otimes}\mu_{n_i \to v} \right)\star T_v = \mu_{v\to n_j}
\end{equation}
\end{defn}
For each edge $e = (v,w)\in E$, we expand the identity on the bond space $\mathbbm{1} \in \LL(\BB_e)$ as: 
\begin{eqnarray} 
    \mathbbm{1}_{vw} = \underbrace{\frac{\mu_{v\to w} \otimes \mu_{w\to v}}{I_{vw}}}_{\text{BP subspace}} + \underbrace{\mathcal{P}_{vw}^\perp}_{\text{excitation subspace}},
\end{eqnarray}
where  $I_{vw} = \mu_{v\to w} \star \mu_{w\to v}$. The normalization ensures that these subspaces are orthogonal: $\mathcal{P}_{vw}^\perp\star \mu_{v\to w} = 0$ and $\mu_{w\to v}\star\mathcal{P}_{vw}^\perp = 0$, hence $\mathcal{P}_{vw}^\perp$ carries contributions orthogonal to the BP vacuum. 
\begin{prop}[BP Approximation]
    For each $v\in V$, the local contribution to the partition function is then given as,
\begin{equation}\label{eq:tensor_z}
        Z^{(v)} := \left[\underset{n \in \NN(v)}{\bigotimes}\frac{\mu_{n\to v}}{\sqrt{I_{vw}}}\right] \star T_v,
\end{equation}
where  $I_{vw} = \mu_{v\to w} \star \mu_{w\to v}$, that is the inner product between $\mu_{v\to w}$ and $\mu_{w\to v}$. The BP approximation to  
\begin{enumerate}
    \item[(i)] the partition function is, 
\begin{equation}
    Z_{\emph{BP}}= \prod_{v\in V} Z^{(v)} 
\end{equation}
\item[(ii)] the free energy is, 
\begin{equation}
    F_{\emph{BP}} = -\sum_{v\in V}\log{Z^{(v)} }
\end{equation}
\end{enumerate}\end{prop}

We now recall the loop expansion~\cite{evenbly2024loopseriesexpansionstensor}. The set of allowed loop excitations is denoted $\LL$. For the na\"{\i}ve loop expansion, $\LL$ is the set of \emph{connected} subgraphs with each vertex having degree at least two. 

\begin{defn}[Generalized loops]\label{def:loop}
    Consider a graph $G=(V,E)$. A generalized loop is \emph{connected} subgraph $l=(W,F)$ with $W \subseteq V$, $F \subseteq E$, with the property that the degree of any $w \in W$ in $C$ is at least two. The weight of a generalized loop is defined as the number of edges $|F|$, also denoted $|l|$. The set of connected generalized loops is denoted $\LL$.
\end{defn}

Given a loop $l$, we define its support to be $\supp(l)$ the set of vertices in the associated subgraph. We can now define if two loops are compatible or incompatible.

\begin{defn}[Compatible loops]\label{def:incompatibility}
Two distinct loops \(l, l' \in \LL \) are said to be \emph{compatible}, written \(l \sim l'\), if they do not overlap; that is, they share no vertex or edge in the underlying graph. A family \(\Gamma \subset \LL\) of loops is called \emph{compatible} if every pair of distinct loops in \(\Gamma\) is compatible.
\end{defn}

The loop expansion is expressed in a sum–product form over connected, compatible loop excitations in $\LL$; as a consequence, it implicitly generates contributions from all disconnected loop configurations.

\begin{prop}[Loop expansion] \label{prop:loopseriesformal}
The tensor network contraction admits the expansion
\begin{equation}\label{eq:loop_expansion_reorganized}
  \ZZ = Z_{\emph{BP}}\left(1 + \sum_{\substack{\Gamma\subseteq\mathcal{L}\\\Gamma\emph{ finite, compatible}}}
  \prod_{l\in\Gamma} Z_l\right)
\end{equation}
where the sum runs over all finite sets \(\Gamma\) of mutually compatible loops.
\end{prop}

\smsubsection{Cluster expansion}

We now define the notion of cluster. A cluster is a multiset of loops.

\begin{defn}[Clusters]
  Given the set of allowed excitations $\LL$, a cluster is a collection of tuples of the form $$\mathbf{W} = \{(l_1, \eta_1), (l_2, \eta_2), \ldots\}$$ where each \(l_i\in \LL\) is a loop and \(\eta_i\) is the multiplicity of the loop \(l_i\) in the cluster. The total number of loops in the cluster is denoted as $n_{\mathbf{W}} = \sum_i \eta_i$.
\end{defn}

We define the following properties of a cluster.

\begin{defn}[Properties of a cluster]
    For a cluster $\mathbf{W}= \{(l_1, \eta_1), (l_2, \eta_2), \ldots\}$, we define the following:
    \begin{enumerate}
        \item[(i)] The \emph{cluster weight} (or ``\emph{order}") is defined as 
        \begin{equation}
            |\mathbf{W}| = \sum_{i} \eta_i |l_i|,
        \end{equation}
        where the number of edges in excitation $l$ be denoted as $|l|$.
        \item[(ii)] The \emph{support} of a cluster is,
        \begin{equation}
        \emph{supp}(Z_\mathbf{W}) = \cup_{l\in\mathbf{W}}\emph{supp}(l)
        \end{equation}
        \item[(iii)] The \emph{correction} of the cluster $Z_{\mathbf{W}}$ as the product of the loop corrections raised to their respective multiplicities:
        \begin{equation}
  Z_{\mathbf{W}} = \prod_iZ_{l_i}^{\eta_i}.
\end{equation}
\item[(iv)] The \emph{loop-set} of a cluster is the set of loops which appear in the cluster with non-zero multiplicity. That is, for a cluster $\mathbf{W} = \{(l_i,\eta_i)\}_i$ we have 
\begin{equation}
    \emph{loop-set}(\mathbf W) := \{l_i \, | \, \eta_i > 0\}
\end{equation}
    \end{enumerate}
\end{defn}


Given a cluster \(\mathbf{W}\), we define the \emph{interaction graph} as follows.

\begin{defn}[Interaction Graph]
  Given a cluster \(\mathbf{W} = \{(l_1, \eta_1), (l_2, \eta_2), \ldots \}\), we define the interaction graph \(G_\mathbf{W} = (V_\mathbf{W}, E_\mathbf{W})\) with $|V_\mathbf{W}| = \sum_i \eta_i$ vertices with each loop \(l_i\) corresponds to \(\eta_i\) vertices. There is an edge $(l,l') \in  E_\mathbf{W}$ either if the loops \(l\) and \(l'\) are incompatible ($l \not\sim l'$), or they are identical ($l=l'$). 
\end{defn}

We call a cluster \(\mathbf{W}\) \emph{connected} if the interaction graph \(G_\mathbf{W}\) is connected, meaning there is a path between any two vertices in the interaction graph. The following two results are from Ref.~\cite{midha2025beyond}.

\begin{prop}[Connected clusters only]\label{prop:clusterexpformal}
The negative free energy can be expressed as
\begin{equation}
  \log \ZZ = \log Z_{BP} +  \sum_{\emph{connected} \, \mathbf{W}} \phi(\mathbf{W}) Z_{\mathbf{W}},
\end{equation}
where the sum runs over all connected clusters \(\mathbf{W}\). The coefficient \(\phi(\mathbf{W})\) is given by
\begin{equation}
  \phi(\mathbf{W}) = \frac{1}{\mathbf{W}!} \sum_{\substack{C \in G_\mathbf{W} \\ C \emph{ connected}}} \sum_{(i,j) \in C} (-1)^{|E(C)|}
\end{equation}
\end{prop}

\begin{theorem}[Convergence of cluster expansion] \label{thm:clusterconvformal}
    Assume there exists a constant \(c>\log(2(\Delta-1)) + \frac{1}{2} \) such that
\begin{equation}
  |Z_l|\le e^{-c |l|}
\end{equation} 
then, the series for \(\log \ZZ \) converges absolutely. Moreover, the error in truncating the series at order \(m\) is bounded by
\begin{equation}
  \left| \FF -  {F}_m \right| \le N e^{-d(m+1)}
\end{equation}
where $d = c - \log(2(\Delta-1)) - \frac{1}{2}$ and $\Delta$ is the degree of the graph. 
\end{theorem}
The main technical argument in the proof of Theorem~\ref{thm:clusterconvformal} is the following tail-bound on the decay of large clusters.
\begin{lemma}[Tail bound for large clusters]
\label{lem:cluster_tail}
Consider a PEPS satisfying the loop decay condition $|Z_\ell| \le e^{-c|\ell|}$ for $c > c_0 = \log(2(\Delta-1)) + 1/2$. For any region $A \subset V$ and cutoff $m \in \N$, the contribution from all connected clusters with weight higher than $m$ supported on $A$ satisfies:
\begin{equation}
\sum_{\substack{\emph{connected} \, \mathbf{W} \\ \emph{supp}(\mathbf{W}) \cap A \neq\emptyset \\ |\mathbf{W}| > m}} |\phi(\mathbf{W}) Z_{\mathbf{W}}| \leq \order{|A|e^{-(c-c_0)(m+1)}}.
\end{equation}
\end{lemma}

\smsection{Cluster-cumulant expansion}

We briefly introduce the setting of the M\"{o}bius transform~\cite{rota1964theory}, used to convert the cluster expansion into the cluster-cumulant expansion in the main text.

\smsubsection{Preliminaries: M\"{o}bius transforms}

Let $(P, \le)$ be a poset. The associated \emph{incidence algebra} is
\[
\mathcal{A}(P)
= \bigl\{ \alpha \in \mathbb{C}^{P \times P} : \alpha(x,y) = 0 \text{ unless } x \le y \bigr\}.
\]
For $\alpha, \beta \in \mathcal{A}(P)$, the product is defined by
\[
(\alpha \beta)(x,y) = \sum_{z \in P} \alpha(x,z)\,\beta(z,y).
\]
It is immediate that $\mathcal{A}(P)$ is closed under addition, scalar multiplication, and multiplication.  
Let $I \in \mathcal{A}(P)$ denote the identity element,
\[
I(x,y) =
\begin{cases}
1, & x = y, \\
0, & x \ne y,
\end{cases}
\]
so that $I$ acts as the multiplicative identity in $\mathcal{A}(P)$.

The \emph{zeta function} of the poset is
\[
\zeta(x,y) =
\begin{cases}
1, & x \le y, \\
0, & \text{otherwise,}
\end{cases}
\]
and clearly $\zeta \in \mathcal{A}(P)$.

The \emph{Möbius function} $\mu \in \mathcal{A}(P)$ is defined recursively by
\[
\mu(x,y) =
\begin{cases}
0, & x \not\le y, \\
1, & x = y, \\
- \displaystyle\sum_{x \le z < y} \mu(x,z), & x < y.
\end{cases}
\]
It satisfies
\begin{equation}
\sum_{x \le z \le y} \mu(x,z) =
\begin{cases}
1, & x = y, \\
0, & \text{otherwise,}
\end{cases}
\label{eq:mobius-left-inverse}
\end{equation}
which expresses the left-inverse relation $\mu * \zeta = I$.  
Similarly, the right-inverse identity
\begin{equation}
\sum_{x \le z \le y} \mu(z,y) =
\begin{cases}
1, & x = y, \\
0, & \text{otherwise,}
\end{cases}
\label{eq:mobius-right-inverse}
\end{equation}
implies $\zeta * \mu = I$.  
Thus $\mu$ is the inverse of $\zeta$ in $\mathcal{A}(P)$ and is integer-valued by construction.

\begin{theorem}[Möbius Inversion]\label{thm:mobius-inversion}
Let $(P,\le)$ be a poset with Möbius function $\mu$.  
If two functions $f,g : P \to \mathbb{C}$ are related by
\[
g(x) = \sum_{y \le x} f(y),
\]
then
\begin{equation}
f(x) = \sum_{y \le x} \mu(y,x)\, g(y).
\label{eq:mobius-inversion}
\end{equation}
\end{theorem}

\begin{proof}
Substituting the definition of $g$ and using~\autoref{eq:mobius-left-inverse},
\begin{align*}
\sum_{y \le x} \mu(y,x) g(y)
&= \sum_{y \le x} \mu(y,x) \Biggl( \sum_{z \le y} f(z) \Biggr)
= \sum_{z \le y \le x} \mu(y,x) f(z) \\
&= \sum_{z \le x} f(z) \Biggl( \sum_{z \le y \le x} \mu(y,x) \Biggr)
= f(x).
\end{align*}
\end{proof}

\begin{prop}[Möbius Function on the Subset Lattice]\label{prop:mobius-subsets}
For the poset $P = \mathcal{P}([n])$ ordered by inclusion, the Möbius function is given by
\[
\mu(A,B) =
\begin{cases}
(-1)^{|B|-|A|}, & A \subseteq B, \\
0, & \text{otherwise.}
\end{cases}
\]
\end{prop}

\begin{proof}
By definition, $\mu(A,B) = 0$ if $A \not\subseteq B$. Now assume $A \subseteq B$. The interval $[A,B] = \{C : A \subseteq C \subseteq B\}$ is isomorphic to the power set $\mathcal{P}(B \setminus A)$ via the map $C \mapsto C \setminus A$. Since the Möbius function is an invariant of poset isomorphism, we have $\mu(A,B) = \mu(\emptyset, B \setminus A)$. 

Let $s = |B \setminus A|$. We prove $\mu(\emptyset, S) = (-1)^s$ by induction on $s$. For the base case: if $s=0$, then $S=\emptyset$ and $\mu(\emptyset,\emptyset) = 1 = (-1)^0$.

Suppose the claim holds for all sets of size $k < s$. By the defining property of the Möbius function:
    \[
    \sum_{\emptyset \subseteq C \subseteq S} \mu(\emptyset, C) = 0 \implies \mu(\emptyset, S) = -\sum_{\emptyset \subseteq C \subset S} \mu(\emptyset, C)
    \]
    Applying the inductive hypothesis $\mu(\emptyset, C) = (-1)^{|C|}$ and counting subsets by their size $i$:
    \[
    \mu(\emptyset, S) = -\sum_{i=0}^{s-1} \binom{s}{i} (-1)^i
    \]
    By the Binomial Theorem, $0 = (1-1)^s = \sum_{i=0}^{s} \binom{s}{i} (-1)^i$. Thus, the sum from $0$ to $s-1$ must equal $- \binom{s}{s}(-1)^s = -(-1)^s$. Substituting this back:
    \[
    \mu(\emptyset, S) = -(-(-1)^s) = (-1)^s
    \]

It follows that $\mu(A,B) = (-1)^{|B \setminus A|} = (-1)^{|B|-|A|}$.
\end{proof}

For us, $|A|$ will denote the number of generalized loops supported on region $A$.

\smsubsection{Cluster-cumulant expansion}

Starting from the ``standard" cluster expansion from \autoref{prop:clusterexpformal},
\begin{equation}
  \log \ZZ
  \;=\;
  \log Z_{BP}
  + \sum_{\text{connected}\,\mathbf W}
    \phi(\mathbf W)\,Z_{\mathbf W},
  \label{eq:cluster-expansion}
\end{equation}
each connected cluster $\mathbf W = \{(l_i,\eta_i)\}$ has weight $Z_{\mathbf W} = \prod_i Z_{l_i}^{\eta_i}$ and some loop-set $\text{loop-set}(\mathbf W) \subseteq \LL$.

Given a subset $\Gamma\subseteq\LL$,  we define the interaction graph $G_\Gamma$ as that of the cluster $\W = \{(l,1)\}_{l\in\Gamma}$. We thus inherit the notion of connectedness of a set $\Gamma\subseteq \LL$ as before. Define, for every finite subset $\Gamma \subset \LL$,
\begin{equation}
  \mathfrak K(\Gamma)
  \;:=\;
  \begin{cases}
    \displaystyle
    \sum_{\substack{\mathbf W\ \mathrm{connected}\\ \text{loop-set}(\mathbf W) =\Gamma}}
    \phi(\mathbf W)\,Z_{\mathbf W},
    & \text{if } G_\Gamma \text{ is connected}, \\[3ex]
    0, & \text{otherwise.}
  \end{cases}
  \label{eq:Kdef}
\end{equation}
That is, $\mathfrak K(\Gamma)$ collects all connected cluster contributions supported on $\Gamma$ and vanishes whenever the induced subgraph $G_\Gamma$ is disconnected.

Regrouping the sum in~\autoref{eq:cluster-expansion} by the distinct supports $\Gamma$ gives
\[
  \sum_{\text{connected }\mathbf W}
  =\sum_{\Gamma\subset\LL}\ \sum_{\substack{\mathbf W\ \text{connected}\\ \text{loop-set}(\mathbf W)=\Gamma}},
  \qquad
  \log\ZZ-\log Z_{BP}
  =\sum_{\substack{\Gamma\subset\LL\\ G_\Gamma\ \mathrm{connected}}}  \mathfrak K(\Gamma),
\]
so that the free energy can be written as a sum purely over finite sets.

For any finite $B \subset \LL$, define the \emph{restricted partition function}
\begin{equation}
\Xi(B)=1+\!\!\sum_{\substack{\Gamma'\subseteq B\\\Gamma'\ \mathrm{compatible}}}\!
   \prod_{l\in\Gamma'} Z_l.
   \label{eq:restricted-partition}
\end{equation}
The cluster expansion implies
\begin{equation}
    \log\Xi(B) = \sum_{\Gamma \subseteq B} \mathfrak{K}(\Gamma),
    \label{eq:logxi-cluster}
\end{equation}
which has the structure of a cumulative function over the subset lattice of $B$.  
Applying the Möbius inversion~\autoref{thm:mobius-inversion} together with \autoref{prop:mobius-subsets} yields the \emph{cumulant-cluster formula}
\begin{equation}
    \mathfrak{K}(A)
    = \sum_{B \subseteq A} (-1)^{|A|-|B|}\, \log\Xi(B).
    \label{eq:cumulant-cluster}
\end{equation}

\begin{theorem}[Cluster-cumulant expansion] \label{thm:clustercumulantexp}
    The free energy can be written as the sum of cumulants of all connected subgraphs $\Gamma$ of $\LL$. 
    \begin{equation}
  \boxed{\;
  \log\ZZ
  =\log Z_{BP}
  +\sum_{\Gamma\subseteq\LL}
   \mathfrak K(\Gamma).
  \;}
  \label{eq:cluster-to-cumulant}
\end{equation}
where the cumulants can be written down as the inclusion-exclusion sum of local partition functions, viz.,
\begin{align}
\Xi(B) &=1+\!\!\sum_{\substack{\Gamma'\subseteq B\\\Gamma'\ \mathrm{compatible}}}\!
   \prod_{l\in\Gamma'} Z_l \\ 
\mathfrak{K}(\Gamma)
    &= \sum_{B \subseteq \Gamma} (-1)^{|\Gamma|-|B|}\, \log\Xi(B).\label{eq:cumulant}
\end{align}

\end{theorem}
The support of a set of loops $\Gamma$, is defined as the union of support of the individual loops,
\begin{equation}
    \supp(\Gamma) := \cup_{l\in\Gamma} \supp(l)
\end{equation}
\smsubsection{Examples}
We illustrate the computation of the cumulants for a few simple cases.
\begin{itemize}
  \item For $A = \{l\}$,
  \[
  \Xi(A) = 1 + Z_l,
  \qquad
  \mathfrak{K}(A) = \log(1 + Z_l).
  \]
  \item For $A = \{l_1,l_2\}$ with $l_1 \not\sim l_2$ (incompatible),
  \[
  \Xi(A) = 1 + Z_{l_1} + Z_{l_2}, \qquad
  \mathfrak{K}(A)
  = \log(1 + Z_{l_1} + Z_{l_2})
    - \log(1 + Z_{l_1})
    - \log(1 + Z_{l_2}).
  \]
  \item For $A = \{l_1,l_2\}$ with $l_1 \sim l_2$ (compatible),
  \[
  \Xi(A) = (1 + Z_{l_1})(1 + Z_{l_2}),
  \qquad
  \mathfrak{K}(A) = 0.
  \]
\end{itemize}

\smsubsection{Radius of convergence}
We can now attempt truncating the cluster-cumulant series at a finite weight $|\Gamma| \leq m$, where for $\Gamma = \{l_1, l_2,\dots\}$ we have $|\Gamma| = |l_1| + |l_2|+\dots $.
    \begin{equation}\label{eq:cc-finite}
  F_m
  =-\log Z_{BP}
  -\sum_{\text{conn.} |\Gamma| \leq m}
   \mathfrak K(\Gamma).
\end{equation}
\begin{lemma}
    The cluster-cumulant expansion has the same radius of convergence as the cluster expansion.
\end{lemma}
\begin{proof}
We have that, 
\begin{equation}
    |F_m - \FF| = \sum_{\text{conn.} |\Gamma| \geq m+1}
   \mathfrak K(\Gamma).
\end{equation}
This can be upper bounded straightforwardly from the bound on $Z_\W$. Particularly, fix a site $v$, and we have for all connected sets $\Gamma$ supported on that site, 
\begin{align}
  \sum_{\Gamma \not\sim v:|\Gamma| \geq m+1} |\mathfrak K(\Gamma)| &\leq   \sum_{\Gamma:|\Gamma| \geq m+1} \sum_{\substack{\mathbf W\ \mathrm{connected}\\ \text{loop-set}(\mathbf W)=\Gamma \\ \W \not\sim i} }
    |\phi(\mathbf W)\,Z_{\mathbf W}| \\ 
   &\leq  \sum_{k=m+1}^{\infty} \sum_{\substack{
    \text{connected }\, \mathbf{W} \\
    \mathbf{W} \not\sim i \\
    |\mathbf{W}| = k
  }} |\phi(\mathbf{W}) Z_{\mathbf{W}}| \le \frac{1}{2} e^{-d(m+1)}
\end{align}
where $d = (c-c_0) > 0$ by decay of loops. From a union bound over all the sites in the lattice, we get,
\begin{equation}
    |F_m - \FF| \leq Ne^{-d(m+1)}
\end{equation}
\end{proof}

\smsubsection{``Top-down'' expansion}
As noted in the main text, there are two different ways to view the cluster-cumulant expansion. One is through a M\"obius transform as above. This construction is ``bottom up'': the cumulant $\mathfrak{K}(\Gamma)$ is defined solely through the subsets of $\Gamma$, per~\autoref{eq:cumulant}. 

Alternatively, we can reorganize the truncated series~\autoref{eq:cc-finite} as in~\autoref{eq:cc-top-down}, repeated below for convenience~\cite{welling2012}:
\begin{subequations}
    \begin{align}
    \label{eq:kB}\sum_{\mathrm{conn.} |\Gamma| \leq m} \mathfrak{K}(\Gamma) &= \sum_{\mathrm{conn.} |\Gamma|\leq m} \sum_{B \subseteq \Gamma} (-1)^{|\Gamma| - |B|} \log \Xi(B) \\
    \label{eq:bmB}&:= \sum_{\mathrm{conn.} |B| \leq m} b_m(B) \log \Xi(B)
\end{align}
\end{subequations}
where
\begin{equation}\label{eq:bm}
    b_m(B) := \sum_{\substack{\mathrm{conn}. \Gamma \supseteq B\\ |\Gamma|\leq m}} \mu(B,\Gamma) = 1 - \sum_{\substack{\mathrm{conn.} \Gamma\supset B \\ |\Gamma|\leq m}} b_m(\Gamma)
\end{equation}
The astute reader may question why we need only consider connected subsets $B$ in~\autoref{eq:bmB}, whereas the sum defining the cumulant (\autoref{eq:cumulant}) includes all subsets of $\Gamma$, including disconnected ones. The reason is simply that the restricted partition function (\autoref{eq:restricted-partition}) for a disconnected subgraph $\Gamma$ factorizes, so we can account for the contribution of a disconnected subgraph simply by adjusting the counting numbers of its connected components. To be concrete: suppose we included a disconnected subset $B=B_1 \cup B_2$ in the sum~\autoref{eq:bmB}, with counting number $\tilde{b}_m(B)$. This modifies $$b_m(B_1)\rightarrow \tilde{b}_m(B_1) = b_m(B_1) - b_m(B), \quad b_m(B_2) \rightarrow \tilde{b}_m(B_2) = b_m(B_2) - b_m(B),$$ while leaving all other counting numbers unchanged. This redefinition of the counting numbers leaves the partition function invariant, as $$\tilde{b}_m(B) \log \Xi(B) + \tilde{b}_m(B_1) \log\Xi(B_1) + \tilde{b}_m(B_2) \log \Xi(B_2) = b_m(B_1) \log\Xi(B_1) + b_m(B_2) \log\Xi(B_2).$$

We therefore arrive at a repackaging of the order-$m$ cluster-cumulant expansion for the free energy, 
\begin{equation}
    \FF = \FF_0 - \sum_{\substack{\mathrm{conn}. B \subseteq \mathcal{L} \\ |B| \leq m}} b_m(B) \log \Xi(B).
\end{equation}
Rather than organizing the expansion in terms of generalized loops, we can alternatively organize it in terms of \textit{regions}: vertex-induced subgraphs of $G$ containing up to $k$ vertices. The restricted partition function on a region $R$ is simply the contraction of the tensor network on that region, with BP messages inserted along its boundary. After normalizing the tensor network by the BP partition function, we exploit two facts: (1) the restricted partition function of an acyclic region vanishes, (2) the cumulant of a disconnected region vanishes~\cite{welling2012}. Together, these properties allow us to write the expansion, analogously to~\autoref{eq:bmB}, as the weighted sum over restricted partition functions on connected, leafless, vertex-induced subgraphs of $G$, 
\begin{equation}
    F_k = -\log Z_{BP} - \sum_{\substack{\mathrm{conn.} R \subseteq G \\ w(R) \leq k}} b_k(R) \log \Xi(R) \quad \mathrm{where} \quad b_k(R) = \sum_{\substack{\mathrm{conn.} A \supset R \\ w(A) \leq k}} b_k(A),
\end{equation}
where $w(R)$ denotes the number of vertices in $R$.
The recursive, top-down expression for the counting number -- a manifestation of the inclusion-exclusion principle -- also suggests an efficient way to identify the regions with nonzero counting numbers: first compute all maximal regions up to $k$ vertices (regions that are not the subgraphs of any other $w\leq k$ regions), then iteratively form intersections until no new regions are found. This procedure is outlined in~\autoref{alg:regions}. After the regions are obtained, all maximal regions are assigned $b_k(R)=1$, and the remaining counting numbers are determined recursively, starting from the largest regions.

\begin{algorithm}[hbtp]
\SetAlgoLined
\SetKwInOut{Output}{Output}
\SetKwInOut{Input}{Input}
\Input{Graph $G$, integer $k$}
\Output{Partially ordered set of regions $\mathcal{R}_1,\mathcal{R}_2,..$ such that no region in $\mathcal{R}_i$ is contained within a region in $\mathcal{R}_{j>i}$}
$\mathcal{M} \gets$ connected, leafless, vertex-induced subgraphs of $G$ containing up to $k$ vertices\;
$\mathcal{R}_1 \gets$ $\{R \in \mathcal{M} \, | \, R \not\subset R' \,  \forall \, R'\in\mathcal{M}\}$    (maximal/parent regions)\;
$i \gets$ 1\; 
\While{$\mathcal{R}_i$ is not empty} {
$\mathcal{R}_{i+1} \gets \emptyset$\;
\For{$R \in \mathcal{R}_i$, $R' \in \mathcal{R}_{j\leq i}$} {
	$P$ $\gets$ intersection of $R$ and $R'$\;
    \If{$P\notin \mathcal{R}_i$ {\normalfont{\textbf{and}}} $P$ is connected and leafless {\normalfont{\textbf{and}}} $P$ has at least one edge} {
    Add $P$ to $\mathcal{R}_{i+1}$\;}
}
$i \gets i+1$\; 
}

\Return $\{\mathcal{R}_1,...,\mathcal{R}_i\}$ \;
\caption{Region finding\label{alg:regions}}
\end{algorithm}


\smsection{Local Observables}

In this section, we state~\autoref{prop:localexpexpansion} and~\autoref{prop:localexpexpansion-derivative} in more detail and prove their convergence given loop decay (\autoref{thm:clusterconv}).

The key idea we will use is the following: the cluster expansion works on top of a set of geometrically allowed excitations $\LL$, with each excitation $l \in \LL$ associated with a weight $Z_l \in \C$. There is a hard-core pairwise interaction between the excitations, for $l,l' \in \LL$ 
\begin{equation}
\Delta(l, l') =
\begin{cases}
1, & \text{if $l$ and $l'$ are incompatible}, \\[4pt]
0, & \text{otherwise.}
\end{cases}
\end{equation}
The essential ingredient in what follows is a modification of the set of allowed excitations on the different TNs we consider. 

In both expansions, we
consider a state $\ket{\psi}$ described as a PEPS on a graph $G = (V,E)$. Consider the tensor network for the norm $\braket{\psi}$, denoted $\ZZ$. We find a set of fixed point messages $\MM$ for the norm network, and hence can expand the contraction, 
\begin{equation} \label{eq:loopnormnetwork}
    \braket{\psi}  =  Z_{BP}\left( 1 + \sum_{\substack{\Gamma\subset\mathcal{L}\\\Gamma\text{ finite, compatible}}}
  \prod_{l\in\Gamma} Z_l\right)
\end{equation}


Now, consider an observable $O_A$ with $\supp(O_A) \subseteq A$ for $A \subset V$ with $|A| = O(1)$. We can evaluate its expectation values in two different ways: (1) as the ratio of two partition functions or (2) as the derivative of a perturbed free energy. We consider these in turn.

\smsubsection{Ratio of partition functions}

Let the TN corresponding to $\expect{\psi}{O_A}{\psi}$ be denoted $\ZZ^A$. Consider expanding $\ZZ^A$ with the fixed points $\MM$ of the norm network. Crucially, self-consistency is preserved except for all $v\in A$. It is easy to see if we define $\LL_A$ to be the set of subgraphs of $V$ where each vertex in $V / A$ has degree at least two, then we get, 
\begin{equation}\label{eq:loopobservableA}
    \expect{\psi}{O_A}{\psi} =  Z^A_{BP}\left( 1 + \sum_{\substack{\Gamma\subset\mathcal{L}_A\\\Gamma\text{ finite, compatible}}}
  \prod_{l\in\Gamma} Z_l\right)
\end{equation}
Note the following points: (i) for all $l \in \LL_A$ with $\supp(l) \cap A = \emptyset$, we have $Z_l$ identical to the corresponding loop contribution of Eq.~\ref{eq:loopnormnetwork}, and (ii) $Z_{BP}^A$ differs from $Z_{BP}$ exactly inside $A$.

\begin{prop}\label{prop:localobsexpansionformal}
Let $\psi$ be a state represented as a PEPS on $G=(V,E)$. Denote $\ZZ$ to be the norm network $\inner{\psi}{\psi}$. For $A\subset V$, let $\ZZ^A$ denote the TN for $\expect{\psi}{O_A}{\psi}$ with $\emph{supp}(O_A) \subseteq A$. Given a cluster $\mathbf{W}$, denote the cluster correction $Z^{(A)}_\mathbf{W}$ on $\ZZ (\ZZ^A)$ respectively. Then we have that, 
    \begin{equation} 
    \expval{O_A} = \expval{O_A}_{\emph{BP}} \cdot \exp{\left(\sum_{\substack{\emph{connected} \, \mathbf{W} \leftarrow \LL_A \\ 
    \emph{supp}(\mathbf{W}) \cap A \neq\emptyset }} \phi_{\mathbf{W}} (Z^A_{\mathbf{W}} - Z_{\mathbf{W}})\right)}
\end{equation}
\end{prop}
\begin{proof}
We have,
\begin{equation}
    \expval{O_A} := \frac{\expect{\psi}{O_A}{\psi}}{\braket{\psi}}
\end{equation}
Consider now $\log\expval{O_A}$. We can apply the cluster theorem to both the TNs separately with appropriate modification of the excitation set ($\LL\to\LL_A$), to get, 
\begin{equation}
    \log\expval{O_A} = \log\frac{Z^A_{BP}}{Z_{BP}} + \sum_{\text{connected} \, \mathbf{W} \leftarrow \LL_A} \phi_{\mathbf{W}} Z^A_{\mathbf{W}} - \sum_{\text{connected} \, \mathbf{W} \leftarrow \LL} \phi_{\mathbf{W}} Z_{\mathbf{W}}
\end{equation}
where we denote $Z$ for loops/clusters on $\ZZ$ and $Z^A$ for loops/clusters on $\ZZ^A$. Note that, 
\begin{equation}
    \expval{O_A}_{\text{BP}} = \frac{Z^A_{BP}}{Z_{BP}}
\end{equation}
which depends on the messages on the boundary of $A$ and the tensors in $A$. Now, all clusters not supported on $A$ have identical $Z_\mathbf{W}$ in both the summations, since we expanded with the same message tensors. Finally, we can replace $\LL \to \LL_A$ for the denominator tensor for convenience, since all $l \in \LL_A / \LL$ evaluate to zero anyway. Hence we get, 
\begin{equation}
    \log\expval{O_A} = \log\expval{O_A}_{\text{BP}} + \sum_{\substack{\text{connected} \, \mathbf{W} \leftarrow \LL_A \\ 
    \supp(\mathbf{W}) \cap A \neq\emptyset }} \phi_{\mathbf{W}} (Z^A_{\mathbf{W}} - Z_{\mathbf{W}})
\end{equation}
Taking the exponential, 
    \begin{equation} 
    \expval{O_A} = \expval{O_A}_{\text{BP}} \cdot \exp{\left(\sum_{\substack{\text{connected} \, \mathbf{W} \leftarrow \LL_A \\ 
    \supp(\mathbf{W}) \cap A \neq\emptyset }} \phi_{\mathbf{W}} (Z^A_{\mathbf{W}} - Z_{\mathbf{W}})\right)}
\end{equation}
\end{proof}

Re-organizing the cluster expansion with the cumulant method as in \autoref{thm:clustercumulantexp}, we get the following cluster-cumulant expansion for local observables.

\begin{corollary}
    The cluster-cumulant version of the local operator expansion is given as,
    \begin{equation}   
    \expval{O_A} =\expval{O_A}_{\emph{BP}} \cdot \exp{\left(\sum_{\substack{\emph{conn.} \Gamma\subseteq\LL_A \\ 
    \emph{supp}(\Gamma) \cap A \neq\emptyset}} (\mathfrak K^A(\Gamma) - \mathfrak K(\Gamma))\right)}
\end{equation}
where the sum is over all connected subsets $\Gamma$ of $\LL_A$ that intersect $A$.
\end{corollary}

Truncating the cluster expansion for a local observable at order $m$ produces the estimate, 
\begin{equation}\label{eq:OA-cc} 
    \expval{O_A}_m = \expval{O_A}_{\text{BP}} \cdot \exp{\sum_{\substack{\text{conn} \, \mathbf{W} \leftarrow \LL_A \\ 
    \supp(\mathbf{W}) \cap A \neq\emptyset \\ |\mathbf W| \leq m}} \phi_{\mathbf{W}} (Z^A_{\mathbf{W}} - Z_{\mathbf{W}})}
\end{equation}

We show now that given decay of excitations on $\LL_A$, this approximates the exact value up to relative error exponentially small with cluster size. 
\begin{theorem}\label{thm:localobsalgorithmformal}
    Given a local region $A \subset V$, assume decay of loops in $\LL_A$ on both $\ZZ=\inner{\psi}{\psi}$ and $\ZZ^A = \expect{\psi}{O_A}{\psi}$. Then, truncating the cluster expansion for local expectation values (in the setting of \autoref{prop:localobsexpansionformal}) leads to a relative error $\delta_m = |\expval{O_A} - \expval{O_A}_m|/|\expval{O_A}|$ bounded by 
    \begin{equation}
        \delta_m \leq \order{|A| e^{-(c-c_0)(m+1)} }
    \end{equation}
    where $d=c-c_0 = \order{1}$.
\end{theorem}

\begin{proof}
    The relative error is by \autoref{prop:localobsexpansionformal}, 
\begin{equation}
    \delta_m := \left|\frac{\expval{O_A}_m - \expval{O_A}}{\expval{O_A}}\right| = \left|\exp{\sum_{\substack{\text{conn} \, \mathbf{W} \leftarrow \LL_A \\ 
    \supp(\mathbf{W}) \cap A \neq\emptyset \\ |\mathbf W| > m}} \phi_{\mathbf{W}} (Z^A_{\mathbf{W}} - Z_{\mathbf{W}})} - 1\right|
\end{equation}
Note that for any $x\in\C$,
\begin{align}
    |e^x - 1| &\leq \Big|\sum_{n \geq 1} \frac{x^n}{n!}\Big| \\
    &\leq \sum_{n\geq 1}\frac{|x|^n}{n!} \\ 
    &\leq e^{|x|} -1.
\end{align}
Now, we bound the relative error as follows,
\begin{align}
  \delta_m   &\leq \exp{\sum_{\substack{\text{conn} \, \mathbf{W} \leftarrow \LL_A \\ 
    \supp(\mathbf{W}) \cap A \neq\emptyset \\ |\mathbf W| > m}} |\phi_{\mathbf{W}} (Z^A_{\mathbf{W}} - Z_{\mathbf{W}})|} - 1 \\ 
    &\leq \prod_{X\in \{\emptyset, A\}} \exp{\sum_{\substack{\text{conn} \, \mathbf{W} \leftarrow \LL_A \\ 
    \supp(\mathbf{W}) \cap A \neq\emptyset \\ |\mathbf W| > m}} |\phi_{\mathbf{W}} Z^X_{\mathbf{W}}|} - 1 \\ 
    &\leq  \exp{r|A| e^{-(c-c_0)(m+1)}} - 1 \\ 
    &\leq r|A| e^{-(c-c_0)(m+1)} \exp{r|A| e^{-(c-c_0)(m+1)}} 
\end{align}
where $r=\order{1}$ is a constant and we used, 
\begin{equation}
    \sum_{\substack{\text{conn} \, \mathbf{W} \leftarrow \LL_A \\ 
    \supp(\mathbf{W})\ni A \\ |\mathbf W| > m}} |\phi_{\mathbf{W}} Z^X_{\mathbf{W}}| \leq \order{|A|e^{-(c-c_0)(m+1)}}
\end{equation}
by \autoref{lem:cluster_tail} and the fact that $e^x - 1 \leq xe^x$.

For large $m$ and $|A| =\order{1}$, the exponential prefactor is uniformly bounded by a constant $C=\order{1}$, and we get, 
\begin{equation}
    \delta_m \leq   \order{|A| e^{-(c-c_0)(m+1)} }
\end{equation}

\end{proof}


\smsubsection{Derivative of free energy}\label{sect:local-free}


In the derivative method, we instead consider the TN corresponding to $\expect{\psi}{\exp[\lambda O_A]}{\psi}$, 
and let $\ZZ_{\lambda}^A$ denote its contraction. We have the following additive version of the local observable expansion.
\begin{prop}\label{prop:localobsfreeformal}
Let $\psi$ be a state represented as a PEPS on $G=(V,E)$. For $A\subset V$, let $\ZZ_{\lambda}^A$ denote the TN for $\expect{\psi}{\lambda O_A}{\psi}$ with $\emph{supp}(O_A) \subseteq A$, while $\ZZ^A$, $\ZZ$ denote the TNs for $\expect{\psi}{O_A}{\psi}$ and $\inner{\psi}{\psi}$, respectively. Given a cluster $\mathbf{W}=\{(l_1,\eta_1),(l_2,\eta_2),...\}$, denote the cluster correction on $\ZZ_\lambda^A$, $\ZZ^A$, and $\ZZ$, by $\loopZ{\mathbf{W}}, Z_{\mathbf{W}}^A, Z_{\mathbf{W}}$, respectively. Also denote $\mathbf{W}_A$ the subset of loops within $\mathbf{W}$ which intersect $A$. Then
\begin{equation}
    \expval{O}_A = \partial_\lambda \log\ZZ_{\lambda}^A\mid_{\lambda = 0} \, = \expval{O_A}_{BP} + \sum_{\substack{\mathrm{conn} \, \mathbf{W} \leftarrow \LL_A \\ 
    \mathrm{supp}(\mathbf{W}) \cap A \neq\emptyset}} \phi_{\mathbf{W}} Z_{\mathbf{W}} \sum_{l \in \mathbf{W}_A} \eta_l \left[\frac{Z_l^A}{Z_l} - \expval{O_A}_{BP} \right].
\end{equation}
\end{prop}
\begin{proof}
The first equality follows immediately from the definition of $\ZZ_\lambda^A$:
\begin{equation}
\partial_\lambda \log
\expect{\psi}{\exp{\lambda O_A}}{\psi}\mid_{\lambda = 0} \, = \frac{\expect{\psi}{O_A}{\psi}}{\bra{\psi}\psi\rangle} = \expval{O_A}.
\end{equation}
Next, we write the cluster expansion for $\log\ZZ_{\lambda}^A$, using the message background 
$\MM$ which is self-consistent for $\ZZ_{\lambda=0}$, and the appropriately modified excitation set ($\LL\rightarrow \LL_A$):
\begin{equation}
\log\ZZ_{\lambda}^A = \log \zbpl + \sum_{\mathrm{conn.} \, \mathbf{W} \leftarrow \LL_A} \phi_{\mathbf{W}} \loopZ{\mathbf{W}}.
\end{equation}

Here, $\loopZ{\mathbf{W}}$ is the partition function of the cluster (with anti-projectors defined on each edge according to the $\lambda=0$ BP background), but with each tensor locally rescaled according to the perturbed BP partition function $\zbpl$.

Without loss of generality, we can take $A$ to be a single site; if $A$ contains more than one vertex, we can merge them into a single ``supervertex.'' Define
\[
\begin{tikzpicture}[
  scale=0.8,
  baseline={([yshift=-0.65ex] current bounding box.center)},
  tensor/.style={
    draw,
    line width=0.9pt,
    fill=tensorcolor!80!black,
    rounded corners=2pt,
    minimum size=8mm
  },
  optensor/.style={
    draw,
    line width=0.9pt,
    fill=opmaroon,
    rounded corners=2pt,
    minimum size=8mm
  },
  bond/.style={line width=2.0pt}
]
  \def\L{0.95}   
  \def\a{0.45}   
  \def\stub{0.45}
  \def\r{0.14}

\def\dx{1.3}
\def\dy{1.3}
\def\shear{0.0}
  \node at (-2.0,0){ $z_v=$};
  \node[tensor] (T-2-2) at (0,0) {};

  \EndVecCirc{-\L,0}{\stub}{\r}{2}{}
    \EndVecCirc{\L,0}{\stub}{\r}{0}{}

  \EndVecCirc{0,\L}{\stub}{\r}{1}{}
  \EndVecCirc{0,-\L}{\stub}{\r}{3}{}

;

\node at (2.5,0){$, \quad z_v(O)=$};
  \node[optensor] (T-2-2) at (5,0) {};

  \EndVecCirc{-\L+5,0}{\stub}{\r}{2}{}
    \EndVecCirc{\L+5,0}{\stub}{\r}{0}{}

  \EndVecCirc{5,\L}{\stub}{\r}{1}{}
  \EndVecCirc{5,-\L}{\stub}{\r}{3}{}
;
\end{tikzpicture},
\]
the local BP factor at vertex $v$ with and without the insertion of an operator $O$, respectively. For notational convenience, let
\begin{equation}
    z_{v,\lambda} = \begin{cases}
        z_v + \lambda z_v(O) & v=A \\
        z_v & \mathrm{otherwise}.
    \end{cases}
\end{equation}
Then, the perturbed BP partition function is
\begin{equation}
    \zbpl = \prod_{v \in V} z_{v,\lambda} = (z_A + \lambda z_A(O))\prod_{v \in V / A} z_v, 
\end{equation}
and the associated cluster expansion is
\begin{align}
    \log\ZZ_{\lambda}^A &= \sum_{v \in V} \log z_{v,\lambda}   + \sum_{\text{conn} \, \mathbf{W} \leftarrow \mathcal{L}_A} \frac{\phi_{\mathbf{W}} \nloopZ{\mathbf{W}}}{\prod_{l \in \mathbf{W}} \prod_{v \in l} z_{v,\lambda}^{\eta_l}} \notag \\
    &= \sum_{v \in V} \log z_{v,\lambda}   + \sum_{\text{conn} \, \mathbf{W} \leftarrow \mathcal{L}_A} \frac{\phi_{\mathbf{W}} {\prod_{l \in \mathbf{W}} {\left(\nloopZ{l}\right)^{\eta_l}}}}{\prod_{l \in \mathbf{W}} \prod_{v \in l} z_{v,\lambda}^{\eta_l}}
\end{align}

where $\nloopZ{\mathbf{W}}$ denotes the unnormalized cluster weight, and $\nloopZ{l}$ the unnormalized loop correction for loop $l$.

We are now ready to take the derivative. The derivative acts trivially on all BP factors $z_{v,\lambda}$ except at $A$. It also acts trivially on any loop that does not touch $A$. Thus, the derivative is nonvanishing only for clusters with support on $A$. For these clusters, we will use $\mathbf{W}_{\bar{A}}$ to denote the subset of strings that do not intersect $A$; for $l\in \mathbf{W}_{\bar{A}}$, the normalized loop correction simplifies to
$$
\frac{\left(\nloopZ{l}\right)^{\eta_l}}{\prod_{v\in l} z_{v,\lambda}^{\eta_l}} = Z_l^{\eta_l}.$$
Thus, taking the derivative, we obtain,
\begin{subequations}
\begin{align}
    \partial_\lambda \ZZ^A_\lambda\mid_{\lambda=0} &= \frac{z_A(O)}{z_A} + \sum_{\substack{\text{conn} \, \mathbf{W} \leftarrow \LL_A \\ 
    \supp(\mathbf{W}) \cap A \neq\emptyset}} \phi_{\mathbf{W}} \frac{ \prod_{l \in \mathbf{W}_{\bar{A}}} Z_l^{\eta_l}} {\prod_{l \in \mathbf{W}_A} z_v^{\eta_l}} \left[\partial_\lambda \eval{\left(\prod_{l \in \mathbf{W}_A} \left(\nloopZ{l}\right)^{\eta_l}\right)}_{\lambda=0} - \prod_{l \in \mathbf{W}_A} Z_l^{\eta_l} \partial_\lambda \eval{\left(\prod_{l \in \mathbf{W}_A} z_{A,\lambda}^{\eta_l}\right)}_{\lambda=0}\right] \notag\\
    &\label{eq:derivative-evaluate}= \expval{O_A}_{BP} + \sum_{\substack{\text{conn} \, \mathbf{W} \leftarrow \LL_A \\ 
    \supp(\mathbf{W}) \cap A \neq\emptyset}} \phi_{\mathbf{W}} \sum_{l \in \mathbf{W}_A} \eta_l \left[Z_l^A (Z_l)^{\eta_l-1} \prod_{\substack{l'\in \mathbf{W} \\ l' \neq l}} (Z_{l'})^{\eta_{l'}} - Z_{\mathbf{W}} \expval{O_A}_{BP} \right] \\
    &\label{eq:derivative-evaluate2}= \expval{O_A}_{BP} + \sum_{\substack{\text{conn} \, \mathbf{W} \leftarrow \LL_A \\ 
    \supp(\mathbf{W}) \cap A \neq\emptyset}} \phi_{\mathbf{W}} Z_{\mathbf{W}} \sum_{l \in \mathbf{W}_A} \eta_l \left[\frac{Z_l^A}{Z_l} - \expval{O_A}_{BP} \right].
    \end{align}
    \end{subequations}
\end{proof}
Since $Z_l=0$ for any open string, the term
$Z_l^A (Z_l)^{\eta_l-1} \prod_{\substack{l'\in \mathbf{W} \\ l' \neq l}} (Z_{l'})^{\eta_{l'}}$ in~\autoref{eq:derivative-evaluate} can be non-vanishing only if all $l'$ are closed loops and either $l$ is a closed loop, or it has multiplicity 1. The term $Z_{\mathbf{W}} \expval{O_A}_{BP}$ can be non-vanishing only if $\mathbf{W} \leftarrow \LL$, i.e. consists of only closed loops. \autoref{eq:derivative-evaluate2} is more concise, but care must be taken in the case where $Z_l^A$ is nonvanishing but $Z_l$ is zero.

The convergence of this expansion follows from~\autoref{eq:derivative-evaluate} and the application of~\autoref{lem:cluster_tail}, as we now prove.
\begin{theorem}\label{thm:localobsalgorithmformal-derivative}
    Given a local region $A \subset V$, assume decay of loops in $\LL_A$ on both $\ZZ=\inner{\psi}{\psi}$ and $\ZZ^A = \expect{\psi}{O_A}{\psi}$. Then, truncating the cluster expansion (derivative version) for local expectation values (\autoref{prop:localobsfreeformal}) of an observable with with $\|O_A\|=1$ in a region $A\subset V$ leads to an additive error bounded by 
    \begin{equation}
        |\expval{O_A} - \expval{O_A}_m| \leq \order{m |A|e^{-d(m+1)}},
    \end{equation}
    where $d=c-c_0$.
\end{theorem}
\begin{proof}
By~\autoref{eq:derivative-evaluate2}, the additive error is
\begin{equation}
 |\expval{O_A}_m - \expval{O_A}_m| = \left|\sum_{\substack{\text{conn} \, \mathbf{W} \leftarrow \LL_A \\ 
    \supp(\mathbf{W}) \cap A \neq\emptyset \\ |\mathbf{W}| > m}} \phi_{\mathbf{W}} Z_{\mathbf{W}} \sum_{l \in \mathbf{W}_A} \eta_l \left[\frac{Z_l^A}{Z_l} - \expval{O_A}_{BP} \right]\right|.
\end{equation}
We bound the sum by 
\begin{align}
&\sum_{\substack{\text{conn} \, \mathbf{W} \leftarrow \LL_A \\ 
    \supp(\mathbf{W}) \cap A \neq\emptyset \\ |\mathbf{W}| > m}} |\phi_{\mathbf{W}} Z_{\mathbf{W}}| \sum_{l \in \mathbf{W}_A} \eta_l \left|\frac{Z_l^A}{Z_l} - \expval{O_A}_{BP} \right| \notag \\
    &\leq \order{\sum_{\substack{\text{conn} \, \mathbf{W} \leftarrow \LL_A \\ 
    \supp(\mathbf{W}) \cap A \neq\emptyset \\ |\mathbf{W}| > m}} |\phi_{\mathbf{W}}|  \max(|Z_{\mathbf{W}}|, |Z^A_{\mathbf{W}}|) \sum_{l \in \mathbf{W}_A} \eta_l} \notag \\
    &\leq \order{\sum_{\substack{\text{conn} \, \mathbf{W} \leftarrow \LL_A \\ 
    \supp(\mathbf{W}) \cap A \neq\emptyset \\ |\mathbf{W}| > m}} |\phi_{\mathbf{W}}| \max(|Z_{\mathbf{W}}|,|Z^A_{\mathbf{W}}|) |\mathbf{W}|} 
\end{align}
In going from the first to the second line, we upper bounded $|Z_{\mathbf{W}}||Z_l^A/Z_l - \expval{O_A}_{BP}| \leq |Z_{\mathbf{W}} \cdot Z_l^A/Z_l| + |Z_{\mathbf{W}}| \cdot |\expval{O_A}_{BP}| \leq \order{\max(|Z_\mathbf{W}|, |Z^A_{\mathbf{W}}|)}$ since $|\expval{O_A}_{BP}| = |z_v(O_A) / z_v| \leq \order{1}$ (as $\|O_A\| =1$). In going from the second to the third line, we used the loose upper bound $\sum_{l\in \mathbf{W}_A} \eta_l \leq \sum_{l\in \mathbf{W}} \eta_l \leq \sum_{l\in \mathbf{W}} \eta_l |l| \leq |\mathbf{W}|$ since $|l| > 1$ by definition.

Now we invoke~\autoref{lem:cluster_tail}, which can be restated as
\begin{equation}
\sum_{\substack{\text{connected} \, \mathbf{W} \\ \supp(\mathbf{W}) \cap A \neq\emptyset \\ |\mathbf{W}| = m}} |\phi_{\mathbf{W}} Z_{\mathbf{W}}| \leq \order{|A|e^{-(c-c_0)m}}.
\end{equation}
 Therefore 
\begin{equation}
\sum_{\substack{\text{conn} \, \mathbf{W} \leftarrow \LL_A \\ 
    \supp(\mathbf{W}) \cap A \neq\emptyset \\ |\mathbf{W}| > m}} |\phi_{\mathbf{W}} Z_{\mathbf{W}}| |\mathbf{W}| 
        \leq \order{|A| \sum_{k=m+1}^\infty k e^{-(c-c_0) k}} \leq \order{m |A| e^{-d(m+1)}}.
\end{equation}
\end{proof}

\smsubsection{Region-based expansions}
Now we turn to the regions-based version of the cluster cumulant expansion, used to obtain ``order $k$'' approximations to local observables in~\autoref{sect:finiteT} of the main text.

The region-finding step is modified from~\autoref{alg:regions} to~\autoref{alg:regions2}. We assume that $A$ is a single vertex; otherwise, $A$ can be merged into a ``supervertex.'' The two main differences from~\autoref{alg:regions} are: (1) Only regions that contain $A$ are included, and $A$ can be a leaf of $R$ (line 1). (2) When taking intersections, $P$ is not required to be leafless, but all branches terminating on a vertex other than $A$ can be pruned (line 9)~\cite{welling2012,gray2025}.

\begin{algorithm}[hbtp]
\SetAlgoLined
\SetKwInOut{Output}{Output}
\SetKwInOut{Input}{Input}
\Input{Graph $G$, integer $k$, vertex $A$}
\Output{Partially ordered set of regions $\mathcal{R}_1,\mathcal{R}_2,..$ such that no region in $\mathcal{R}_i$ is contained within a region in $\mathcal{R}_{j>i}$}
$\mathcal{M} \gets$ connected, vertex-induced subgraphs of $G$ that contain up to $k$ vertices including $A$, and no other leaves\;
$\mathcal{R}_1 \gets$ $\{R \in \mathcal{M} \, | \, R \not\subset R' \,  \forall \, R'\in\mathcal{M}\}$    (maximal/parent regions)\;
$i \gets$ 1\; 
\While{$\mathcal{R}_i$ is not empty} {
$\mathcal{R}_{i+1} \gets \emptyset$\;
\For{$R \in \mathcal{R}_i$, $R' \in \mathcal{R}_{j\leq i}$} {
	$P$ $\gets$ intersection of $R$ and $R'$\;
    \If{$P$ is connected {\normalfont{\textbf{and}}} $P$ contains $A$} {
    $P \gets$ prune branches not ending on $A$\;
    \lIf{$P \notin \mathcal{R}_i$}{
    add $P$ to $\mathcal{R}_{i+1}$}}
}
$i \gets i+1$\; 
}

\Return $\{\mathcal{R}_1,...,\mathcal{R}_i\}$ \;
\caption{Region finding, local observables \label{alg:regions2}}
\end{algorithm}

Having obtained the pertinent regions, we turn to evaluating the derivative of the perturbed partition function. Here, we note that it is more convenient \textit{not} to normalize the tensor network by the BP partition function. We will use $\widetilde{\Xi}_\lambda^A(R)$ to denote the unnormalized partition function on region $R$, with the operator insertion $\exp(\lambda O_A)$. For maximal region size $k$, and an associated set of regions $\mathcal{R}$, we arrive at the approximation
\begin{equation}
    \partial_\lambda \ZZ_\lambda^A\mid_{\lambda = 0} \, \approx \sum_{R \in \mathcal{R}} b_k(R) \partial_\lambda \left(\log \widetilde{\Xi}_\lambda^A(R)\right)|_{\lambda = 0} \, = \sum_{R\in \mathcal{R}} b_k(R) \expval{O_A}_R.
\end{equation}
This is the form of the expansion used for the finite-temperature numerics in the main text. It is the ``weighted arithmetic mean'' used in Ref.~\cite{park2025}, dubbed the ``loop cluster sum formula'' in Ref.~\cite{gray2025}.

Ref.~\cite{gray2025} also considers the ``loop cluster product formula,'' interpreted as a ``weighted geometric mean'' in Ref.~\cite{park2025}. That formula can be seen as the top-down, region-based version of our ratio expansion (\autoref{eq:OA-cc}):
\begin{equation}
\expval{O_A} = \frac{\mathcal{Z}_{A}}{\mathcal{Z}} \approx \frac{\exp\left[\sum_R b_k(R) \log \widetilde{\Xi}_{A}(R)\right]}{\exp\left[\sum_R b_k(R) \log \widetilde\Xi(R)\right]} = \prod_R \left(\frac{\widetilde{\Xi}_{A}(R)}{\widetilde{\Xi}(R)}\right)^{b_k(R)} = \prod_R (\langle O_A \rangle_R)^{b_k(R)}.
\end{equation}

\smsubsection{A note on loop vs. string decay}
We conclude this section by commenting upon the relationship between loop decay in the norm network $\ZZ$ and the decay of excitations on $\ZZ^A$. 

Suppose the tensor network $\ZZ$ satisfies $c$-decay with $c>c_0$. Does the cluster expansion for an arbitrary local observable $O_A$ converge? To answer this question, consider a string $l$ in the excitation set $\LL_A$. To guarantee a convergent expansion, a sufficient condition is (1) $|Z_l| \leq 
\order{e^{-c|l|}}$ and (2) $|Z_l^{A}| \leq \order{e^{-c_A |l|}}$ with $c_A>c_0$.

If $|Z_l|$ is nonzero, the first inequality implies the second, with $c_A=c$: changing a local region does not modify the rate of exponential decay. However, if $|Z_l| = 0$, then (1) is trivially satisfied, but $|\loopZ{l}|$ need not decay for $\lambda \neq 0$. Open strings in $\LL_A / \LL$ present one example of this caveat: the fact that $|Z_l| = 0$ for such strings follows directly from the fixed point condition, but says nothing about the decay when the tensor network is perturbed.

Despite these caveats, in states with clustering of correlations, we \textit{do} expect loop decay on $\ZZ$ and string decay on $\ZZ^A$ to be parametrically similar. Indeed, we can think of $c-c_0$ as an effective ``string tension'' which is stable under local perturbations. For the same reason, the performance of the expansion will generally be insensitive to the choice of bounded observable. The exception is again when certain sets of loops/strings are identically zero: for example, in the PEPS of the 2D classical Ising model, $\expval{X}_{BP} = 0$ and all loop corrections vanish for any choice of fixed point, including ``poor'' fixed points. It would be interesting if one could observe different decay patterns of strings on $\LL_A$ for different choices of local observables. This may lead to certain observables being ``harder" to estimate than others, e.g., order parameters near a phase transition.
\smsection{Correlation Functions}

We can now apply a similar reasoning to connected correlation functions $\expval{O_AO_B}_c := \expval{O_AO_B} - \expval{O_A}\expval{O_B}$, followed by a generalization to $p$-pointed connected correlations.
\smsubsection{Ratio of partition functions}
\begin{prop}\label{prop:correlatorexpansionformal}
Let $\psi$ be a state represented as a PEPS on $G=(V,E)$. The ratio form of the cluster expansion for connected correlators is given as,
\begin{align}
\expval{O_AO_B}_c &= 
 \expval{O_A} \expval{O_B}
\left[
\exp\!\Biggl\{
\sum_{\substack{\emph{connected}\, \mathbf{W} \leftarrow \LL_{AB} \\ 
\emph{supp}(\mathbf{W}) \cap A \neq \emptyset \\ 
\emph{supp}(\mathbf{W}) \cap B \neq \emptyset}}
\phi_{\mathbf{W}}\left(Z^{AB}_\mathbf{W} + Z_\mathbf{W} - Z^{A}_\mathbf{W} - Z^{B}_\mathbf{W}\right)
\Biggr\}
- 1 \right]
\end{align}
where $Z_\mathbf{W}^X$ refers to evaluation of the cluster $\mathbf{W}$ on the network $\ZZ_{X} = \expect{\psi}{O_X}{\psi}$, where $O_X = O_{A(B)}$ for $X = A(B)$ and $O_X = O_AO_B$ for $X= AB$.
\end{prop}

\begin{proof}
    Now, note that $\expval{O_AO_B}_c = \expval{O_AO_B} - \expval{O_A}\expval{O_B}$. The individual parts are, by repeated application of \autoref{prop:localobsexpansionformal},
    \begin{align}
    \expval{O_A} &= \expval{O_A}_{\text{BP}} \cdot \exp{\left(\sum_{\substack{\text{connected} \, \mathbf{W} \leftarrow \LL_A \\ 
    \supp(\mathbf{W}) \cap A \neq\emptyset }} \phi_{\mathbf{W}} (Z^A_{\mathbf{W}} - Z_{\mathbf{W}})\right)} \\ 
    \expval{O_B} &= \expval{O_B}_{\text{BP}}\cdot \exp{\left(\sum_{\substack{\text{connected} \, \mathbf{W} \leftarrow \LL_B \\ 
    \supp(\mathbf{W}) \cap B \neq\emptyset }} \phi_{\mathbf{W}} (Z^B_{\mathbf{W}} - Z_{\mathbf{W}})\right)} \\ 
    \expval{O_AO_B} &= \expval{O_A}_{\text{BP}} \expval{O_B}_{\text{BP}} \cdot \exp{\left(\sum_{\substack{\text{connected} \, \mathbf{W} \leftarrow \LL_{AB} \\ 
   \supp(\mathbf{W}) \cap\{A,B\} \neq  \emptyset}} \phi_{\mathbf{W}} (Z^{AB}_{\mathbf{W}} - Z_{\mathbf{W}})\right)}
    \end{align}

We note the following,

\begin{align}
    \sum_{\substack{\text{connected} \, \mathbf{W} \leftarrow \LL_{AB} \\ \supp(\mathbf{W}) \cap\{A,B\} \neq  \emptyset}} &=  \textcolor{red}{\sum_{\substack{\text{connected} \, \mathbf{W} \leftarrow \LL_{AB} \\ \supp(\mathbf{W}) \cap\{A\} = \emptyset \\ \supp(\mathbf{W}) \cap\{B\} \neq \emptyset}}} + \textcolor{blue}{\sum_{\substack{\text{connected} \, \mathbf{W} \leftarrow \LL_{AB} \\ \supp(\mathbf{W}) \cap\{A\} \neq \emptyset \\ \supp(\mathbf{W}) \cap\{B\} = \emptyset}}}  + \sum_{\substack{\text{connected} \, \mathbf{W} \leftarrow \LL_{AB} \\ \supp(\mathbf{W}) \cap\{A\} \neq \emptyset \\ \supp(\mathbf{W}) \cap\{B\} \neq \emptyset}} \\  
    \sum_{\substack{\text{connected} \, \mathbf{W} \leftarrow \LL_{A} \\ \supp(\mathbf{W}) \cap\{A\} \neq  \emptyset}} &=  \textcolor{blue}{\sum_{\substack{\text{connected} \, \mathbf{W} \leftarrow \LL_{A} \\ \supp(\mathbf{W}) \cap\{A\} \neq \emptyset \\ \supp(\mathbf{W}) \cap\{B\} = \emptyset}}}  + \sum_{\substack{\text{connected} \, \mathbf{W} \leftarrow \LL_{A} \\ \supp(\mathbf{W}) \cap\{A\} \neq \emptyset \\ \supp(\mathbf{W}) \cap\{B\} \neq \emptyset}} \\  
    \sum_{\substack{\text{connected} \, \mathbf{W} \leftarrow \LL_{B} \\ \supp(\mathbf{W}) \cap\{B\} \neq  \emptyset }} &=  \textcolor{red}{\sum_{\substack{\text{connected} \, \mathbf{W} \leftarrow \LL_{B} \\ \supp(\mathbf{W}) \cap\{A\} = \emptyset \\ \supp(\mathbf{W}) \cap\{B\} \neq \emptyset}}}    + \sum_{\substack{\text{connected} \, \mathbf{W} \leftarrow \LL_{B} \\ \supp(\mathbf{W}) \cap\{A\} \neq \emptyset \\ \supp(\mathbf{W}) \cap\{B\} \neq \emptyset}} 
\end{align}
Since $\TT^{AB}$ differs from $\TT^{A(B)}$ only at $B(A)$, we get, 
\begin{align}
\expval{O_AO_B}_c &= 
\expval{O_A}_{\text{BP}} \expval{O_B}_{\text{BP}} 
\,\textcolor{blue}{f_A}\,\textcolor{red}{f_B}
\times \nonumber \\[2mm]
&\biggl[
\exp\!\Biggl\{
\sum_{\substack{\text{connected}\, \mathbf{W} \leftarrow \LL_{AB} \\ 
\supp(\mathbf{W}) \cap A \neq \emptyset \\ 
\supp(\mathbf{W}) \cap B \neq \emptyset}}
\phi_{\mathbf{W}}(Z^{AB}_\mathbf{W} - Z_\mathbf{W})
\Biggr\}
-  \\ &
\exp\!\Biggl\{
\sum_{\substack{\text{connected}\, \mathbf{W} \leftarrow \LL_A \\ 
\supp(\mathbf{W}) \cap A \neq \emptyset \\ 
\supp(\mathbf{W}) \cap B \neq \emptyset}}
\phi_{\mathbf{W}}(Z^{A}_\mathbf{W} - Z_\mathbf{W})
\Biggr\}
\exp\!\Biggl\{
\sum_{\substack{\text{connected}\, \mathbf{W} \leftarrow \LL_B \\ 
\supp(\mathbf{W}) \cap A \neq \emptyset \\ 
\supp(\mathbf{W}) \cap B \neq \emptyset}}
\phi_{\mathbf{W}}(Z^{B}_\mathbf{W} - Z_\mathbf{W})
\Biggr\}
\biggr]
\end{align}
with (this factors out since, e.g, for $\mathbf{W}$ supported exclusively on $A$ but not $B$, we have $Z_\mathbf{W}^{AB} = Z_\mathbf{W}^{A}$)
\begin{align}
   \textcolor{blue}{f_A} &= \exp{\textcolor{blue}{\sum_{\substack{\text{connected} \, \mathbf{W} \leftarrow \LL_{AB} \\ \supp(\mathbf{W}) \cap\{A\} \neq \emptyset \\ \supp(\mathbf{W}) \cap\{B\} = \emptyset}}}   \phi_{\mathbf{W}}(Z^{A}_\mathbf{W} - Z_\mathbf{W})} \\ 
   \textcolor{red}{f_B} &=\exp{  \textcolor{red}{\sum_{\substack{\text{connected} \, \mathbf{W} \leftarrow \LL_{AB} \\ \supp(\mathbf{W}) \cap\{A\} = \emptyset \\ \supp(\mathbf{W}) \cap\{B\} \neq \emptyset}}} \phi_{\mathbf{W}}(Z^{B}_\mathbf{W} - Z_\mathbf{W})}
\end{align}
Factoring out further the local terms, we thus arrive at,

\begin{align}
\expval{O_AO_B}_c &= 
 \expval{O_A} \expval{O_B}
\left[
\exp\!\Biggl\{
\sum_{\substack{\text{connected}\, \mathbf{W} \leftarrow \LL_{AB} \\ 
\supp(\mathbf{W}) \cap A \neq \emptyset \\ 
\supp(\mathbf{W}) \cap B \neq \emptyset}}
\phi_{\mathbf{W}}\left(Z^{AB}_\mathbf{W} + Z_\mathbf{W} - Z^{A}_\mathbf{W} - Z^{B}_\mathbf{W}\right)
\Biggr\}
- 1 \right]
\end{align}
where we used, 
\begin{align}
   \expval{O_A} &= \expval{O_A}_{\text{BP}} \cdot \exp{\left(\sum_{\substack{\text{connected} \, \mathbf{W} \leftarrow \LL_A \\ 
    \supp(\mathbf{W}) \cap A \neq\emptyset }} \phi_{\mathbf{W}} (Z^A_{\mathbf{W}} - Z_{\mathbf{W}})\right)} \\ 
    &\equiv \expval{O_A}_{\text{BP}} \cdot \textcolor{blue}{f_A} \cdot \exp\!\Biggl\{
\sum_{\substack{\text{connected}\, \mathbf{W} \leftarrow \LL_A \\ 
\supp(\mathbf{W}) \cap A \neq \emptyset \\ 
\supp(\mathbf{W}) \cap B \neq \emptyset}}
\phi_{\mathbf{W}}(Z^{A}_\mathbf{W} - Z_\mathbf{W})
\Biggr\}
\end{align}
and similarly for $O_B$.
\end{proof}

This can again be re-organized into a cluster-cumulant expansion as follows.

\begin{corollary}
Let $\psi$ be a state represented as a PEPS on $\Lambda$. The cumulant-cluster expansion for connected correlators is given as,
\begin{align}
\expval{O_AO_B}_c &= 
\expval{O_A} \expval{O_B}
\left[
\exp\!\Biggl\{
\sum_{\substack{\emph{conn.}\, \Gamma  \subseteq \LL_{AB} \\ 
\emph{supp}(\Gamma) \cap A \neq \emptyset \\ 
\emph{supp}(\Gamma) \cap B \neq \emptyset}}
\left(\mathfrak{K}^{AB}_\mathbf{W} + \mathfrak{K}_\mathbf{W} - \mathfrak{K}^{A}_\mathbf{W} - \mathfrak{K}^{B}_\mathbf{W}\right)
\Biggr\}
- 1 \right]
\end{align}
\end{corollary}

The cluster expansion for connected correlators can then be used to show that decay of loops ensures clustering of bipartite correlations as follows.    

\begin{theorem}\label{thm:correlatorboundformal}
Decay of loops in $\LL_{AB}$ on the networks $\ZZ^{X}$ for $X\in \{\emptyset,A,B,AB\}$ ensures that bipartite connected correlators $\expval{O_AO_B}_c := \expval{O_AO_B} - \expval{O_A} \expval{O_B}$ for bounded local observables $O_{A(B)}$ with $\emph{supp}(O_{A(B)})\subseteq A(B)$ and $\|O_{A(B)}\|=1$ cluster, 
\[
\Bigl|\expval{O_AO_B}_c\Bigr| \leq \order{e^{-d(A,B) / \xi} }
\]
for a finite correlation length $\xi \le \order{(1/(c-c_0))}$ and $d(A,B)$ being the graph distance.
\end{theorem}
\begin{proof}
    We recall, by \autoref{prop:correlatorexpansionformal}
    \begin{align}
\expval{O_AO_B}_c &= 
\expval{O_A} \expval{O_B}
\left[
\exp\!\Biggl\{
\sum_{\substack{\text{connected}\, \mathbf{W} \leftarrow \LL_{AB} \\ 
\supp(\mathbf{W}) \cap A \neq \emptyset \\ 
\supp(\mathbf{W}) \cap B \neq \emptyset}}
\phi_{\mathbf{W}}\left(Z^{AB}_\mathbf{W} + Z_\mathbf{W} - Z^{A}_\mathbf{W} - Z^{B}_\mathbf{W}\right)
\Biggr\}
- 1 \right]
\end{align}

    Now,
\begin{align}
\Bigl|\expval{O_AO_B}_c\Bigr|
&\le \Big|\expval{O_A} \expval{O_B}\left[
\exp\!\Biggl\{
\sum_{\substack{\text{connected}\, \mathbf{W} \leftarrow \LL_{AB} \\ 
\supp(\mathbf{W}) \cap A \neq \emptyset \\ 
\supp(\mathbf{W}) \cap B \neq \emptyset}}
\left|\phi_{\mathbf{W}}\left(Z^{AB}_\mathbf{W} + Z_\mathbf{W} - Z^{A}_\mathbf{W} - Z^{B}_\mathbf{W}\right)\right|
\Biggr\}
- 1 \right] \Big|\\[1em]
&\le |\expval{O_A}| |\expval{O_B}|\left[
\exp\!\Biggl\{
\sum_{\substack{\text{connected}\, \mathbf{W} \leftarrow \LL_{AB} \\ 
\supp(\mathbf{W}) \cap A \neq \emptyset \\ 
\supp(\mathbf{W}) \cap B \neq \emptyset}}
\Bigl|\phi_{\mathbf{W}}\Bigr|\,
\Bigl(\,|Z^{AB}_\mathbf{W}| + |Z_\mathbf{W}| + |Z^{A}_\mathbf{W}| + |Z^{B}_\mathbf{W}|\,\Bigr)
\Biggr\}
- 1 \right] \\[1em]
&= \|O_A\| \|O_B\| \left[
\exp\!\Biggl\{
\sum_{X\in\{\emptyset,A,B,AB\}}
\sum_{\substack{\text{connected}\, \mathbf{W} \leftarrow \LL_{AB} \\ 
\supp(\mathbf{W}) \cap A \neq \emptyset \\ 
\supp(\mathbf{W}) \cap B \neq \emptyset}}
\Bigl|\phi_{\mathbf{W}} Z^X_\mathbf{W}\Bigr|
\Biggr\}
- 1 \right] \\[1em]
&= \left[
\prod_{X\in\{\emptyset,A,B,AB\}}
\exp\!\Biggl\{
\sum_{\substack{\text{connected}\, \mathbf{W} \leftarrow \LL_{AB} \\ 
\supp(\mathbf{W}) \cap A \neq \emptyset \\ 
\supp(\mathbf{W}) \cap B \neq \emptyset}}
\Bigl|\phi_{\mathbf{W}} Z^X_\mathbf{W}\Bigr|
\Biggr\}
- 1 \right]. \\[0.5em]
\end{align}
where we used $|e^x -1 | \leq e^{|x|} -1$, and the fact that $\|O_{A(B)}\| = 1$. We have that for all $X\in\{\emptyset, A,B,AB\}$ and for all excitations $l \in \LL_{AB}$, 
\begin{equation}
    |Z^X_l| \leq \order{ e^{-c|l|}}
\end{equation}
with $c>c_0$. Since all clusters with support on both $A$ and $B$ have weight at least $d(A,B)$, we get, 
\begin{align}
    \Bigl|\expval{O_AO_B}_c\Bigr| &\leq    \left[
\prod_{X\in \{\emptyset, A,B,AB\}}\exp\!\Biggl\{
r |AB| e^{-(c-c_0)(d(A,B)+1)}
\Biggr\}
- 1 \right] \\[1em]
&\leq   \left[
\exp\!\Biggl\{
4r|AB| e^{-(c-c_0)(d(A,B)+1)}
\Biggr\}
- 1 \right] \\[1em]
&\leq   4r |AB|e^{-(c-c_0)(d(A,B)+1)} \exp\!\bigl\{4r|AB| e^{-(c-c_0)(d(A,B)+1)}\bigr\}.
\end{align}

where $r=\order{1}$ is a constant and we used,
\begin{equation}
    \sum_{\substack{\text{conn} \, \mathbf{W} \leftarrow \LL_A \\ 
    \supp(\mathbf{W})\ni A \\ |\mathbf W| > m}} |\phi_{\mathbf{W}} Z^X_{\mathbf{W}}| \leq \order{|A|e^{-(c-c_0)(m+1)}}
\end{equation}
by \autoref{lem:cluster_tail} and the fact that $e^x - 1 \leq xe^x$.

\noindent
Since $e^{-(c-c_0)(d(A,B)+1)} \to 0$ exponentially as $d(A,B)\to\infty$, for all $d(A,B)$ large enough we have $4r e^{-(c-c_0)(d(A,B)+1)} \le 1$, and hence there is a $C = \order{|AB|}$ such that 
\[
\Bigl|\expval{O_AO_B}_c\Bigr| \leq    C e^{-(c-c_0) d(A,B)}.
\]
Thus, we have clustering of correlations,
\[
\Bigl|\expval{O_AO_B}_c\Bigr| \;\le\; \order{e^{-d(A,B)/\xi}}, 
\qquad \text{with } \xi \leq \order{\frac{1}{c-c_0}}.
\]

\end{proof}


\smsubsection{Derivative of free energy}

Here, we provide the additive version of cluster expansion for $p-$point connected correlation functions.

\begin{prop}\label{prop:correlatorexpansionformalderivative}
Let $\psi$ be a state represented as a PEPS on $G=(V,E)$. Given local disjoint regions $A_i \subset V$, and observables $O_{A_i}$ satisfying $\emph{supp}(O_{A_i}) \subseteq A_i$, we have the following cluster expansion for the $p-$point connected correlation functions,

\begin{equation}\label{eq:conn-corr-formal}
   \expval{O_{A_1}\dots O_{A_p}}_c = \sum_{\substack{\emph{conn.}\, \mathbf{W} \leftarrow \LL_{\mathbf{A}} \\
   \emph{supp}(\mathbf{W}) \cap A_i \neq \emptyset}} \phi_{\mathbf{W}} \partial_{\bm \lambda} Z^{\mathbf{A}}_{\mathbf{W},\bm{\lambda}}\Big|_{\bm{\lambda}=0}
\end{equation}
where we denote $\bm{\lambda} := (\lambda_1,\dots,\lambda_p)$, $\partial_{\bm{\lambda}} := \prod_i \partial_{\lambda_i}$, $\mathbf{A} = (A_1,\dots,A_p)$ and the cluster evaluation $Z^{\mathbf{A}}_{\mathbf{W},\bm{\lambda}}$ is done over the network $\ZZ_{\bm{\lambda}}^\mathbf{A} := \expect{\psi}{\exp(\sum_i O_{A_i}\lambda_i)}{\psi}$. The clusters $\mathbf{W}$ are constructed over $\LL_\mathbf{A}$, the set of connected sub-graphs required to be degree $\geq 2$ everywhere except in all regions $A_i$.

\end{prop}
\begin{proof}
    A standard result in statistical mechanics (see, e.g., Ref.~\cite{Friedli_Velenik_2017}) treats the free energy as a cumulant generating function, whose derivatives are the connected correlation functions:
    \begin{equation}
         \expval{O_{A_1}\dots O_{A_p}}_c = -\partial_{\bm{\lambda}} \FF_{\bm{\lambda}}^\mathbf{A}\Big|_{\bm{\lambda} = 0}
    \end{equation}
    where $\FF_{\bm{\lambda}}^\mathbf{A}:=-\log \expect{\psi}{\exp (\sum_i O_{A_i}\lambda_i)}{\psi}$. When $p=1$, this reduces to the expectation value $\expval{O_A}$, while for $p=2$, we recover $\expval{O_{A_1}O_{A_2}}_c = \expval{O_{A_1} O_{A_2}} - \expval{O_{A_1}} \expval{O_{A_2}}$.  Now, we apply the cluster expansion to $\FF_{\bm{\lambda}}^\mathbf{A}$. Owing to the $\bm{\lambda}$ derivatives, the only clusters with a non-zero contribution are those which intersect all of the regions $A_i$, leading to the condition $\supp(\mathbf W) \cap A_i \neq \emptyset $ for all $i\in[p]$.
\end{proof}

Following a similar approach to~\hyperref[sect:local-free]{Sec. S3.2}, we can explicitly evaluate~\autoref{eq:conn-corr-formal} for the most physically relevant case of $p=2$. A (tedious) calculation yields
\begin{align}
\expval{O_A O_B}_c &=\partial_{\lambda_a} \partial_{\lambda_b} \ZZ_{(\lambda_a,\lambda_b)}^{AB} \mid_{\lambda_a=\lambda_b=0} \, = \partial_{\lambda_a}\partial_{\lambda_b} \eval{\left(\sum_{\substack{\mathrm{conn.} \, \mathbf{W} \leftarrow \LL_{AB} \\ \mathrm{supp}(\mathbf{W}) \cap A \neq \emptyset \\ \mathrm{supp} (\mathbf{W}) \cap B \neq \emptyset}} \phi_{\mathbf{W}} \tloopZ{\mathbf{W}}\right)}_{\lambda_a=\lambda_b=0} \notag \\     &= \sum_{\substack{\mathrm{conn.} \, \mathbf{W} \leftarrow \LL_{AB} \\ \mathrm{supp}(\mathbf{W}) \cap A \neq \emptyset \\ \mathrm{supp} (\mathbf{W}) \cap B \neq \emptyset}} \phi_{\mathbf{W}} Z_{\mathbf{W}} \left[\sum_{l \in \mathbf{W}_{AB}} \eta_l \left(\frac{Z_l^{AB}}{Z_l} + (\eta_l-1) \frac{Z_l^A Z_l^B}{Z_l^2}\right) + \sum_{\substack{l \in \mathbf{W}_A \\ l' \in \mathbf{W}_B \\ l' \neq l}} \eta_l \eta_{l'} \frac{Z_l^A Z_{l'}^B}{Z_l Z_{l'}} - \sum_{\substack{l \in \mathbf{W}_A\\ l' \in \mathbf{W}_B}} \eta_l \eta_{l'} \expval{O_A}_{BP} \expval{O_B}_{BP}\right].
\end{align}
Here, $\mathbf{W}_A$ and $\mathbf{W}_B$ denote the subsets of loops in $\mathbf{W}$ that intersect $A$ and $B$, respectively, while $\mathbf{W}_{AB}$ denotes the subset of strings that intersect both $A$ and $B$. We emphasize that $\mathbf{W}_{AB}$ may be non-empty, as in the case of a cluster $\mathbf{W}=\{(l_A,1),(l_B,1)\}$ where $l_A(B)$ intersects $A(B)$ and the two strings intersect in the bulk. $Z_l^{AB}, Z_l^A,$ and $Z_l^{B}$ denote the normalized weights of loop $l$ on the tensor networks $\ZZ^{AB}, \ZZ^A, $ and $\ZZ^B$, respectively.

\smsection{Additional numerics}

Here we present additional numerical results on ground states and at finite temperature.

\smsubsection{Ground state iPEPS}
The iPEPS ground state is obtained through CTMRG-based variational optimization on a $C_4-$symmetric representation with bond dimension $D=4$ using \texttt{peps-torch}. The optimization is approximate, and the apparent critical point is obtained at $h_x=3.06$, vs. the high-precision estimate of $h_{x}^{c} = 3.044330(6)$ obtained via quantum Monte Carlo~\cite{Huang2020}. The message-passing fixed points are obtained through iterating the equation $\mu \to T\star \mu^{\otimes 3} / \| T\star \mu^{\otimes 3}\|_2$ until convergence, starting from the all-ones vector. To locate the unstable fixed point for $h_x \in (3.06, 3.2)$ [in the regime where the CTMRG magnetization vanishes but the BP magnetization does not], we employ a damped Newton search initialized using the fixed point at $h_x = 3.2$. The initialization consists of 500 random seeds drawn from a Gaussian cloud with standard deviation $0.5$ centered around this reference fixed point. We verify that the newly discovered (approximate) fixed point is unstable by perturbing it slightly and observing convergence towards the stable (symmetry-broken) fixed point under message-passing.

\smsubsection{Decay of loops}
In the main text, we claimed that odd-$|l|$ loops are strongly suppressed compared to even $|l|$ loops. This suppression comes from degree 3 vertices, which are necessarily present in odd-$|l|$ generalized loops. \autoref{fig:loops-supp}a-b shows the scaling of maximum-weight odd and even loops in the 3D and 2D TFIM, respectively, at fixed high temperature. For a given $|l|$, the loops of highest weight are those with the minimal number of degree-3 vertices, which is 0(1) for even(odd) loops. Thus, both parities scale approximately as $\tanh(\beta)^{|l|}$, with a constant prefactor that differs based on parity. As $h_x$ increases, the gap between odd and even loops narrows.

\begin{figure}[hbtp]
\includegraphics[width=\linewidth]{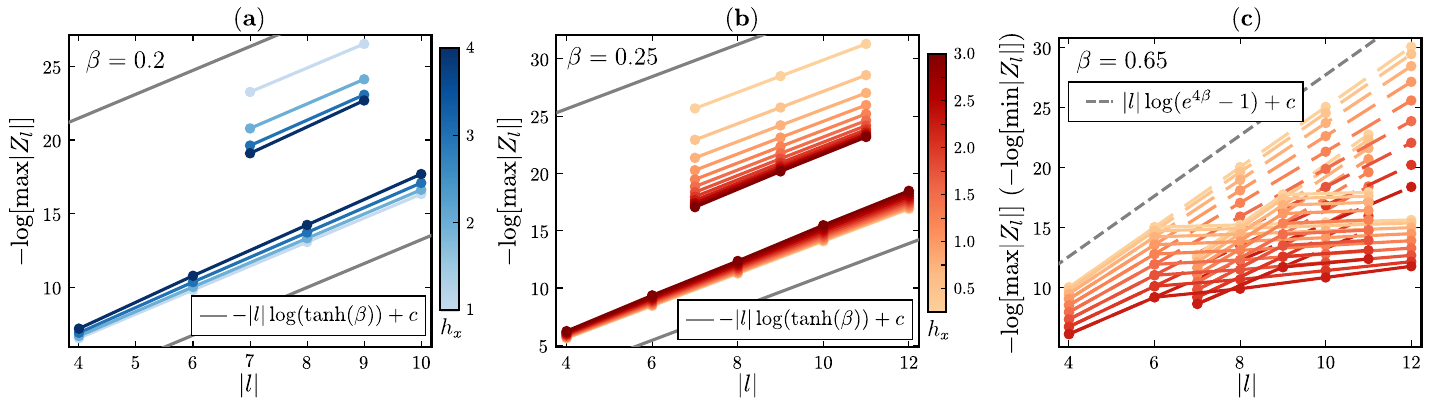}
\caption{\label{fig:loops-supp}Decay of loops around a BP fixed point in the transverse field Ising model. (a) 3D TFIM at inverse temperature $\beta=0.2$, transverse field $h_x=1,2,3,4$. (b) 2D TFIM at inverse temperature $\beta=0.25$, with $h_x$ ranging from 0.25 to 3 in increments of 0.25. In (a) and (b), the gray lines have slope $-\log(\tanh(\beta))$, which is the classical scaling of loops with respect to the paramagnetic fixed point. (c) 2D TFIM at $\beta=0.65$, which is in the ordered phase for the transverse fields shown (up to $h_x=2.25$). Solid and dashed lines indicate loops of maximal and minimal weight, respectively. Dashed gray line indicates the classical scaling with slope $\log(e^{4\beta}-1)$.}
\end{figure}

\autoref{fig:loops-supp}c shows the scaling of loops around the ferromagnetic fixed point, at low temperatures in the 2D TFIM. Here, two distinct scalings emerge. Maximal weight loops (solid lines) have the maximal number of degree-4 and degree-3 vertices, owing to the broken symmetry of the fixed point which favors the clustering of excitations. In the classical 2D Ising model, each degree-4 vertex carries a weight of $1-\frac{1}{2\sinh^2(2\beta)}$. Meanwhile, minimal-weight loops (dashed lines) have the maximal number of degree-2 vertices, and scale as $(e^{4\beta}-1)^{-|l|}$.

\smsubsection{Cluster-cumulant expansion}
\begin{figure}[hbtp]
\includegraphics[width=0.6\linewidth]{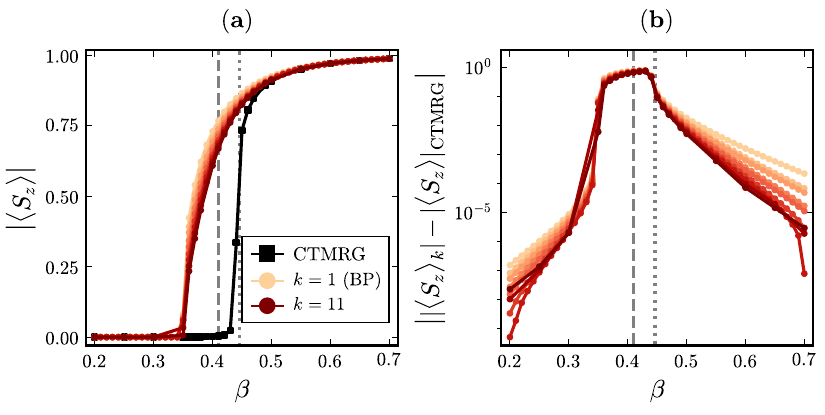}
\caption{Same as~\autoref{fig:finite-temp}b-c in the main text, but for $|\langle S_z\rangle|$ instead of $\expval{S_x}$. (a) shows $|\langle S_z\rangle|$ as a function of $\beta$ at $h_x=0.25$, evaluated using CTMRG (black squares), vs. increasing orders of the region-based cumulant cluster expansion (ranging from $k=1$ in light orange, to $k=11$ in dark red). (b) shows the additive error in the cumulant cluster expansion, using CTMRG as the ``ground truth.'' Dashed gray line marks the inverse temperature at which $\expval{S_x}_{\rm CTMRG}$ has a peak, while the dotted gray line is the inverse critical temperature $\beta_c \approx 0.446$. \label{fig:0.25-Z}}
\end{figure}
In the main text, we presented data on the region-based cluster cumulant expansion for the observable $S_x$ as a function of inverse temperature $\beta$. As a complement to~\autoref{fig:finite-temp}c-d, ~\autoref{fig:0.25-Z} shows the corresponding data for the observable $S_z$. All states were prepared with a small longitudinal field $h_z = 10^{-5}$, so that $\langle S_z\rangle$ is slightly positive in the paramagnetic phase. We chose this field to bias the ferromagnetic phase towards a definite sign in the magnetization, while keeping the true phase transition fairly sharp. At this value of the longitudinal field, both BP and CTMRG still sometimes converge toward a negative magnetization instead, and thus we plot $\left|\langle S_z\rangle\right|$. The error has a broad peak within the ``confusion regime'', reaching a maximum near the true critical temperature (dotteed line). We observe that the peak in $\expval{S_x}$ and its corresponding error occurs at a slightly higher temperature ($\beta = 0.41$, dashed line) than the actual phase transition (estimated as $\beta_c(0.25)\approx 0.446$ in Ref.~\cite{roughening}).

This general behavior -- peak in $\expval{S_x}$ coinciding with a peak in the $\left|\expval{S_x}_k - \expval{S_x}_{\mathrm{CTMRG}}\right|$ followed soon after by symmetry breaking in $\expval{S_z}$ coinciding with a peak in $\Big||\langle{S_z}\rangle_k| - |\langle{S_z}\rangle_\mathrm{CTMRG}|\Big|$ -- persists to larger transverse field. All four quantities are plotted for $g=1.5$ and $g=2.75$ in \autoref{fig:1.5} and~\autoref{fig:2.75}, respectively. 

\begin{figure}[hbtp]
\includegraphics[width=\linewidth]{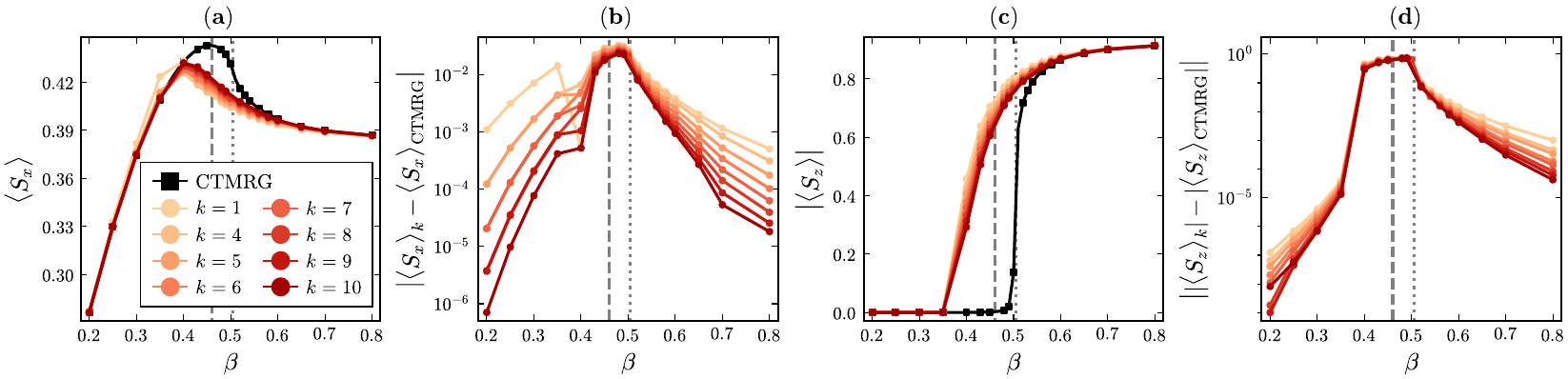}
\caption{Comparison of regions-based cluster cumulant expansion up to $k=10$ vs. CTMRG, on Gibbs states of the 2D TFIM at $h_x=1.5$. (a), (b): expectation values and error in $\expval{S_x}$. (c), (d): expectation values and error in $|\langle S_z\rangle|$. In all panels, the dashed line indicates the inverse temperature at which $\expval{S_x}_\mathrm{CTMRG}$ has a peak, while the dotted line is the inverse critical temperature from Ref.~\cite{roughening}.\label{fig:1.5}}
\end{figure}

\begin{figure}[hbtp]
\includegraphics[width=\linewidth]{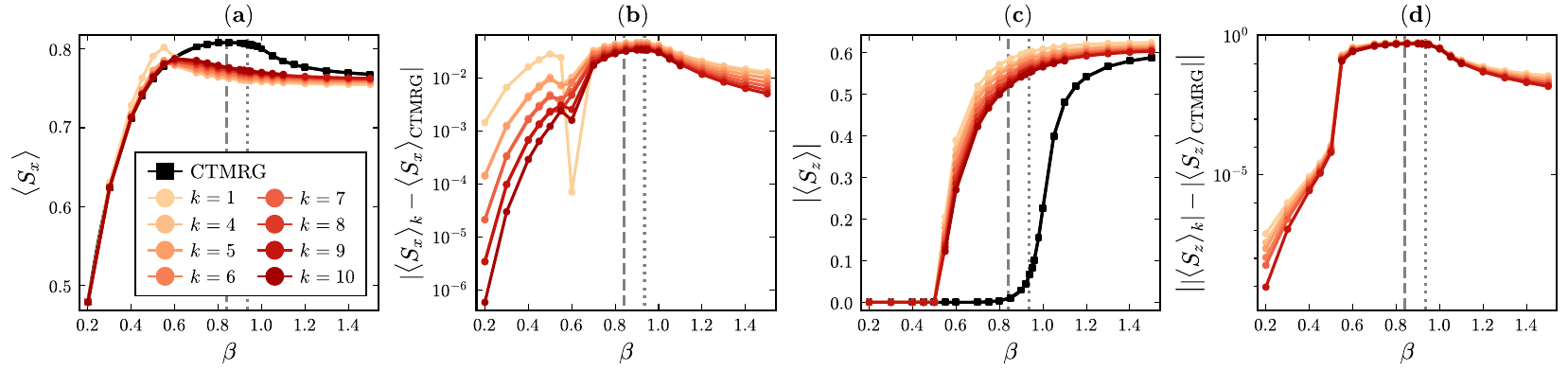}\caption{Same as~\autoref{fig:1.5}, but for $h_x=2.75$.\label{fig:2.75}}
\end{figure}
\end{document}